\documentclass[]{interact}

\usepackage{epstopdf}
\usepackage[caption=false]{subfig}

\usepackage[natbibapa,nodoi]{apacite}
\usepackage{hyperref}

\theoremstyle{plain}
\newtheorem{theorem}{Theorem}[section]

\newtheorem{assumption}{Assumption}

\theoremstyle{definition}

\theoremstyle{remark}

\usepackage{enumitem}
\newlist{Properties}{enumerate}{2}
\setlist[Properties]{label=Property \arabic*.,itemindent=*}

\begin{document}

\articletype{Preprint accepted to International Journal of Control}

\title{Modelling and Control of a Knuckle Boom Crane \footnote[1]{This research has been funded by The Brussels Institute for Research and Innovation (INNOVIRIS) of the Brussels Region through the Applied PHD grant: Brickiebots - Robotic Bricklayer: a multi-robot system for sand-lime blocks masonry (réf :  19-PHD-12)}}

\author{
\name{M. Ambrosino\textsuperscript{a}\thanks{Corresponding author: Michele.Ambrosino@ulb.ac.be} and E. Garone\textsuperscript{a}}
\affil{\textsuperscript{a}Service d’Automatique et d’Analyse des Systèmes, Université Libre de Bruxelles, Brussels, Belgium. }
}

\maketitle

\begin{abstract}
Cranes come in various sizes and designs to perform different tasks. Depending on their dynamic properties, they can be classified as gantry cranes and rotary cranes. In this paper we will focus on the so called 'knuckle boom' cranes which are among the most common types of rotary cranes. Compared with the other kinds of cranes (e.g. boom cranes, tower cranes, overhead cranes, etc), the study of knuckle cranes is still at an early stage and very few control strategies for this kind of crane have been proposed in the literature. Although fairly simple mechanically, from the control viewpoint the knuckle cranes present several challenges. A first result of this paper is to present for the first time a complete mathematical model for this kind of crane where it is possible to control the three rotations of the crane (known as luff, slew, and jib movement), and the cable length. The only simplifying assumption of the model is that the cable is considered rigid. On the basis of this model, we propose a nonlinear control law based on energy considerations which is able to perform position control of the crane while actively damping the oscillations of the load. The corresponding stability and convergence analysis is carefully proved using the \textit{LaSalle’s invariance principle}. The effectiveness of the proposed control approach has been tested in simulation with realistic physical parameters and in the presence of model mismatch. 
\end{abstract}


\section{Introduction}\label{sec:intro}

To handle heavy loads and materials, different types of cranes are widely used in different industrial fields. Although the tasks performed by the cranes are fairly simple (\textit{i.e.} moving or lifting heavy materials), it is quite challenging for a human operator to achieve accurate positioning and swing elimination simultaneously. Moreover, there are many practical problems associated with manual operation, such as low efficiency, poor safety, etc. Therefore, the control problem of cranes has been studied for a long time. In this paper we focused on the modelling and control of knuckle cranes, that are among the most common types of rotary crane. 

\medskip

Knuckle cranes are special kinds of boom cranes. Compared with other cranes, this type of crane has higher flexibility and lower energy consumption \citep{ref1} and for this reason it is probably the most common kind of crane in heavy industry. Knuckle cranes are special boom cranes that have an auxiliary jib connected to the boom to enhance the maneuverability and increase the workspace of the crane. In this paper we present for the first time a complete mathematical model for this kind of crane which takes into account not only the three main rotations (\textit{i.e.} luff, slew, and jib movement) but also the cable dynamic and the payload oscillations.

\medskip

As all cranes,  knuckle cranes are nonlinear systems with underactuated dynamics. The problem of controlling underactuated systems has been a topic of great interest in different industrial fields \citep{ref2,u1,u2,u3,u4,u5}. The condition of underactuation refers to a system 
having fewer actuators (input variables) than degrees of freedom (number of independent variables that define the system configuration). This implies that some of the states of the system cannot be directly commanded, which highly complicates the design of control algorithms. In particular for a knuckle crane (as for any crane) the non-actuated variable are the swing angles of the payload, whereas the four actuated variables are the three main rotation (\textit{i.e.} luff, slew, and jib movements) and the length of the cable. While in the past few years several solutions have been proposed for the control of boom cranes control, the control problem of a knuckle crane is still an opening and challenging problem 

\medskip


Generally speaking, control schemes for boom cranes present in literature can be roughly categorised into open loop and closed loop techniques \citep{ref5}. 

\medskip

Input shaping is one of the most used open loop techniques, that can be applied in real time, mainly for the control of the oscillations of the payload. \citep{in1} proposes a input-shaping control for a in-scale boom crane to reduce the residual oscillations. \citep{in2} discusses an input-shaping strategy that minimizes the oscillations of the payload caused by an external disturbance during the luff command. Open-loop trajectory planning methods such as S-curve trajectories approaches are proposed in \citep{46,Uchiyama} to achieve anti-sway control for the payload. Open loop control schemes have been widely used because they are easy to implement. However, open loop control schemes are sensitive to external disturbances and to model mismatch. Therefore, several closed loop control methods have been proposed to increase robustness and achieve better performance in presence of perturbations. As closed loop technique, the Linear Quadratic Regulator (LQR) is one the most common techniques applied to cranes \citep{ref5}. In \citep{sun2017} the authors provided a LQR control approach for boom crane considering a fixed cable length. Another closed-loop control scheme widely used for boom crane is the Proportional Integral Derivative (PID) control. In \citep{yang2017} the authors designed a Proportional-Derivative (PD) controller with gravity compensation based on the nonlinear model of the boom crane. In \citep{PFLboom} is shown a \textit{partial feedback linearization} (PFL) based on detailed mathematical model of boom cranes. In \citep{pole} is shown a pole placement approach for a linearized model of boom crane. In \citep{ref8} a model predictive control (MPC) approach is used for a boom crane in order to reduce the swing angles as much as possible. A constraint control based on the theory of the \textit{Explicit Reference Governor} (ERG), is discussed in \citep{ERGboom}, where the authors focused on controlling a simplified model of boom crane in the presence of constraints. A second-order sliding mode control law is proposed in \citep{slid} for trajectory tracking and anti-sway control.  Control schemes based on a combination of open and closed loop techniques have also been proposed. In \citep{51} a combination of input shaping and feedback control is proposed to reduce the effect of gusts of wind. 

\medskip

Compared with boom cranes, the state of the art of knuckle cranes control is much less developed. In \citep{ref12} the authors focus on controlling mobile electro-hydraulic proportional valves to move the crane to a desired position. In \citep{ref13} the authors solve the problem of controlling knuckle crane through the inverse kinematics without take into account the dynamic of the cable and the payload. In \citep{ref14} an anti-sway control is shown which is performed by simplifying the dynamic model and assuming that the tip of the crane can be controlled directly. To the best of our knowledge, no research has been carried out to derive a detailed mathematical model of a knuckle cranes and to develop a control strategy taking into account the nonlinear nature of this type of crane.

\medskip

The main contributions of this paper are:

\begin{enumerate}
    \item The development of a complete mathematical model for a knuckle crane which takes into account all of the degrees of freedom (DoFs) of this type of cranes (e.g. the three rotations, the length of the rope and the payload swing angles). The only simplifying assumptions of the proposed model is that the cable is considered rigid. 
    \item The design of a novel control strategy designed directly on the nonlinear model and for which a detailed proof of asymptotic stability is provided making use of the \textit{LaSalle’s invariance principle}.
    \item Different simulation scenarios are presented to demonstrate that the proposed control method is able to perform position control of the crane while actively damping the oscillations of the load, even in presence of model mismatch.  
\end{enumerate}

\medskip

The rest of this paper is organized as follows. In Section 2,
the dynamic model of a knuckle crane and the control objectives are provided. In Section 3, the proposed
controller is designed, and in Section 4 the corresponding stability analysis is provided in detail. Section 5 shows the results of the simulations for the proposed control strategy.

\section{Dynamic Model}\label{sec:prob}

A schematic view of a knuckle crane is depiced in Fig.\ref{fig:crane}. The knuckle crane consists of a first boom of length $l_b$ and mass $m_b$ connected to the tower with one revolute joint. The auxiliary jib of length $l_j$ and mass $m_j$ is linked to the first boom by a revolute joint. For the sake of simplicity, in this paper all the links and joints are considered to be rigid. The cable is supposed to be massless and rigid, thus the lifting mechanism can be described as a prismatic joint. The payload of mass $m$ is described as a lumped mass. The two swing angles of the payload are represented in the Fig.\ref{fig:pay}.

\begin{figure}[ht!]
\centering

\subfloat[Model of a knuckle crane]{%
\resizebox*{5cm}{!}{\includegraphics{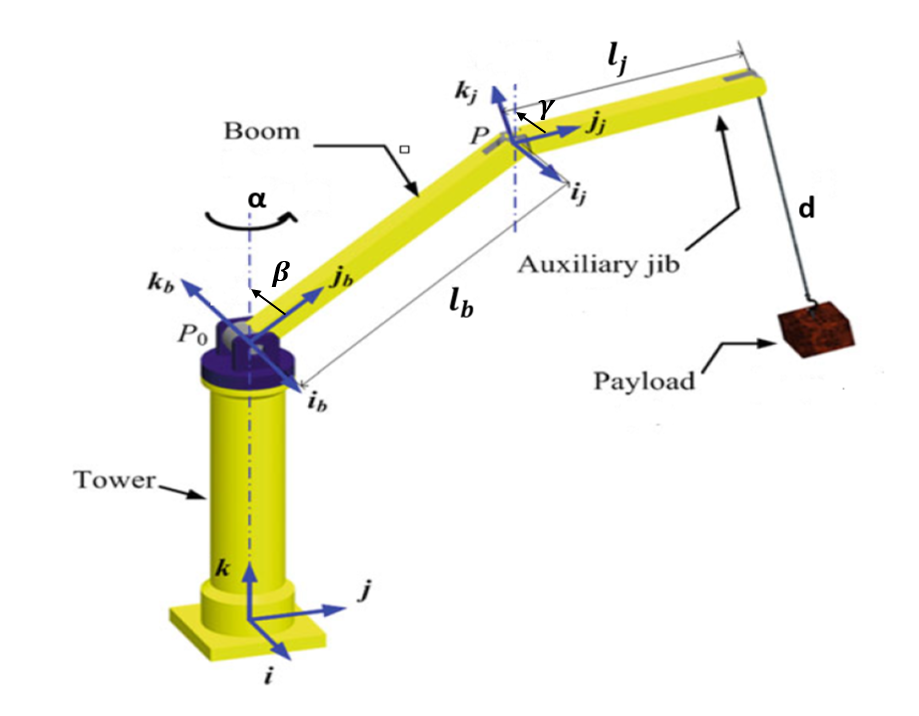}}\label{fig:crane}}\hspace{5pt}
\subfloat[Payload swing angles.]{%
\resizebox*{5cm}{!}{\includegraphics{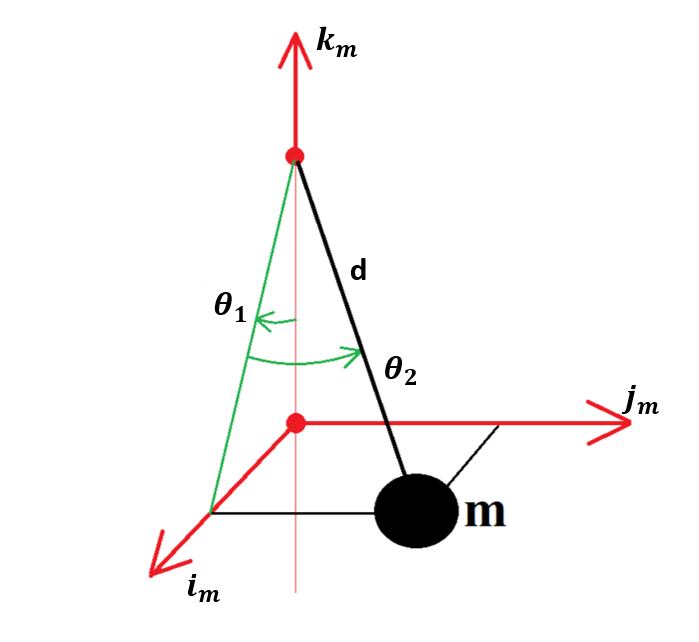}} \label{fig:pay}}
\caption{Schematic view of a knuckle crane}

\end{figure}

\medskip

The configuration of the crane is conveniently described by six generalized coordinates. In particular: $\alpha$ is the slew angle of
the tower, $\beta$ is the luff angle of the boom, $\gamma$ is the luff angle of the jib, \textit{d} is the length of
the rope, $\theta_1$ is the tangential pendulation mainly due to the motion of the tower, and $\theta_2$ is the radial sway mainly due to the motion of the boom. 

\medskip

The dynamic model of the knuckle crane can be obtained using the \textit{Euler-Lagrange method}. Firstly, we need to express the system kinematic energy $T(t)$, which consists of three parts: the boom kinematic energy $T_b(t)$, the jib kinematic energy $T_j(t)$, and the payload kinematic energy $T_p(t)$. Then, we define the system potential energy $U(t)$ consisting of the boom potential energy $U_b(t)$, jib potential energy $U_j(t)$, and the payload potential energy $U_p(t)$. The \textit{Lagrangian} of the knuckle crane is
\begin{equation}\label{eq:lag}
\begin{split}
 \mathcal{L}(t) = T(t) - U(t)  = T_b(t) + T_j(t) + T_p(t) - U_b(t) - U_j(t) - U_p(t),
\end{split}
\end{equation}

where:
\begin{equation}\label{eq:T}
\begin{split}
T(t) = {1\over{8}}(m_B((l_BC_{\beta}S_{\alpha}\dot{\alpha} + l_BC_{\alpha}S_{\beta}\dot{\beta})^2
+ (l_BC_{\alpha}C_{\beta}\dot{\alpha} - l_BS_{\alpha}S_{\beta}\dot{\beta})^2 
+ l_B^2C_{\beta}^2\dot{\beta}^2)) \\
+ {1\over{2}}(m_J((l_BC_{\beta}S_{\alpha}\dot{\alpha}
+ l_BC_{\alpha}S_{\beta}\dot{\beta}+  {1\over{2}}(l_JC_{\gamma}S_{\alpha}\dot{\alpha}) 
+ {1\over{2}}(l_JC_{\alpha}S_{\gamma}\dot{\gamma}))^2 + (l_BC_{\alpha}C_{\beta}\dot{\alpha} \\
+ {1\over{2}}(l_JC_{\alpha}C_{\gamma}\dot{\alpha}) - l_BS_{\alpha}S_{\beta}\dot{\beta} 
- {1\over{2}}(l_JS_{\alpha}S_{\gamma}\dot{\gamma}))^2 + (l_BC_{\beta}\dot{\beta} 
+ {1\over{2}}(l_JC_{\gamma}\dot{\gamma}))^2)) \\ + {1\over{2}}(m((C_{\theta_2}S_{\alpha}S_{\theta_1}\dot{d} 
- C_{\alpha}S_{\theta_2}\dot{d} + l_BC_{\beta}S_{\alpha}\dot{\alpha} + l_BC_{\alpha}S_{\beta}\dot{\beta} +  l_JC_{\gamma}S_{\alpha}\dot{\alpha} + l_JC_{\alpha}S_{\gamma}\dot{\gamma} \\ - C_{\alpha}C_{\theta_2}\dot{\theta}_2d + S_{\alpha}S_{\theta_2}\dot{\alpha}d + C_{\alpha}C_{\theta_2}S_{\theta_1}\dot{\alpha}d + C_{\theta_1}C_{\theta_2}S_{\alpha}\dot{\theta}_1d - S_{\alpha}S_{\theta_1}S_{\theta_2}\dot{\theta}_2d)^2 + (l_BC_{\beta}\dot{\beta}\\ - C_{\theta_1}C_{\theta_2}\dot{d} + l_JC_{\gamma}\dot{\gamma} + C_{\theta_2}S_{\theta_1}\dot{\theta}_1d + C_{\theta_1}S_{\theta_2}\dot{\theta}_2d)^2 + (S_{\alpha}S_{\theta_2}\dot{d}  + l_BC_{\alpha}C_{\beta}\dot{\alpha} + l_JC_{\alpha}C_{\gamma}\dot{\alpha} \\ - l_BS_{\alpha}S_{\beta}\dot{\beta} - l_JS_{\alpha}S_{\gamma}\dot{\gamma} + C_{\alpha}S_{\theta_2}\dot{\alpha}d + C_{\theta_2}S_{\alpha}\dot{\theta}_2d + C_{\alpha}C_{\theta_2}S_{\theta_1}\dot{d} + C_{\alpha}C_{\theta_1}C_{\theta_2}\dot{\theta}_1d \\ - C_{\theta_2}S_{\alpha}S_{\theta_1}\dot{\alpha}d - C_{\alpha}S_{\theta_1}S_{\theta_2}\dot{\theta}_2d)^2)) + {1\over{2}}I_{tot}\dot{\alpha}^2 + {1\over{2}}I_B\dot{\beta}^2 + {1\over{2}}I_J\dot{\gamma}^2,
\end{split}
\end{equation}
\begin{equation}\label{eq:U}
\begin{split}
U(t) = gm(l_BS_{\beta} + l_JS_{\gamma} - C_{\theta_1}C_{\theta_2}d) 
+ gm_J(l_BS_{\beta} + {1\over{2}}l_JS_{\gamma}) + {1\over{2}}gl_Bm_BS_{\beta},
\end{split}
\end{equation}

and where 

\begin{align*}
S_{\alpha} \triangleq \sin(\alpha),\quad S_{\beta} \triangleq sin(\beta),\quad S_{\gamma} \triangleq \sin(\gamma),\quad S_{\theta_1} \triangleq \sin(\theta_1),\quad S_{\theta_2} \triangleq \sin(\theta_2),\\ C_{\alpha} \triangleq \cos(\alpha),\quad C_{\beta} \triangleq \cos(\beta),\quad C_{\gamma} \triangleq \cos(\gamma),\quad C_{\theta_1} \triangleq \cos(\theta_1),\quad C_{\theta_2} \triangleq \cos(\theta_2).
\end{align*}

The meaning of the system physical parameters in (\ref{eq:T})-(\ref{eq:U}) are reported in
Tab.~\ref{tb:par}. 

\begin{table}[hb]
\begin{center}
\caption{Parameters of the knuckle crane system}\label{tb:par}
\begin{tabular}{cccc}
Parameters & Physical & Units \\\hline
$m_b$ & Boom mass & $\mathrm{kg}$ \\ 
$m_j$ & Jib mass & $\mathrm{kg}$ \\
m & Payload mass & $\mathrm{kg}$ \\ 
$l_b$ & Boom length & $\mathrm{m}$ \\ 
$l_j$ & jib length & $\mathrm{m}$ \\  
$I_{tot}$ & Tower inertia moment & $\mathrm{kg\cdot m^2}$ \\
$I_b$ & Boom inertia moment & $\mathrm{kg\cdot m^2}$ \\
$I_j$ & Jib inertia moment & $\mathrm{kg\cdot m^2}$ \\\hline
\end{tabular}
\end{center}
\end{table}

Moreover, to simplify the next analysis we introduce the auxiliary variables
\begin{itemize}
    \item $A_1 = l_B^2m + (l_B^2m_B)/4+l_B^2 m_J$;
    \item $A_2 = l_J^2m + (l_J^2 m_J)/4$;
    \item $A_3 = 2l_Bl_J m+ l_Bl_J m_J$;
    \item $A_4 = 2l_B m$;
    \item $A_5 = 2l_J m$.
\end{itemize}

The equations of the motion of the crane are derived using the \textit{Euler-Lagrange equation}
\begin{equation}
\frac{d}{dt}\left(\frac{\partial  \mathcal{L}(q, \dot{q})}{\partial{\dot{q}}}\right) -  \frac{\partial  \mathcal{L}(q, \dot{q})}{\partial{q}} = \zeta, 
\end{equation}

where q = $[\alpha, \beta, \gamma, d, \theta_1, \theta_2]^T \in{\mathbb{R}^6}$ is the system state vector, and $ \zeta  = [u_1, u_2, u_3, u_4, 0, 0]^T \in{\mathbb{R}^6}$ is the control input vector. The dynamic model of the knuckle crane can be described by the following equations

\begin{equation}\label{eq:alp}
    \begin{split}
    I_{tot}\ddot{\alpha} + A_1\ddot{\alpha}C_{\beta}^2 + A_2\ddot{\alpha}C_{\gamma}^2 + d^2\ddot{\alpha}m  + 2d^2\dot{\alpha}\dot{\theta}_1mC_{\theta_1}C_{\theta_2}^2S_{\theta_1} +  2d^2\dot{\alpha}\dot{\theta}_2mC_{\theta_1}^2C_{\theta_2}S_{\theta_2} - \\ -  A_1\dot{\alpha}\dot{\beta}S_{2\beta} - A_2\dot{\alpha}\dot{\gamma}S_{2\gamma} + A_3\ddot{\alpha}C_{\beta}C_{\gamma} - d^2\ddot{\theta_2}mS_{\theta_1} + \\ + 2\dot{d}d\dot{\alpha}m  +  2A_4d\ddot{\alpha}C_{\beta}S_{\theta_2} + 2A_5d\ddot{\alpha}C_{\gamma}S_{\theta_2} + 2A_4\dot{d}\dot{\alpha}C_{\beta}S_{\theta_2} + \\ +  2A_5\dot{d}\dot{\alpha}C_{\gamma}S_{\theta_2} -  A_3\dot{\alpha}\dot{\beta}C_{\gamma}S_{\beta}  - A_3\dot{\alpha}\dot{\gamma}C_{\beta}S_{\gamma} -  2d^2\dot{\theta}_1\dot{\theta}_2mC_{\theta_1} - \\ - d^2\ddot{\alpha}mC_{\theta_1}^2C_{\theta_2}^2 + A_4\ddot{d}C_{\beta}C_{\theta_2}S_{\theta_1} +  A_5\ddot{d}C_{\gamma}C_{\theta_2}S_{\theta_1} - 2\dot{d}d\dot{\theta}_2mS_{\theta_1} + \\ + A_4d\ddot{\beta}C_{\theta_2}S_{\beta}S_{\theta_1} - A_4d\ddot{\theta_2}C_{\beta}S_{\theta_1}S_{\theta_2} + A_5d\ddot{\gamma}C_{\theta_2}S_{\gamma}S_{\theta_1} - A_5d\ddot{\theta_2}C_{\gamma}S_{\theta_1}S_{\theta_2} - \\ -  2A_4\dot{d}\dot{\theta}_2C_{\beta}S_{\theta_1}S_{\theta_2} -  2A_5\dot{d}\dot{\theta}_2C_{\gamma}S_{\theta_1}S_{\theta_2} + A_4d\dot{\beta}^2C_{\beta}C_{\theta_2}S_{\theta_1}  - A_4d\dot{\theta}_1^2C_{\beta}C_{\theta_2}S_{\theta_1} - \\ - A_4d\dot{\theta}_2^2C_{\beta}C_{\theta_2}S_{\theta_1} +  A_5d\dot{\gamma}^2C_{\gamma}C_{\theta_2}S_{\theta_1} - A_5d\dot{\theta}_1^2C_{\gamma}C_{\theta_2}S_{\theta_1} - A_5d\dot{\theta}_2^2C_{\gamma}C_{\theta_2}S_{\theta_1} - \\ - 2\dot{d}d\dot{\alpha}mC_{\theta_1}^2C_{\theta_2}^2 + 2d^2\dot{\theta}_1\dot{\theta}_2mC_{\theta_1}C_{\theta_2}^2 + 2A_4d\dot{\alpha}\dot{\theta}_2C_{\beta}C_{\theta_2} + 2A_5d\dot{\alpha}\dot{\theta}_2C_{\gamma}C_{\theta_2} + \\ + d^2\ddot{\theta_1}mC_{\theta_1}C_{\theta_2}S_{\theta_2} -  2A_4d\dot{\alpha}\dot{\beta}S_{\beta}S_{\theta_2} - 2A_5d\dot{\alpha}\dot{\gamma}S_{\gamma}S_{\theta_2} + A_4d\ddot{\theta_1}C_{\beta}C_{\theta_1}C_{\theta_2} + \\ + A_5d\ddot{\theta_1}C_{\gamma}C_{\theta_1}C_{\theta_2} + 2A_4\dot{d}\dot{\theta}_1C_{\beta}C_{\theta_1}C_{\theta_2}  + 2A_5\dot{d}\dot{\theta}_1C_{\gamma}C_{\theta_1}C_{\theta_2} - d^2\dot{\theta}_1^2mC_{\theta_2}S_{\theta_1}S_{\theta_2} - \\ - 2A_4d\dot{\theta}_1\dot{\theta}_2C_{\beta}C_{\theta_1}S_{\theta_2} - 2A_5d\dot{\theta}_1\dot{\theta}_2C_{\gamma}C_{\theta_1}S_{\theta_2} + 2\dot{d}d\dot{\theta}_1mC_{\theta_1}C_{\theta_2}S_{\theta_2} = u_1,
    \end{split}
\end{equation}

\begin{equation}\label{eq:bet}
    \begin{split}
        A_1\ddot{\beta} + I_B\ddot{\beta} + (A_1\dot{\alpha}^2S_{2\beta})/2 + (A_3\dot{\alpha}^2C_{\gamma}S_{\beta})/2 - (A_3\dot{\gamma}^2C_{\beta}S_{\gamma})/2 +  (A_3\dot{\gamma}^2C_{\gamma}S_{\beta})/2 \\ + gl_BmC_{\beta} + (gl_Bm_BC_{\beta})/2 + gl_Bm_JC_{\beta} + \\ (A_3\ddot{\gamma}C_{\beta}C_{\gamma})/2 - A_4\ddot{d}S_{\beta}S_{\theta_2} + (A_3\ddot{\gamma}S_{\beta}S_{\gamma})/2 - A_4d\ddot{\theta_2}C_{\theta_2}S_{\beta} - \\ 2A_4\dot{d}\dot{\theta}_2C_{\theta_2}S_{\beta} -  A_4\ddot{d}C_{\beta}C_{\theta_1}C_{\theta_2} + A_4d\dot{\alpha}^2S_{\beta}S_{\theta_2} + A_4d\dot{\theta}_2^2S_{\beta}S_{\theta_2} \\+ A_4d\ddot{\theta_1}C_{\beta}C_{\theta_2}S_{\theta_1} + A_4d\ddot{\theta_2}C_{\beta}C_{\theta_1}S_{\theta_2} + \\ 2A_4\dot{d}\dot{\theta}_1C_{\beta}C_{\theta_2}S_{\theta_1} + 2A_4\dot{d}\dot{\theta}_2C_{\beta}C_{\theta_1}S_{\theta_2} + A_4d\ddot{\alpha}C_{\theta_2}S_{\beta}S_{\theta_1} \\+ 2A_4\dot{d}\dot{\alpha}C_{\theta_2}S_{\beta}S_{\theta_1} + A_4d\dot{\theta}_1^2C_{\beta}C_{\theta_1}C_{\theta_2} + A_4d\dot{\theta}_2^2C_{\beta}C_{\theta_1}C_{\theta_2} + \\ 2A_4d\dot{\alpha}\dot{\theta}_1C_{\theta_1}C_{\theta_2}S_{\beta} - 2A_4d\dot{\theta}_1\dot{\theta}_2C_{\beta}S_{\theta_1}S_{\theta_2} - 2A_4d\dot{\alpha}\dot{\theta}_2S_{\beta}S_{\theta_1}S_{\theta_2} = u_2,
    \end{split}
\end{equation}

\begin{equation}\label{eq:gam}
    \begin{split}
        A_2\ddot{\gamma}  + I_J\ddot{\gamma} + (A_2\dot{\alpha}^2S_{2\gamma})/2 + (A_3\dot{\alpha}^2C_{\beta}S_{\gamma})/2 + (A_3\dot{\beta}^2C_{\beta}S_{\gamma})/2 -\\- (A_3\dot{\beta}^2C_{\gamma}S_{\beta})/2 + gl_JmC_{\gamma} + (gl_Jm_JC_{\gamma})/2 +  (A_3\ddot{\beta}C_{\beta}C_{\gamma})/2 -\\- A_5\ddot{d}S_{\gamma}S_{\theta_2} + (A_3\ddot{\beta}S_{\beta}S_{\gamma})/2 - A_5d\ddot{\theta_2}C_{\theta_2}S_{\gamma} - 2A_5\dot{d}\dot{\theta}_2C_{\theta_2}S_{\gamma} - A_5\ddot{d}C_{\gamma}C_{\theta_1}C_{\theta_2} + \\ + A_5d\dot{\alpha}^2S_{\gamma}S_{\theta_2} +  A_5d\dot{\theta}_2^2S_{\gamma}S_{\theta_2} + A_5d\ddot{\theta_1}C_{\gamma}C_{\theta_2}S_{\theta_1} + A_5d\ddot{\theta_2}C_{\gamma}C_{\theta_1}S_{\theta_2} + \\ + 2A_5\dot{d}\dot{\theta}_1C_{\gamma}C_{\theta_2}S_{\theta_1} +  2A_5\dot{d}\dot{\theta}_2C_{\gamma}C_{\theta_1}S_{\theta_2} + A_5d\ddot{\alpha}C_{\theta_2}S_{\gamma}S_{\theta_1} + 2A_5\dot{d}\dot{\alpha}C_{\theta_2}S_{\gamma}S_{\theta_1} + \\ + A_5d\dot{\theta}_1^2C_{\gamma}C_{\theta_1}C_{\theta_2} + A_5d\dot{\theta}_2^2C_{\gamma}C_{\theta_1}C_{\theta_2} + 2A_5d\dot{\alpha}\dot{\theta}_1C_{\theta_1}C_{\theta_2}S_{\gamma} - \\ - 2A_5d\dot{\theta}_1\dot{\theta}_2C_{\gamma}S_{\theta_1}S_{\theta_2} - 2A_5d\dot{\alpha}\dot{\theta}_2S_{\gamma}S_{\theta_1}S_{\theta_2} = u_3,
    \end{split}
\end{equation}

\begin{equation}\label{eq:len}
    \begin{split}
        \ddot{d}m - d\dot{\alpha}^2m - d\dot{\theta}_2^2m - A_4\dot{\alpha}^2C_{\beta}S_{\theta_2} - A_4\dot{\beta}^2C_{\beta}S_{\theta_2} - A_5\dot{\alpha}^2C_{\gamma}S_{\theta_2} -
        \\ - A_5\dot{\gamma}^2C_{\gamma}S_{\theta_2} - d\dot{\theta}_1^2mC_{\theta_2}^2 - A_4\ddot{\beta}S_{\beta}S_{\theta_2} - A_5\ddot{\gamma}S_{\gamma}S_{\theta_2} - \\ - gmC_{\theta_1}C_{\theta_2} + d\dot{\alpha}^2mC_{\theta_1}^2C_{\theta_2}^2  - A_4\ddot{\beta}C_{\beta}C_{\theta_1}C_{\theta_2} - A_5\ddot{\gamma}C_{\gamma}C_{\theta_1}C_{\theta_2} + \\ + A_4\ddot{\alpha}C_{\beta}C_{\theta_2}S_{\theta_1} + A_5\ddot{\alpha}C_{\gamma}C_{\theta_2}S_{\theta_1} + 2d\dot{\alpha}\dot{\theta}_2mS_{\theta_1} + A_4\dot{\beta}^2C_{\theta_1}C_{\theta_2}S_{\beta} + \\ + A_5\dot{\gamma}^2C_{\theta_1}C_{\theta_2}S_{\gamma} - 2A_4\dot{\alpha}\dot{\beta}C_{\theta_2}S_{\beta}S_{\theta_1} - 2A_5\dot{\alpha}\dot{\gamma}C_{\theta_2}S_{\gamma}S_{\theta_1} - 2d\dot{\alpha}\dot{\theta}_1mC_{\theta_1}C_{\theta_2}S_{\theta_2} = u_4,
    \end{split}
\end{equation}

\begin{equation}\label{eq:th1}
\begin{split}
    dC_{\theta_2}(gmS_{\theta_1} - A_4\dot{\beta}^2S_{\beta}S_{\theta_1} - A_5\dot{\gamma}^2S_{\gamma}S_{\theta_1} + d\ddot{\theta_1}mC_{\theta_2} + 2\dot{d}\dot{\theta}_1mC_{\theta_2}  + \\ + A_4\ddot{\alpha}C_{\beta}C_{\theta_1} + A_5\ddot{\alpha}C_{\gamma}C_{\theta_1} + A_4\ddot{\beta}C_{\beta}S_{\theta_1} + A_5\ddot{\gamma}C_{\gamma}S_{\theta_1} - 2A_4\dot{\alpha}\dot{\beta}C_{\theta_1}S_{\beta} - \\ - 2A_5\dot{\alpha}\dot{\gamma}C_{\theta_1}S_{\gamma} + d\ddot{\alpha}mC_{\theta_1}S_{\theta_2} + 2\dot{d}\dot{\alpha}mC_{\theta_1}S_{\theta_2} - 2d\dot{\theta}_1\dot{\theta}_2mS_{\theta_2}  -\\- d\dot{\alpha}^2mC_{\theta_1}C_{\theta_2}S_{\theta_1} + 2d\dot{\alpha}\dot{\theta}_2mC_{\theta_1}C_{\theta_2}) = 0,
\end{split}
\end{equation}

\begin{equation}\label{eq:th2}
    \begin{split}
        -d(A_4\dot{\alpha}^2C_{\beta}C_{\theta_2} - 2\dot{d}\dot{\theta}_2m - d\ddot{\theta_2}m + A_4\dot{\beta}^2C_{\beta}C_{\theta_2} + A_5\dot{\alpha}^2C_{\gamma}C_{\theta_2} + \\+ A_5\dot{\gamma}^2C_{\gamma}C_{\theta_2} - (d\dot{\theta}_1^2mS_{2\theta_2})/2 + d\ddot{\alpha}mS_{\theta_1} + 2\dot{d}\dot{\alpha}mS_{\theta_1} + A_4\ddot{\beta}C_{\theta_2}S_{\beta} +\\+ A_5\ddot{\gamma}C_{\theta_2}S_{\gamma} - gmC_{\theta_1}S_{\theta_2} + A_4\dot{\beta}^2C_{\theta_1}S_{\beta}S_{\theta_2} + A_5\dot{\gamma}^2C_{\theta_1}S_{\gamma}S_{\theta_2} - \\ - A_4\ddot{\beta}C_{\beta}C_{\theta_1}S_{\theta_2} - A_5\ddot{\gamma}C_{\gamma}C_{\theta_1}S_{\theta_2} + A_4\ddot{\alpha}C_{\beta}S_{\theta_1}S_{\theta_2} + A_5\ddot{\alpha}C_{\gamma}S_{\theta_1}S_{\theta_2} + \\ + 2d\dot{\alpha}\dot{\theta}_1mC_{\theta_1}C_{\theta_2}^2 - 2A_4\dot{\alpha}\dot{\beta}S_{\beta}S_{\theta_1}S_{\theta_2} - 2A_5\dot{\alpha}\dot{\gamma}S_{\gamma}S_{\theta_1}S_{\theta_2} + d\dot{\alpha}^2mC_{\theta_1}^2C_{\theta_2}S_{\theta_2}) = 0.
    \end{split}
\end{equation}

Finally, the system (\ref{eq:alp})-(\ref{eq:th2}) can be compactly rewritten as

\begin{equation}\label{eq:modelmatrix}
{{M(q)\ddot{q} + C(q,\dot{q})\dot{q} + g(q)} = {\begin{bmatrix}
I_{4\times4} \\ 0_{2\times2} \end{bmatrix}}u. }
\end{equation}

The matrices $M(q) \in{\mathbb{R}^{6\times6}}$,$ C(q,\dot{q})\in{\mathbb{R}^{6\times6}}$, and $g(q) \in{\mathbb{R}^{6}}$ represent the inertia, centripetal-Coriolis, and gravity term, respectively. $I_{4 \times 4}$ and $0_{2 \times 2}$ are the Identity matrix and the Null matrix, respectively.

The system matrices (see Appendix A for the detailed description) are defined as

\begin{equation}
M(q) = \begin{bmatrix}
m_{11}&m_{12}&m_{13}&m_{14}&m_{15}&m_{16}\\
m_{21}&m_{22}&m_{23}&m_{24}&m_{25}&m_{26}\\
m_{31}&m_{32}&m_{33}&m_{34}&m_{35}&m_{36}\\
m_{41}&m_{42}&m_{43}&m_{44}&0&0\\
m_{51}&m_{52}&m_{53}&0&m_{55}&0\\
m_{61}&m_{62}&m_{63}&0&0&m_{66}\\
\end{bmatrix}, 
\end{equation}

\begin{equation}
C(q,\dot q) = \begin{bmatrix}
c_{11}&c_{12}&c_{13}&c_{14}&c_{15}&c_{16}\\
c_{21}&c_{22}&c_{23}&c_{24}&c_{25}&c_{26}\\
c_{31}&c_{32}&c_{33}&c_{34}&c_{35}&c_{36}\\
c_{41}&c_{42}&c_{43}&0&c_{45}&c_{46}\\
c_{51}&c_{52}&c_{53}&c_{54}&c_{55}&c_{56}\\
c_{61}&c_{62}&c_{63}&c_{64}&c_{65}&c_{66}\\
\end{bmatrix}, 
\end{equation}

\begin{equation}\label{eq:g}
{
g(q) = \begin{bmatrix} 0,g_2,g_3,g_4,g_5,g_6
\end{bmatrix}}^T. 
\end{equation}

Although the equation of motion (\ref{eq:modelmatrix}) is quite complicated, as all mechanical systems, it has several fundamental properties that can be exploited to facilitate the design of the control law. The two main properties that will be exploited in the next sections are:
\begin{Properties}\label{Pr}
  \item The matrix ${1\over{2}}\dot{M}(q) - C(q,\dot{q})$ is skew symmetric which means that:
  \begin{equation*}
      \eta^T\left[{1\over{2}}\dot{M}(q) - C(q,\dot{q})\right]\eta = 0, \quad \eta \in \mathbb{R}^{6} 
  \end{equation*}
  \item  The gravity vector (\ref{eq:g}) can be  obtained as the gradient of the crane potential energy (\ref{eq:U}), \textit{i.e.}, $g(q) = \frac{\partial U(q)}{\partial{\dot{q}}}$
\end{Properties}

\section{Control objective}\label{sec:cntr}

The control objective consists of two main tasks: \textit{i)} move the crane to the desired configuration, \textit{ii)} dampen the load swings at the same time. This control objective can be described compactly as

\begin{equation}\label{eq:constraints}
\begin{split}
\lim_{t\to \infty}[\alpha(t),\beta(t),\gamma(t),d(t),\theta_1(t),\theta_2(t)] = [\alpha_{d},\beta_{d},\gamma_{d},d_{d},0,0], \\
\lim_{t\to \infty}[\dot{\alpha}(t),\dot{\beta}(t),\dot{\gamma}(t),\dot{d}(t),\dot{\theta_1}(t),\dot{\theta_2}(t)] = [0,0,0,0,0,0],
\end{split}
\end{equation}

where $\alpha_d, \beta_d, \gamma_d, d_d$ are the desired references for the actuated states. 

\medskip

To design the control law and to perform the corresponding stability and convergence analysis (see Section 4) we will consider the following reasonable assumptions.

\begin{assumption}\label{ass:1}
The payload swings are such that: $\lvert \theta_{i} \rvert\ <{\frac{\pi}{2}}, i = 1,2$.

\end{assumption}
\begin{assumption}\label{ass:2}
The cable length is always greater than zero to avoid singularity in the model (\ref{eq:modelmatrix}): $d(t) > ,\forall t\geq0$.
\end{assumption}

\begin{assumption}\label{ass:3}
The boom and the jib angles, according to Fig.\ref{fig:crane}, are mechanically constrained in the range:
\begin{align*}
-{\frac{\pi}{2}} < \beta < {\frac{\pi}{2}},\\
-{\frac{\pi}{2}} < \gamma < {\frac{\pi}{2}}.\\ 
\end{align*}
\end{assumption}

It is worth noting that in most real-world crane the Assumption \ref{ass:3} can be more restrictive according to specific mechanical constraints. 

\section{Control design and stability analysis}\label{sec:cnt}

The control strategy proposed in this paper consists of a nonlinear control law based on energy consideration. To design the control law, firstly we considered a Lyapunov function candidate partially based on the energy of system (\ref{eq:modelmatrix}). The resulting control law will be a proportional-derivative (PD) controller with a gravity compensation. A detailed stability analysis of the resulting closed loop system will be carried out based on the \textit{LaSalle's invariance principle}. 

\medskip

In order to develop our control law, we start from a energy function

\begin{equation}\label{eq:E}
    E(t) = {1\over{2}}\dot{q}^TM(q)\dot{q} + mgd(1-C_{\theta_1}C_{\theta_2}),
\end{equation}

where the first term is the kinetic energy of the crane, whereas the second term represents the potential energy of the payload. If one takes the time derivative of (\ref{eq:E}), it follows that
\begin{equation}\label{eq:Edot}
{\dot E(t) = {1\over{2}}\dot{q}^T\dot{M}(q)\dot{q}+\dot{q}^TM(q)\ddot{q}+mg\dot{d}(1-C_{\theta_1}C_{\theta_2})+\dot{\theta}_1mgdS_{\theta_1}C_{\theta_2} +\dot{\theta}_2mgdS_{\theta_2}C_{\theta_1}}.
\end{equation}

Using (\ref{eq:modelmatrix}) and Property 1, it follows that 
\begin{equation}
    {\dot E(t) = \dot{\alpha}u_1+\dot{\beta}(u_2 -gl_BC_{\beta}(m+{1\over{2}}m_B+m_J))+\dot{\gamma}(u_3-gl_JC_{\gamma}(m + {1\over{2}}m_J))+ \dot{d}(u_4+mg)}.
\end{equation}

\medskip

Based on (\ref{eq:E}), we can define the following Lyapunov
function candidate:

\begin{equation}\label{eq:V}
{V(t) = {1\over{2}}\dot{q}^TM(q)\dot{q} + mgd(1-C_{\theta_1} C_{\theta_2} )+{1\over{2}}k_{p\alpha} e_{\alpha}^2 +{1\over{2}}k_{p\beta} e_{\beta}^2 +{1\over{2}}k_{p\gamma} e_{\gamma}^2 +{1\over{2}}k_{pd} e_{d}^2,}
\end{equation}

where $e_{\alpha},e_{\beta},e_{\gamma},e_{d}$ are the error signals defined as:  
\begin{equation}\label{eq:erro}
\begin{split}
e_{\alpha} = \alpha_d - \alpha, \quad\quad e_{\beta} = \beta_d - \beta, \quad\quad
e_{\gamma} = \gamma_d - \gamma, \quad\quad e_{d} = d_d - d.
\end{split}
\end{equation}

Differentiating  (\ref{eq:V}) with respect to the time and using (\ref{eq:modelmatrix}) and Property 1, we obtain
\begin{equation}\label{eq:Vdot}
\begin{split}
\dot V(t) = \dot{\alpha}(u_1-k_{p\alpha} e_{\alpha})+ \dot{\beta}(u_2-k_{p\beta} e_{\beta} -gl_BC_{\beta}(m+{1\over{2}}m_B+m_J)) \\ + \dot{\gamma}(u_3-k_{p\gamma} e_{\gamma} -gl_JC_{\gamma}(m + {1\over{2}}m_J))+\dot{d}(u_4-k_{pd} e_{d}+mg). 
\end{split}
\end{equation}

In order to cancel the gravitational terms and keep $\dot{V}(t)$ non-positive, the following control law is designed
\begin{equation}\label{eq:u1}
u_1 = k_{p\alpha} e_{\alpha} - k_{d\alpha} \dot{\alpha}, 
\end{equation}
\begin{equation}\label{eq:u2}
u_2 = k_{p\beta} e_{\beta} - k_{d\beta} \dot{\beta} + gl_BC_{\beta}(m+{1\over{2}}m_B+m_J), 
\end{equation}
\begin{equation}\label{eq:u3}
u_3 = k_{p\gamma} e_{\gamma} - k_{d\gamma} \dot{\gamma} +gl_JC_{\gamma}(m + {1\over{2}}m_J), 
\end{equation}
\begin{equation}\label{eq:u4}
u_4 = k_{pd} e_{d} - k_{dd} \dot{d} - mg, 
\end{equation}

where $k_{p\alpha}$, $k_{p\beta}$, $k_{p\gamma}$, $k_{pd}$, $k_{d\alpha}$, $k_{d\beta}$, $k_{d\gamma}$, $k_{dd}$ $\in\mathbb{R}$ are positive control gains.
\newline
Substituting (\ref{eq:u1})-(\ref{eq:u4}) into (\ref{eq:Vdot}) one obtains
\begin{equation}\label{eq:Vdot2}
{\dot V(t) = -k_{d\alpha} \dot{\alpha}^2 - k_{d\beta} \dot{\beta}^2 -k_{d\gamma} \dot{\gamma}^2 - k_{dd} \dot{d}^2 \leq 0. } 
\end{equation}

\medskip

The following theorem describes the stability property of the crane using the control law  (\ref{eq:u1})-(\ref{eq:u4}).

\begin{theorem}
Consider the system (\ref{eq:alp})-(\ref{eq:th2}). Under Assumptions \ref{ass:1}-\ref{ass:3}, the control law (\ref{eq:u1})-(\ref{eq:u4}) makes every equilibrium point (\ref{eq:constraints}) asymptotically stable. 
\end{theorem}

\begin{proof}
 Choosing (\ref{eq:V}) as a Lyapunov function with the system (\ref{eq:alp})-(\ref{eq:th2}), the control law (\ref{eq:u1})-(\ref{eq:u4}) leads to (\ref{eq:Vdot2}). Noticing that $V(0)$ is bounded, it is easy to infer that:
 \begin{equation}\label{cond}
    {V(t) \in L_\infty \Rightarrow \dot{q},e_{\alpha},e_{\beta},e_{\gamma},e_{d},\theta_1,\theta_2 \in L_\infty. }
\end{equation}

 At this point, let $\Phi$ be defined as the set where $\dot{V}(t)=0$, \textit{i.e.},
 \begin{equation}
     \Phi = {\{q,\dot{q}|\dot{V}(t) = 0\}}.
 \end{equation}

Furthermore, let $\Gamma$ represent the largest invariant set in $\Phi$ where Assumptions 1-2 are verified. Based on (\ref{eq:Vdot2}), $\Gamma$ is the set such that:
 \begin{equation}\label{eq:inv}
 \begin{split}
    {\dot{\alpha}=0,\dot{\beta}=0,\dot{\gamma}=0,\dot{d}=0 \Rightarrow \ddot{\alpha}=0,\ddot{\beta}=0,\ddot{\gamma}=0,\ddot{d}=0}, \\
    {\dot{e}_{\alpha} = 0, \dot{e}_{\beta} = 0, \dot{e}_{\gamma} = 0, \dot{e}_{d} = 0 \Rightarrow e_{\alpha} = \phi_1,e_{\beta} = \phi_2, e_{\gamma} = \phi_3 ,e_{d} = \phi_4, }
 \end{split}
 \end{equation}
where $\phi_{1,2,3,4}$ are constants to be determined.\newline
Plugging (\ref{eq:inv}) and (\ref{eq:u1}) in (\ref{eq:alp}) one obtains
\begin{equation}\label{eq:th1eq}
\begin{split}
d^2\ddot{\theta_1}C_{\theta_1}C_{\theta_2}S_{\theta_2}-d^2\ddot{\theta_2}S_{\theta_1}-2d^2\dot{\theta}_1\dot{\theta}_2C_{\theta_1}-A_4\ddot{\theta_1}^2C_{\bar{\beta}}C_{\theta_2}S_{\theta_1}
\\ -A_5d\dot{\theta}_1^2 C_{\bar{\gamma}}  C_{\theta_2}S_{\theta_1}-A_4d\dot{\theta}_2^2  C_{\bar{\beta}}  C_{\theta_2}  S_{\theta_1} - A_5 d \dot{\theta}_2^2  C_{\bar{\gamma}}  C_{\theta_2}  S_{\theta_1} \\
- d^2 \dot{\theta}_1^2  C_{\theta_2}  S_{\theta_1}  S_{\theta_2} + A_4 d \ddot{\theta_1}  C_{\bar{\beta}}  C_{\theta_1}  C_{\theta_2} + A_5 d \ddot{\theta_1}  C_{\bar{\gamma}}  C_{\theta_1}  C_{\theta_2} \\ - A_4 d \ddot{\theta_2}  C_{\bar{\beta}}  S_{\theta_1}  S_{\theta_2} - A_5 d \ddot{\theta_2}  C_{\bar{\gamma}}  S_{\theta_1}  S_{\theta_2} 
+ 2 d^2 \dot{\theta}_1 \dot{\theta}_2  C_{\theta_1}  C_{\theta_2}^2 \\ - 2A_4 d \dot{\theta}_1 \dot{\theta}_2  C_{\bar{\beta}}  C_{\theta_1}  S_{\theta_2} - 2A_5 d \dot{\theta}_1 \dot{\theta}_2  C_{\bar{\gamma}}  C_{\theta_1}  S_{\theta_2} = k_{p\alpha}e_{\alpha}, 
\end{split}
\end{equation}
where $C_{\bar{\beta}} = cos(\beta_d - \phi_2)$ and $C_{\bar{\gamma}} = cos(\gamma_d - \phi_3)$.
\newline
To continue the analysis, equation (\ref{eq:th1eq}) can be rewritten as
\begin{equation}\label{eq:derth1}
    \frac{d}{dt}[\dot{\theta}_1C_{\theta_1}C_{\theta_2}(d^2S_{\theta_2}+ A_4dC_{\bar{\beta}}+A_5dC_{\bar{\gamma}})
    -\dot{\theta}_2S_{\theta_1}(d^2+ A_4dC_{\bar{\beta}}S_{\theta_2}+A_5dC_{\bar{\gamma}}S_{\theta_2})] = k_{p\alpha}e_{\alpha}.
\end{equation}
Integrating equation (\ref{eq:derth1}) with respect to the time we obtain
\begin{equation}\label{eq:intth}
    [\dot{\theta}_1C_{\theta_1}C_{\theta_2}(d^2S_{\theta_2}+ A_4dC_{\bar{\beta}}+A_5dC_{\bar{\gamma}})
    -\dot{\theta}_2S_{\theta_1}(d^2+ A_4dC_{\bar{\beta}}S_{\theta_2}+A_5dC_{\bar{\gamma}}S_{\theta_2})] = k_{p\alpha}e_{\alpha}t + c_1,
\end{equation}
where $c_1$ denotes a constant to be determined.\newline

If $k_{p\alpha}e_{\alpha} \ne 0$, then if $t \Rightarrow \infty$, one would have that the left side of equation (\ref{eq:intth}) tends to infinity, \textit{i.e.}
\begin{equation*}
    [\dot{\theta}_1C_{\theta_1}C_{\theta_2}(d^2S_{\theta_2}+ A_4dC_{\bar{\beta}}+A_5dC_{\bar{\gamma}})
    -\dot{\theta}_2S_{\theta_1}(d^2+ A_4dC_{\bar{\beta}}S_{\theta_2}+A_5dC_{\bar{\gamma}}S_{\theta_2})] \Rightarrow \infty,
\end{equation*}

which conflicts with (\ref{cond}).\newline
As a result:
\begin{align}\label{eq:intth2}
    [\dot{\theta}_1C_{\theta_1}C_{\theta_2}(d^2S_{\theta_2}+ A_4dC_{\bar{\beta}}+A_5dC_{\bar{\gamma}})
    -\dot{\theta}_2S_{\theta_1}(d^2+ A_4dC_{\bar{\beta}}S_{\theta_2}+A_5dC_{\bar{\gamma}}S_{\theta_2})] = c_1 \\
    k_{p\alpha}e_{\alpha} = 0.
\end{align}
Since $k_{p\alpha} > 0$, it is clear that
\begin{equation}\label{eq:erral}
    e_{\alpha} = 0 \Rightarrow \phi_1 = 0 \Rightarrow \alpha = \alpha_d.
\end{equation}
Similar to (\ref{eq:th1eq})-(\ref{eq:erral}), by plugging (\ref{eq:inv}) and (\ref{eq:u2})-(\ref{eq:u3}) into (\ref{eq:bet})-(\ref{eq:gam}) one achieves
\begin{equation}\label{eq:th2eq}
\begin{split}
A_4d\dot{\theta}_2^2S_{\bar{\beta}}S_{\theta_2} 
- A_4d\ddot{\theta_2}C_{\theta_2}S_{\bar{\beta}} + A_4d\dot{\theta}_1^2C_{\bar{\beta}}C_{\theta_1}C_{\theta_2} + A_4d\dot{\theta}_2^2C_{\bar{\beta}}C_{\theta_1}C_{\theta_2} \\ + A_4d\ddot{\theta_1}C_{\bar{\beta}}C_{\theta_2}S_{\theta_1} 
+ A_4d\ddot{\theta_2}C_{\bar{\beta}}C_{\theta_1}S_{\theta_2} - 2A_4d\dot{\theta}_2\dot{\theta}_1C_{\bar{\beta}}S_{\theta_1}S_{\theta_2} =  k_{p\beta}e_{\beta}, 
\end{split}
\end{equation}

\begin{equation}\label{eq:th3eq}
\begin{split}
A_5d\dot{\theta}_2^2S_{\bar{\gamma}}S_{\theta_2} 
- A_5d\ddot{\theta_2}C_{\theta_2}S_{\bar{\gamma}} +  A_5d\dot{\theta}_1^2C_{\bar{\gamma}}C_{\theta_1}C_{\theta_2} 
+ A_5d\dot{\theta}_2^2C_{\bar{\gamma}}C_{\theta_1}C_{\theta_2} + A_5d\ddot{\theta_1}C_{\bar{\gamma}}C_{\theta_2}S_{\theta_1} \\ 
+ A_5d\ddot{\theta_2}C_{\bar{\gamma}}C_{\theta_1}S_{\theta_2} - 2A_5d\dot{\theta}_1\dot{\theta}_2C_{\bar{\gamma}}S_{\theta_1}S_{\theta_2} = kp_{\gamma}e_{\gamma}. 
\end{split}
\end{equation}
Equations (\ref{eq:th2eq})-(\ref{eq:th3eq}) can be rewritten as
\begin{equation}\label{eq:derth}
    \frac{d}{dt}[-A_4d\dot{\theta}_2(C_{\theta_2}S_{\bar{\beta}}-C_{\bar{\beta}}C_{\theta_1}S_{\theta_2}) + A_4d\dot{\theta}_1C_{\bar{\beta}}C_{\theta_2}S_{\theta_1}] = k_{p\beta}e_{\beta},
\end{equation}
\begin{equation}\label{eq:derth}
    \frac{d}{dt}[- A_5d\dot{\theta}_2(C_{\theta_2}S_{\bar{\gamma}}+C_{\bar{\gamma}}C_{\theta_1}S_{\theta_2})+ A_5d\ddot{\theta_1}C_{\bar{\gamma}}C_{\theta_2}S_{\theta_1}] = kp_{\gamma}e_{\gamma}.
\end{equation}
In the same way of (\ref{eq:intth})-(\ref{eq:erral}) it is clear that:
\begin{equation}\label{eq:errb}
    e_{\beta} = 0 \Rightarrow \phi_2 = 0 \Rightarrow \beta = \beta_d,
\end{equation}
\begin{equation}\label{eq:errg}
    e_{\gamma} = 0 \Rightarrow \phi_3 = 0 \Rightarrow \gamma = \gamma_d.
\end{equation}
By using the (\ref{eq:inv}) and (\ref{eq:u4}) in (\ref{eq:len})-(\ref{eq:th1})-(\ref{eq:th2}), we get
\begin{equation}\label{eq:eqd}
    - d\dot{\theta}_2^2 - d(\dot{\theta}_1^2C_{\theta_2}^2)+ g(1-C_{\theta_1}C_{\theta_2}) = kp_{d}e_{d},
\end{equation}
\begin{equation}\label{eq:eq1}
    d\ddot{\theta_1}C_{\theta_2}^2 = -gS_{\theta_1}C_{\theta_2}+ 2d\dot{\theta}_1\dot{\theta}_2S_{\theta_2}C_{\theta_2}, 
\end{equation}
\begin{equation}\label{eq:eq2}
    d\ddot{\theta_2} = - {1\over{2}}(d\dot{\theta}_1^2C_{\theta_2}S_{\theta_2})- gS_{\theta_2}C_{\theta_1}.
\end{equation}

Substituting (\ref{eq:erral}) and (\ref{eq:eqd})-(\ref{eq:eq2}) into (\ref{eq:th1eq}), we obtain
\begin{equation}\label{eq:solution}
    \theta_1 = 0 \lor \theta_2 = \pm \pi/2 \lor (\beta = \gamma = \pm \pi/2).   
\end{equation}

According to Assumptions \ref{ass:1} and \ref{ass:3}, the solutions $\beta = \gamma = \pm \pi/2$ and $\theta_2 = \pm \pi/2$ are not considered. Thus, it can be concluded that
\begin{equation}\label{eq:solution2}
    \theta_1 = 0  \Rightarrow  \dot{\theta}_1 = 0 \Rightarrow  \ddot{\theta_1} = 0.
\end{equation}


Using (\ref{eq:errb}) and (\ref{eq:solution2}) into (\ref{eq:th2eq}) one obtains
\begin{equation}\label{eq:th51}
    A_4d\dot{\theta}_2^2C_{\beta_d}C_{\theta_2} + A_4d\dot{\theta}_2^2S_{\beta_d}S_{\theta_2} + A_4d\ddot{\theta_2}C_{\beta_d}S_{\theta_2} - A_4d\ddot{\theta_2}C_{\theta_2}S_{\beta_d}  = 0.
\end{equation}


Taking into account Assumption 2, equation (\ref{eq:th51}) can be rewritten as:
\begin{equation}\label{eq:derth51}
    \frac{d}{dt}[\dot{\theta}_2(C_{\beta_d}S_{\theta_2}-C_{\theta_2}S_{\beta_d})] = 0.
\end{equation}
After that, one integrates (\ref{eq:derth51}) in respect to the time
\begin{equation}\label{eq:derth52}
    [\dot{\theta}_2(C_{\beta_d}S_{\theta_2}-C_{\theta_2}S_{\beta_d})] = c_2,
\end{equation}

where $c_2$ denotes constant to be determined..

The equation (\ref{eq:derth52}) can be rewritten as:
\begin{equation}\label{eq:derth53}
    \frac{d}{dt}[-(C_{\beta_d}C_{\theta_2}+S_{\theta_2}S_{\beta_d})] = c_2.
\end{equation}
Integrating  (\ref{eq:derth53}) one obtains
\begin{equation}\label{eq:derth54}
    [-(C_{\beta_d}C_{\theta_2}+S_{\theta_2}S_{\beta_d})] = c_2t + c_3,
\end{equation}

where $c_3$ is a constant. Assume that $c_2 \neq 0$, then if $t \Rightarrow \infty$, one will have that the left side of equation (\ref{eq:derth54})

\begin{equation*}
   [-(C_{\beta_d}C_{\theta_2}+S_{\theta_2}S_{\beta_d})] \Rightarrow \infty,    
\end{equation*}

which conflicts with \ref{cond}. \newline
Than, $c_2 = 0$ and  in (\ref{eq:derth54}) is easy to verify that $\theta_2$ must be a constant. Consequently, it is clear that:
\begin{equation}\label{eq:solution4}
      \dot{\theta}_2 = 0 \Rightarrow  \ddot{\theta_2} = 0. 
\end{equation}
Including (\ref{eq:solution2}) and (\ref{eq:solution4}) into (\ref{eq:eq2}) and considering the Assumption 1, it can be obtained that:
\begin{equation}\label{eq:solution5}
    gS_{\theta_2} = 0 \Rightarrow {\theta_2} = 0.
\end{equation}
Including (\ref{eq:solution2}) and (\ref{eq:solution5}) into (\ref{eq:eqd}) we obtain
\begin{equation}\label{eq:solution6}
    k_{pd}e_d = 0 \Rightarrow e_d = 0 \Rightarrow d = d_d.
\end{equation}

\end{proof}

\section{Simulation results}\label{sec:sim}

To demonstrate the effectiveness of the control law (\ref{eq:u1})-(\ref{eq:u4}), five different simulation scenarios will be shown. In each of them the goal is to move the crane to a desired position while dampening the oscillations of the payload as much as possible. Finally, at the end of the section, the performances of the proposed control approach are compared with a linear quadratic regulator (LQR) obtained by linearization which represents the most commonly crane control technique used in the literature \citep{ref5}.

\medskip

To get realistic values for the simulation tests, we consider a small knuckle crane: the NK375b from NEMAASKO (Fig.\ref{fig:nem}). See \citep{NK375_user_manual} for more details.

\begin{figure}[ht!]
\centering
\includegraphics[width=0.3\linewidth]{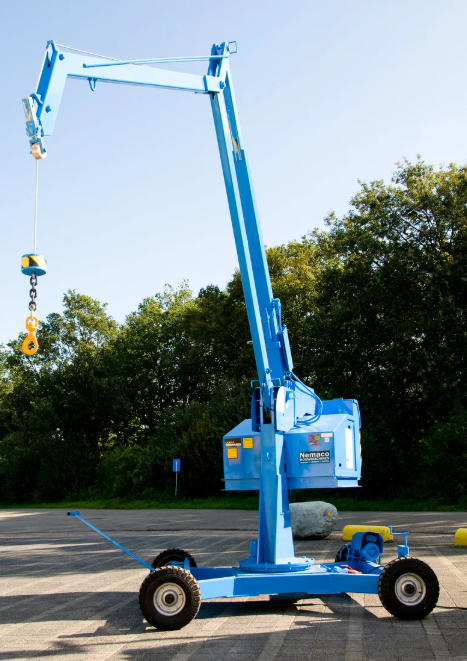}
\caption{\label{fig:nem}NK375b knuckle crane}
\end{figure}

\medskip

In all the simulation scenarios, without loss of generality, the desired angles are selected as $\alpha_d=60^{\circ}$, $\beta_d=30^{\circ}$, $\gamma_d=22^{\circ}$ and the desired cable length is selected as $d = 2\mathrm{m}$ according to the workspace of the knuckle crane. Matlab\textsuperscript{\textregistered} and Simulink\textsuperscript{\textregistered} code is released as open-source on GitHub: \url{https://github.com/MikAmb95/Knuckle_crane_simulator-}.

\begin{table}[hb]
\begin{center}
\caption{Parameters of the knuckle crane system}\label{tb:valpar}
\begin{tabular}{cccc}
Parameters & Physical & Units \\\hline
$m_b$ & 300 & kg \\ 
$m_j$ & 250 & kg \\
m & 100 & kg \\ 
$l_b$ & 2 & m \\ 
$l_j$ & 2.3 & m \\\hline  
\end{tabular}
\end{center}
\end{table}

The parameters for the control law (\ref{eq:u1})-(\ref{eq:u4}) are 

\begin{align*}
     k_{p\alpha} = 10^{3},\quad\quad k_{p\beta} = 10^{4},\quad\quad k_{p\gamma} = 10^{4},\quad\quad k_{pd} = 10^{3},\quad \\k_{d\alpha} = 10^{2},\quad\quad k_{d\beta} = 10^{3},\quad\quad k_{d\gamma} = 10^{3},\quad\quad k_{dd} = 10^{2}.
\end{align*}

\textit{Scenario 1.} In this simulation we show the performance of the control law (\ref{eq:u1})-(\ref{eq:u4}) considering for the parameters of the kuckle crane model (\ref{eq:modelmatrix}) the nominal values collected in Tab.\ref{tb:valpar} and setting zero initial conditions for the two swing angles of the payload (\textit{i.e.} $\theta_1$ and $\theta_2$). The simulation results are reported in Figg.\ref{fig:al}-\ref{fig:d}. We can see that the three actuated angles (\textit{i.e.} $\alpha, \beta,$ and $\gamma$) reach  the desired angular values in around 100 seconds. Additionally, the cable achieves the desired length. As one can see in Figg.\ref{fig:th1}-\ref{fig:th2}, the payload swing angles (\textit{i.e.} $\theta_1$ and $\theta_2$) exhibit a first oscillation due to the movement of the crane. However, when the crane reaches the desired position, the swing angles have negligible residual oscillations. 
As one can see from Figg.\ref{fig:u1}-\ref{fig:u2}, the inputs profile and values are reasonable and well within the typical limits of the crane actuators.

\medskip

\begin{figure}[ht!]
\centering

\subfloat[Scenario 1. Tower angle $\alpha$. Blue line: Nonlinear controller. Red line: Desired reference.]{%
\resizebox*{5cm}{!}{\includegraphics{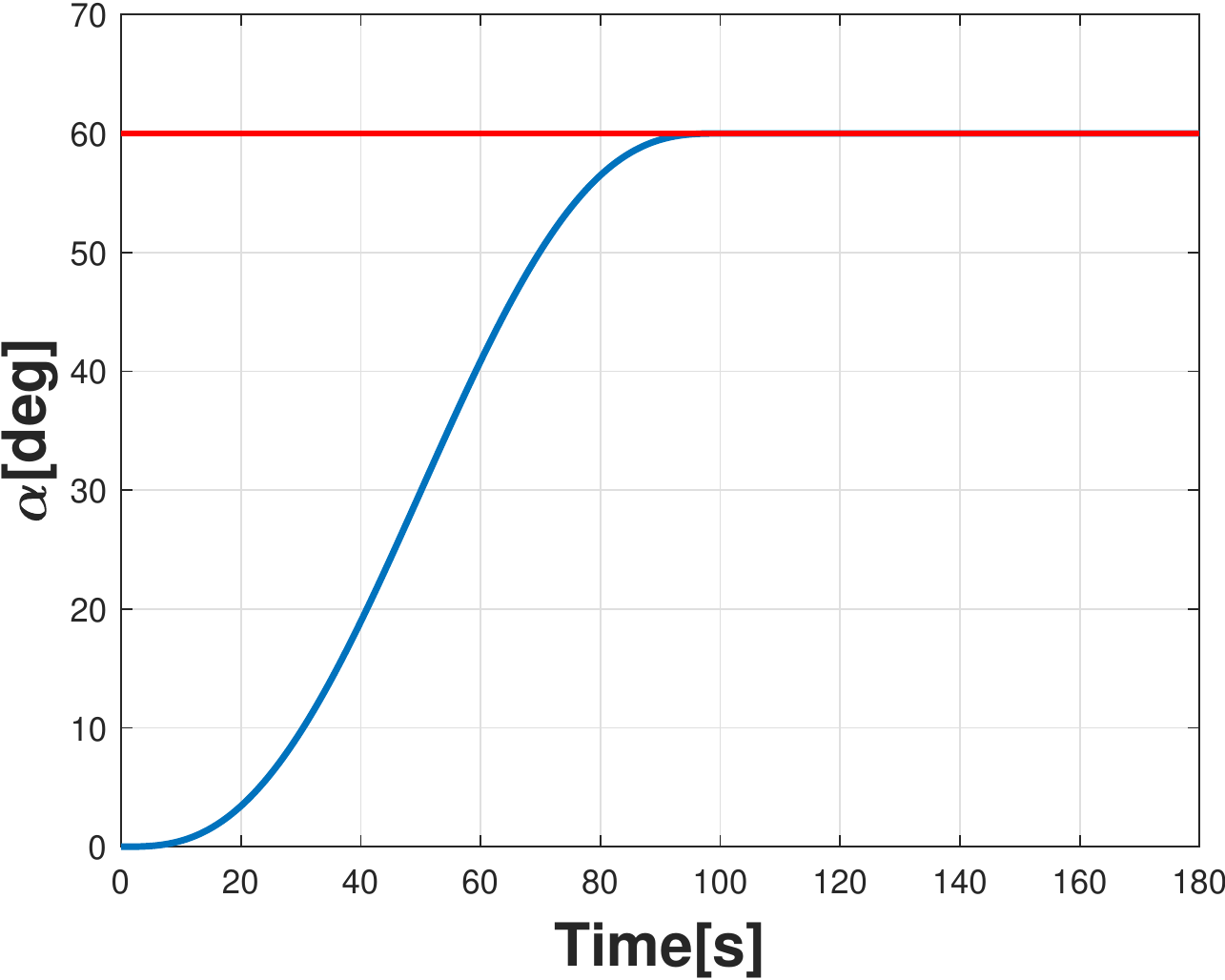}}\label{fig:al}}\hspace{100pt}
\subfloat[Scenario 1. Boom angle $\beta$. Blue line: Nonlinear controller. Red line: Desired reference.]{%
\resizebox*{5cm}{!}{\includegraphics{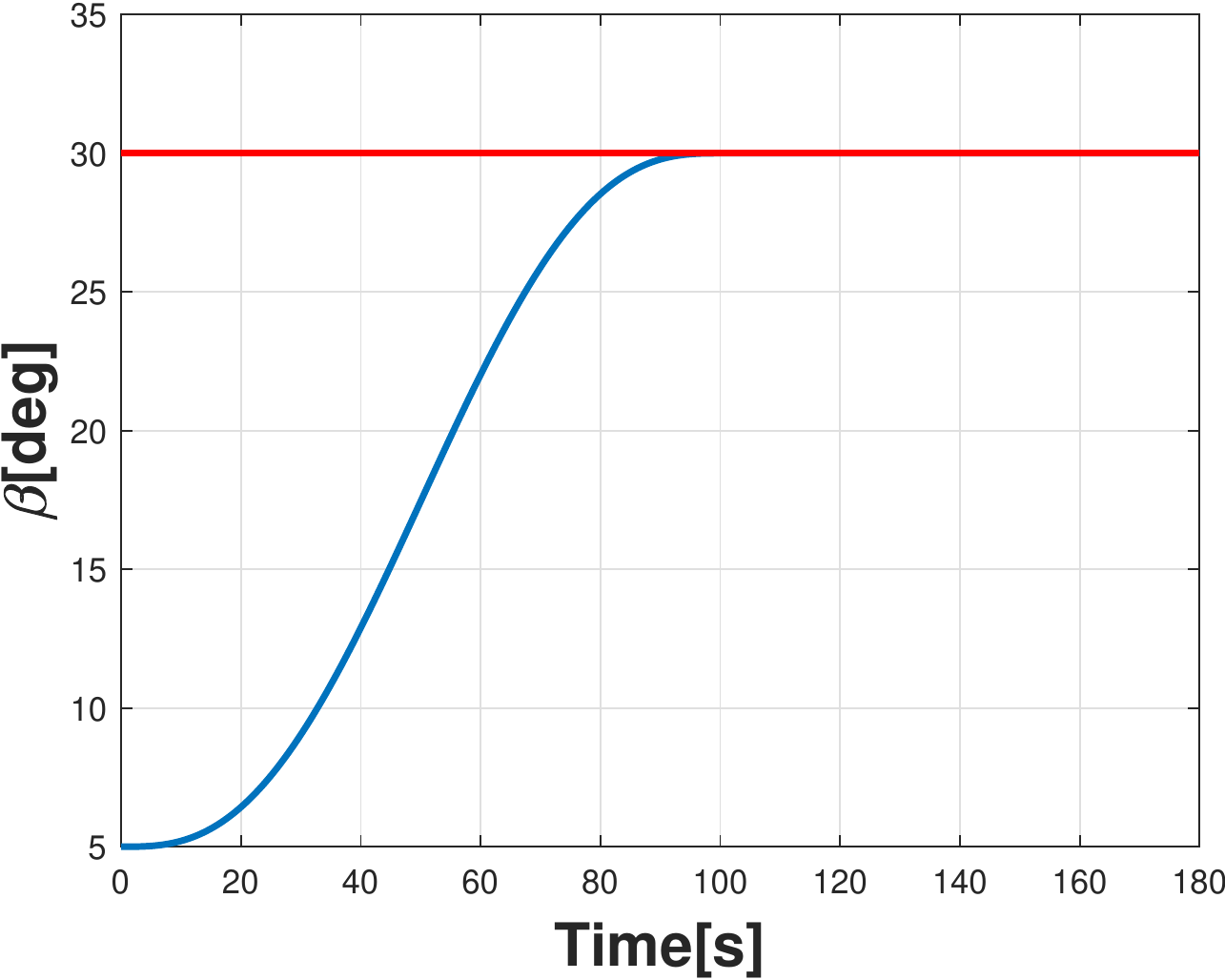}}\label{fig:by}}
\caption{}
\end{figure}

\begin{figure}[ht!]
\centering

\subfloat[Scenario 1. Jib angle $\gamma$. Blue line: Nonlinear controller. Red line: Desired reference.]{%
\resizebox*{5cm}{!}{\includegraphics{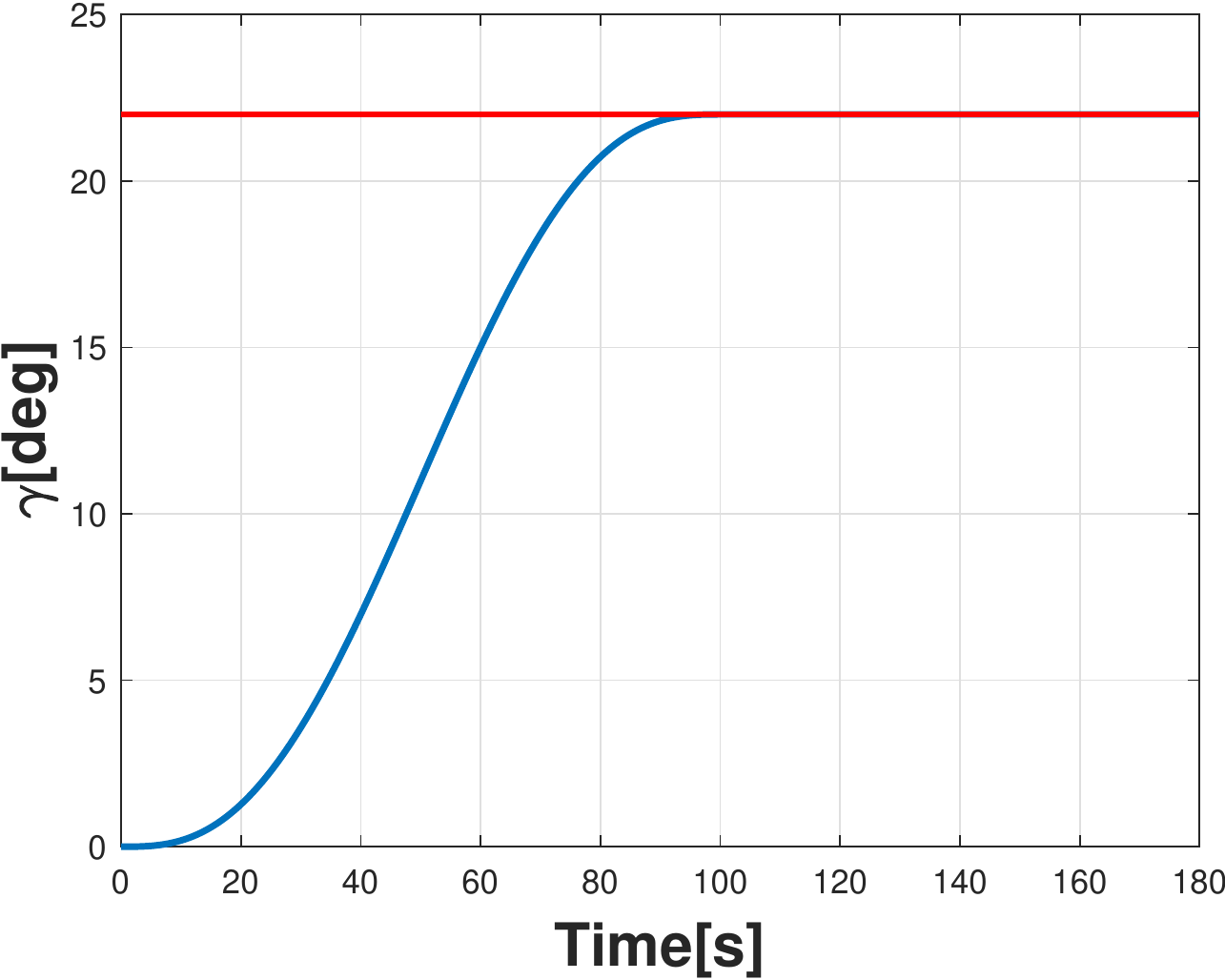}}\label{fig:gm}}\hspace{100pt}
\subfloat[Scenario 1. Rope length $d$. Blue line: Nonlinear controller. Red line: Desired reference.]{%
\resizebox*{5cm}{!}{\includegraphics{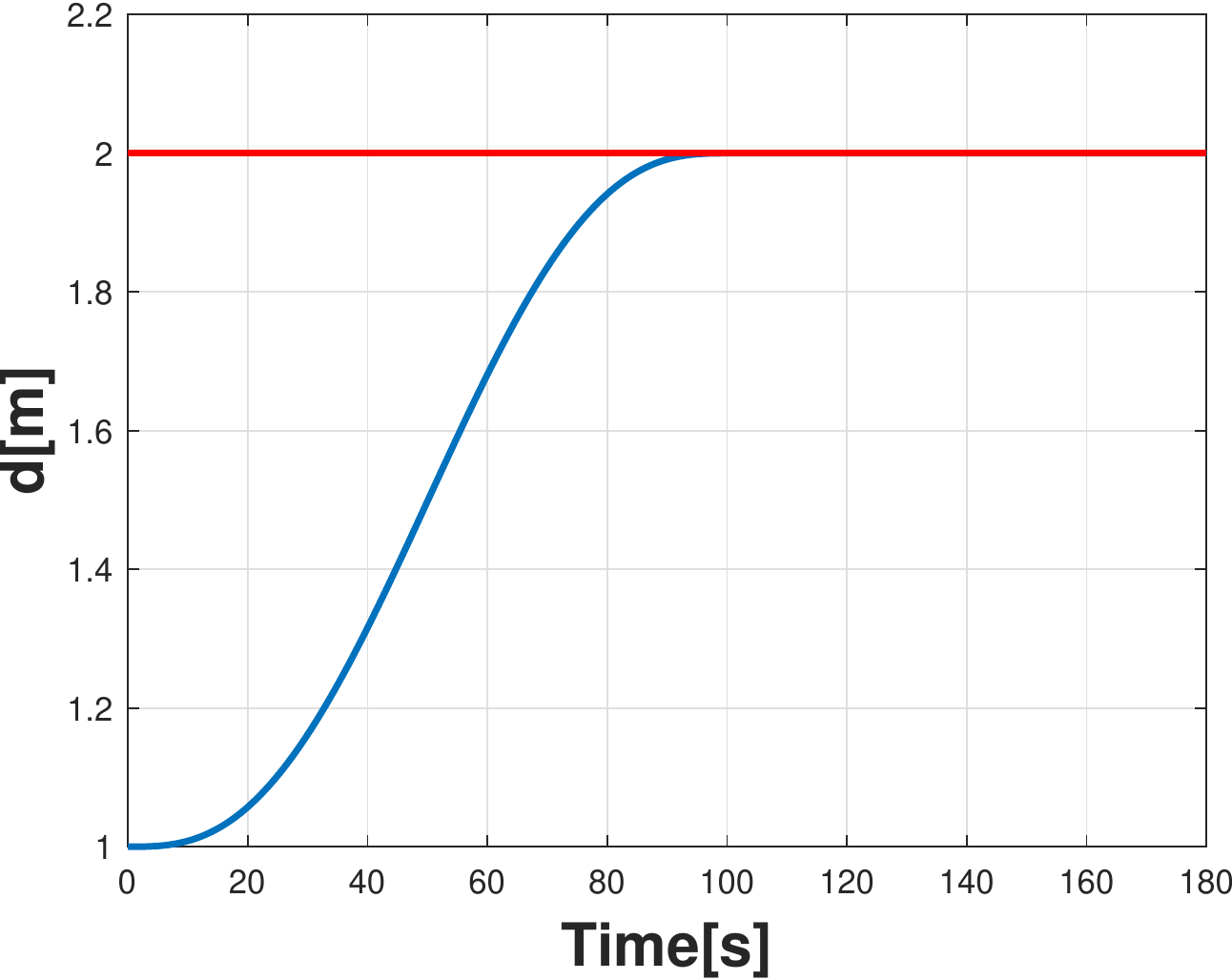}}\label{fig:d}}
\caption{}
\end{figure}

\begin{figure}[ht!]
\centering

\subfloat[Scenario 1. Payload swing angle $\theta_1$.]{%
\resizebox*{5cm}{!}{\includegraphics{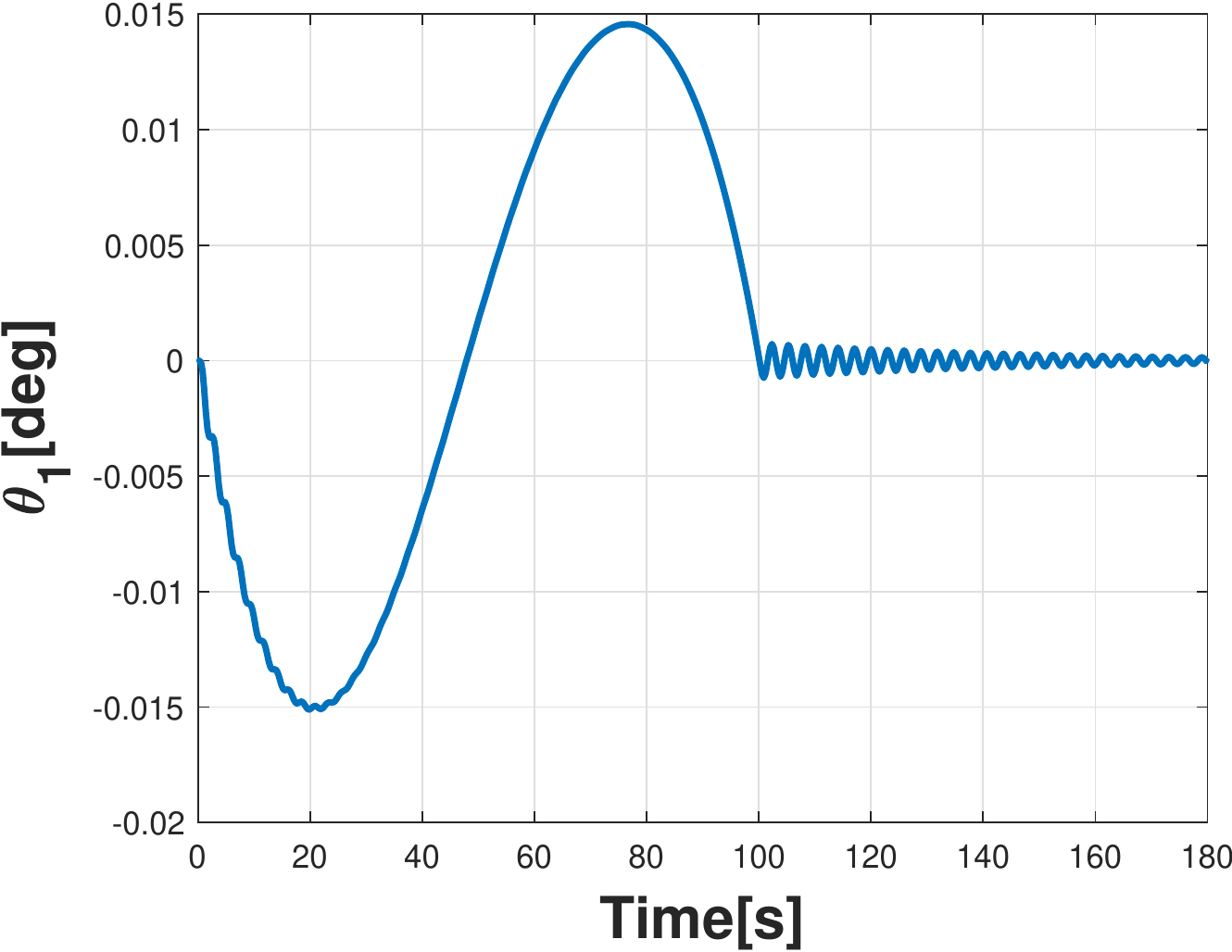}}\label{fig:th1}}\hspace{100pt}
\subfloat[Scenario 1. Payload swing angle $\theta_2$]{%
\resizebox*{5cm}{!}{\includegraphics{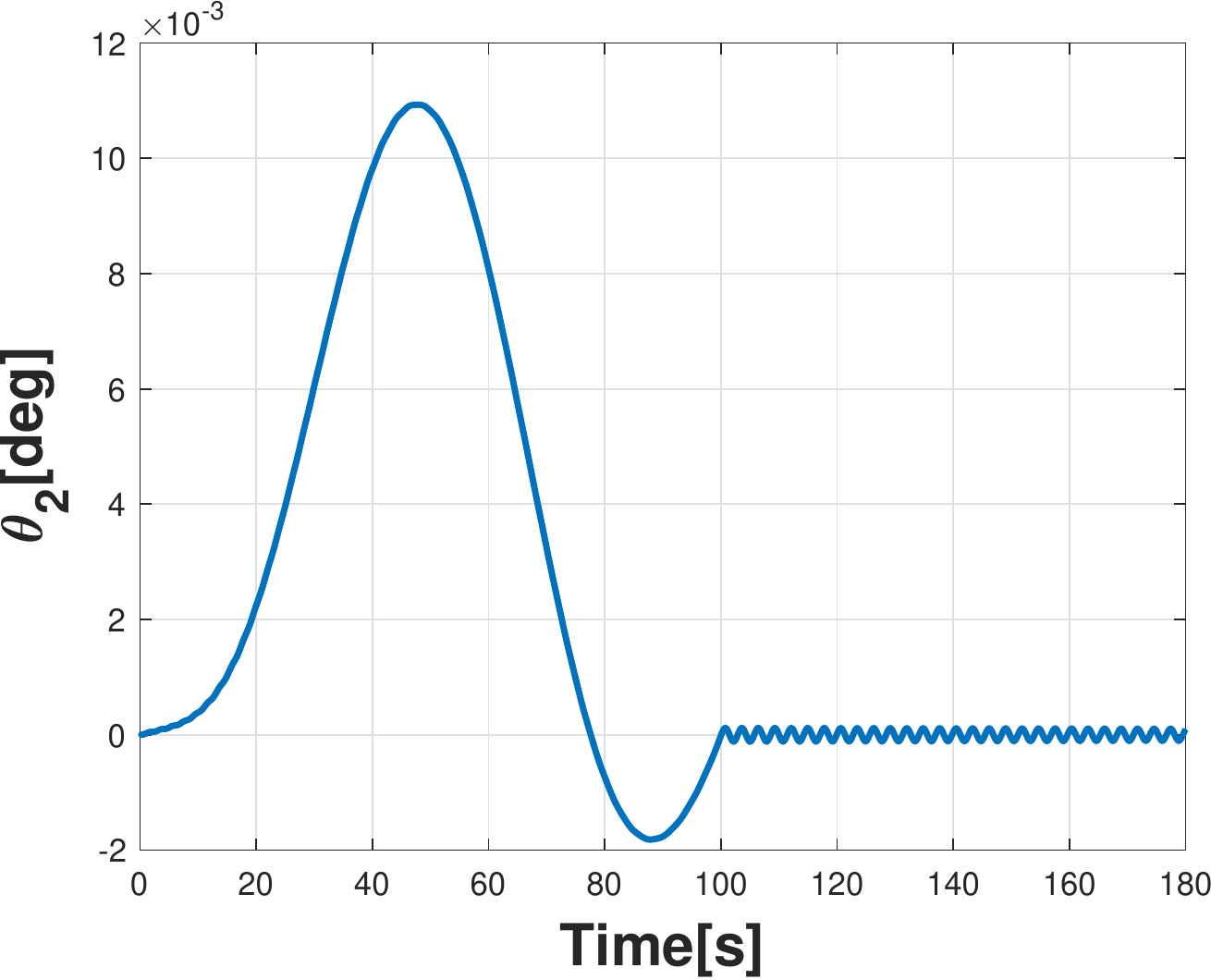}}\label{fig:th2}}
\caption{}
\end{figure}

\begin{figure}[ht!]
\centering
\subfloat[Scenario 1. Control input $u_1, u_2$.]{%
\resizebox*{5cm}{!}{\includegraphics{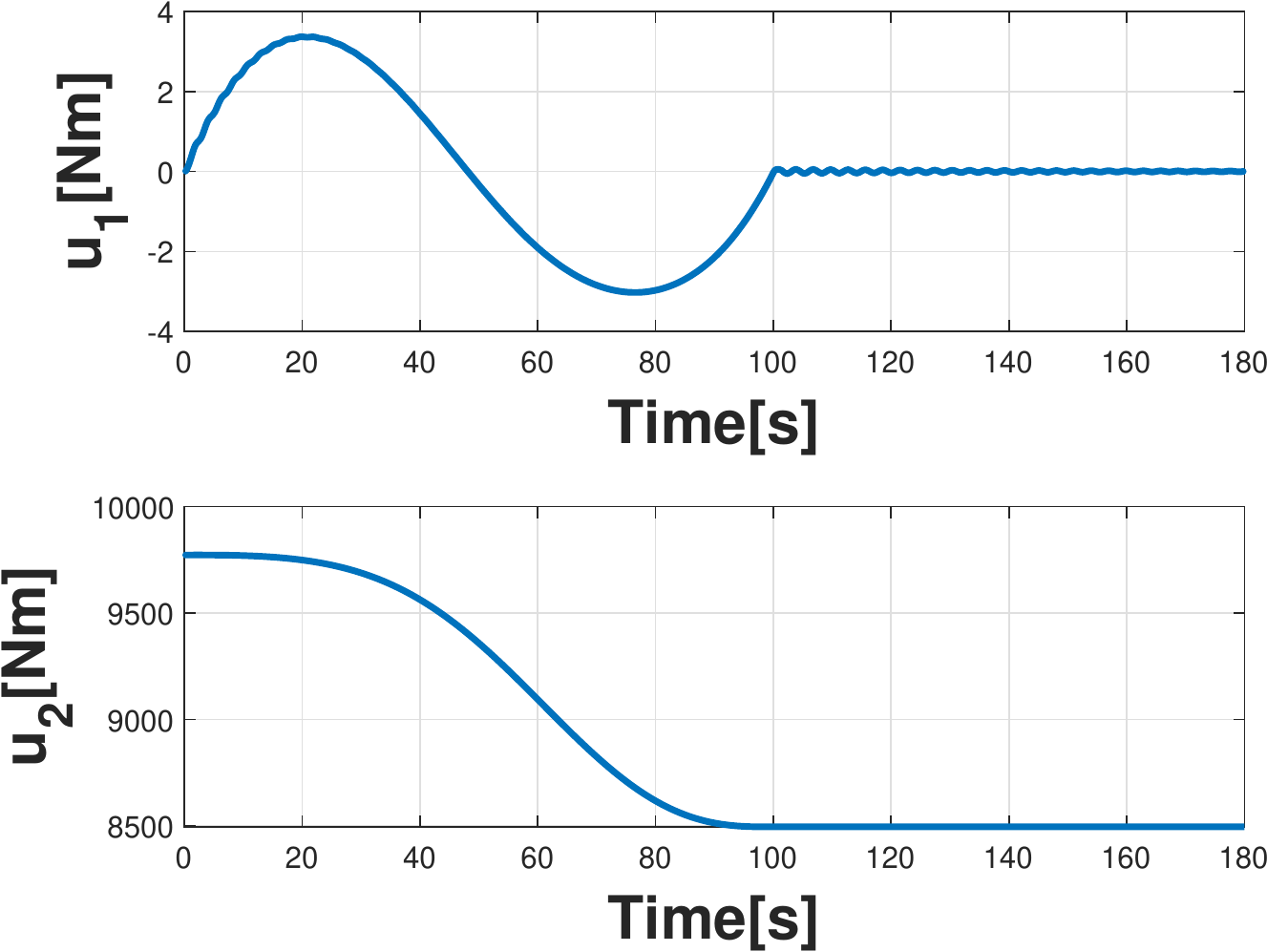}}\label{fig:u1}}\hspace{100pt}
\subfloat[Scenario 1. Control input $u_3, u_4$.]{%
\resizebox*{5cm}{!}{\includegraphics{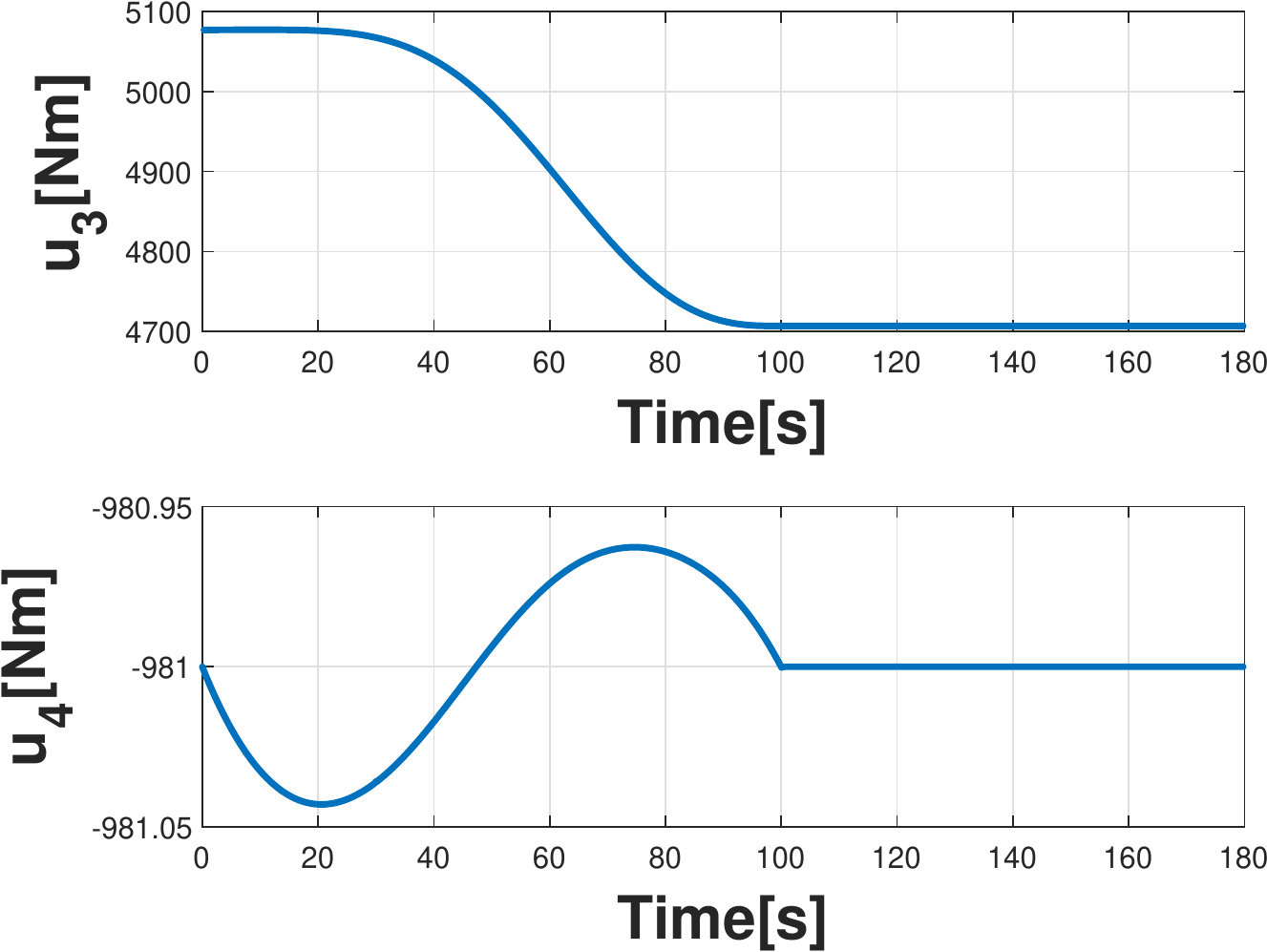}}\label{fig:u2}}
\caption{}
\end{figure}

\textit{Scenario 2.} In this simulation the payload is perturbed by
setting as initial swing angles: $\theta_1(0) \approx 11.5^{\circ}$ and $\theta_2(0) \approx 5.7^{\circ}$ as shown in Figg. \ref{fig:th11}-\ref{fig:th21}. The proposed control scheme is able to steer the knuckle crane to the desired configuration (see Figg.\ref{fig:al1}-\ref{fig:d1}). The behavior of the crane is similar to the previous case. However, the control input have an oscillating trend (see Fig.\ref{fig:u11}-\ref{fig:u21}) due to the residual oscillations of the payload swing angles (see Figg.\ref{fig:th11}-\ref{fig:th21}). Due to the initial perturbations, the residual oscillations of the payload swing angles $\theta_1$ and $\theta_2$ are confined both within $1^{\circ}$ in the time window considered in the simulations.

\medskip

The reader is referred to \url{https://youtu.be/aBA6CoGARvs} for a video containing extra material.

\medskip

\begin{figure}[ht]
\centering

\subfloat[Scenario 2. Tower angle $\alpha$. Blue line: Nonlinear controller. Red line: Desired reference.]{%
\resizebox*{5cm}{!}{\includegraphics{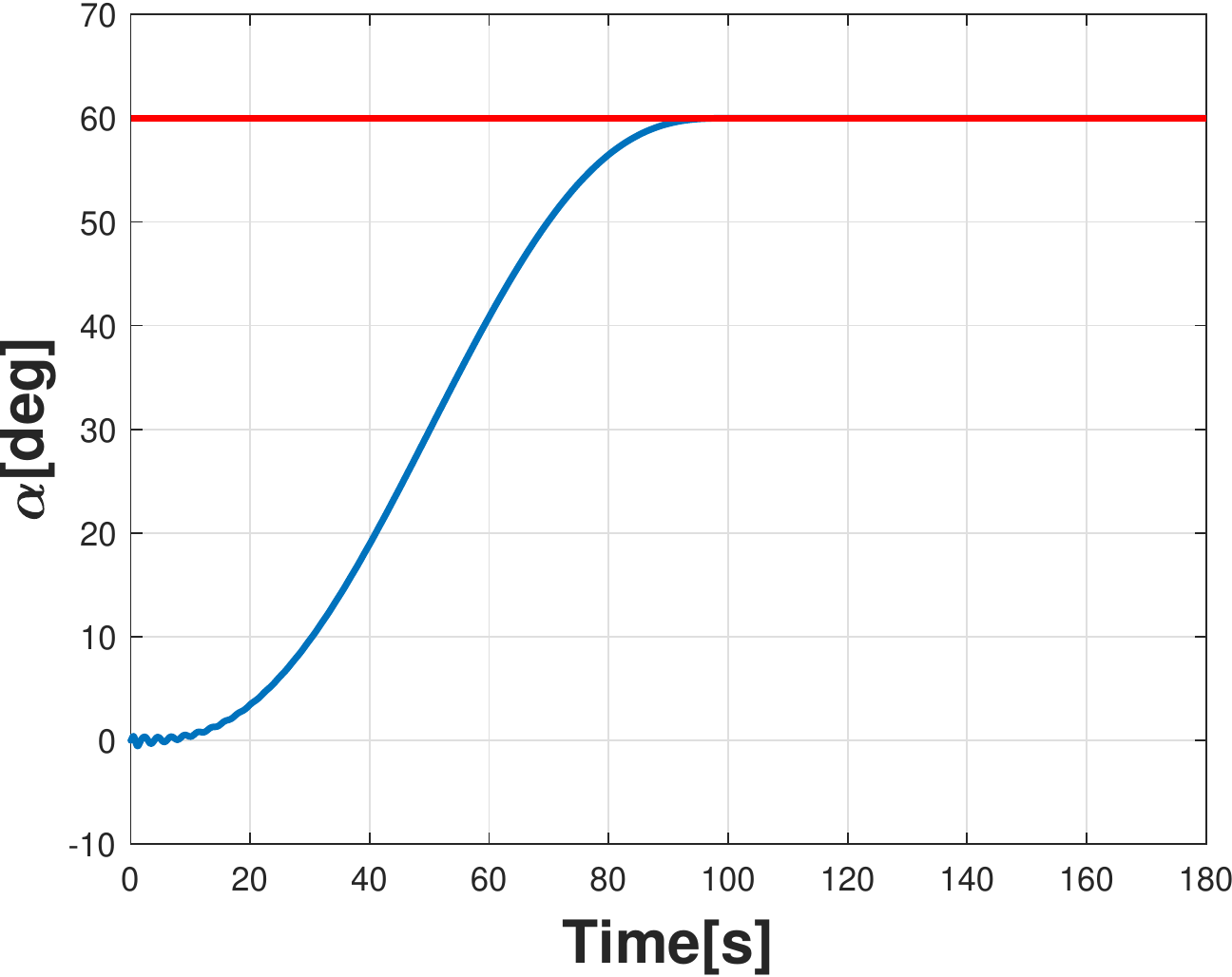}}\label{fig:al1}}\hspace{100pt}
\subfloat[Scenario 2. Boom angle $\beta$. Blue line: Nonlinear controller. Red line: Desired reference.]{%
\resizebox*{5cm}{!}{\includegraphics{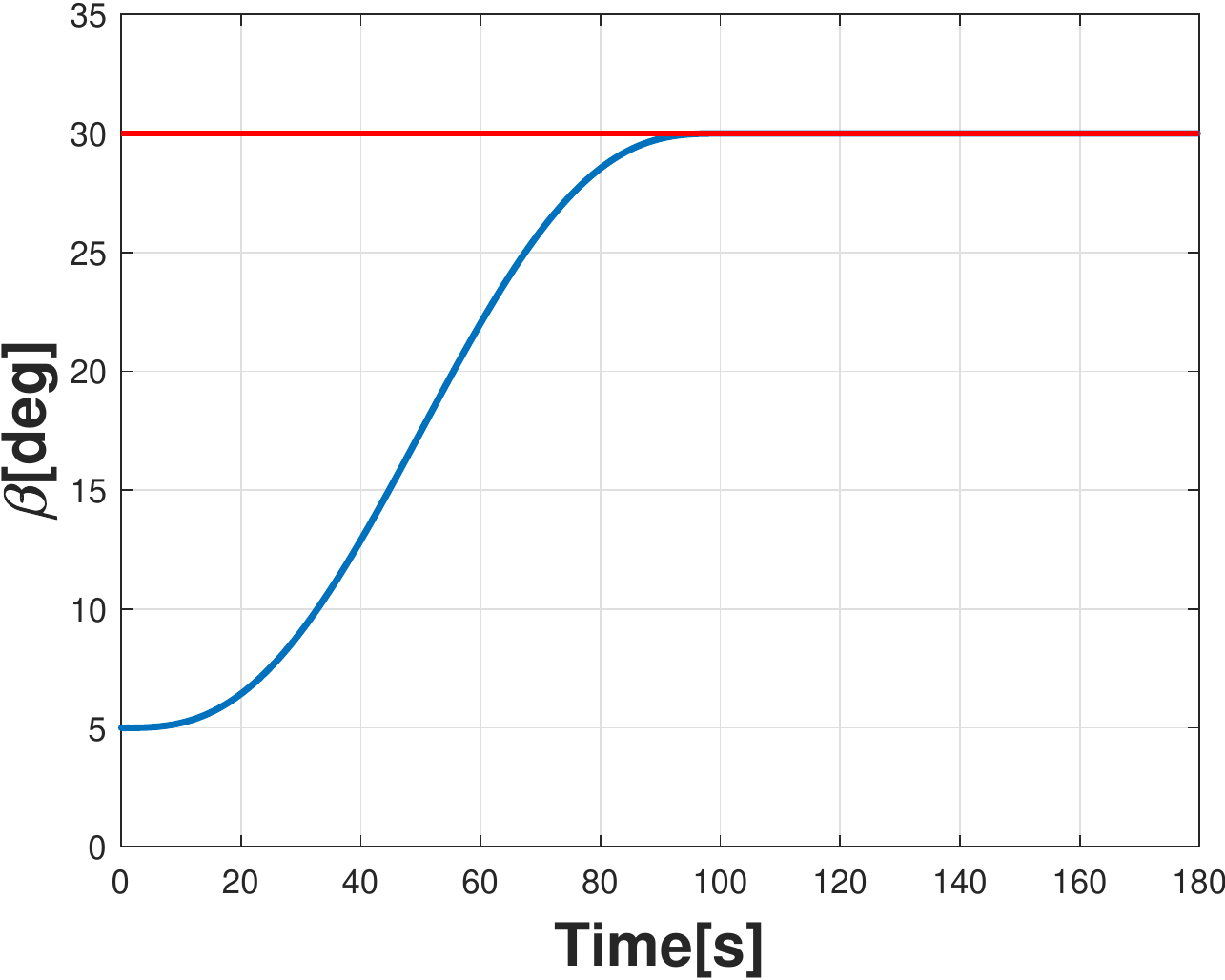}}\label{fig:bt1}}
\caption{}
\end{figure}

\begin{figure}[ht]
\centering

\subfloat[Scenario 2. Jib angle $\gamma$. Blue line: Nonlinear controller. Red line: Desired reference.]{%
\resizebox*{5cm}{!}{\includegraphics{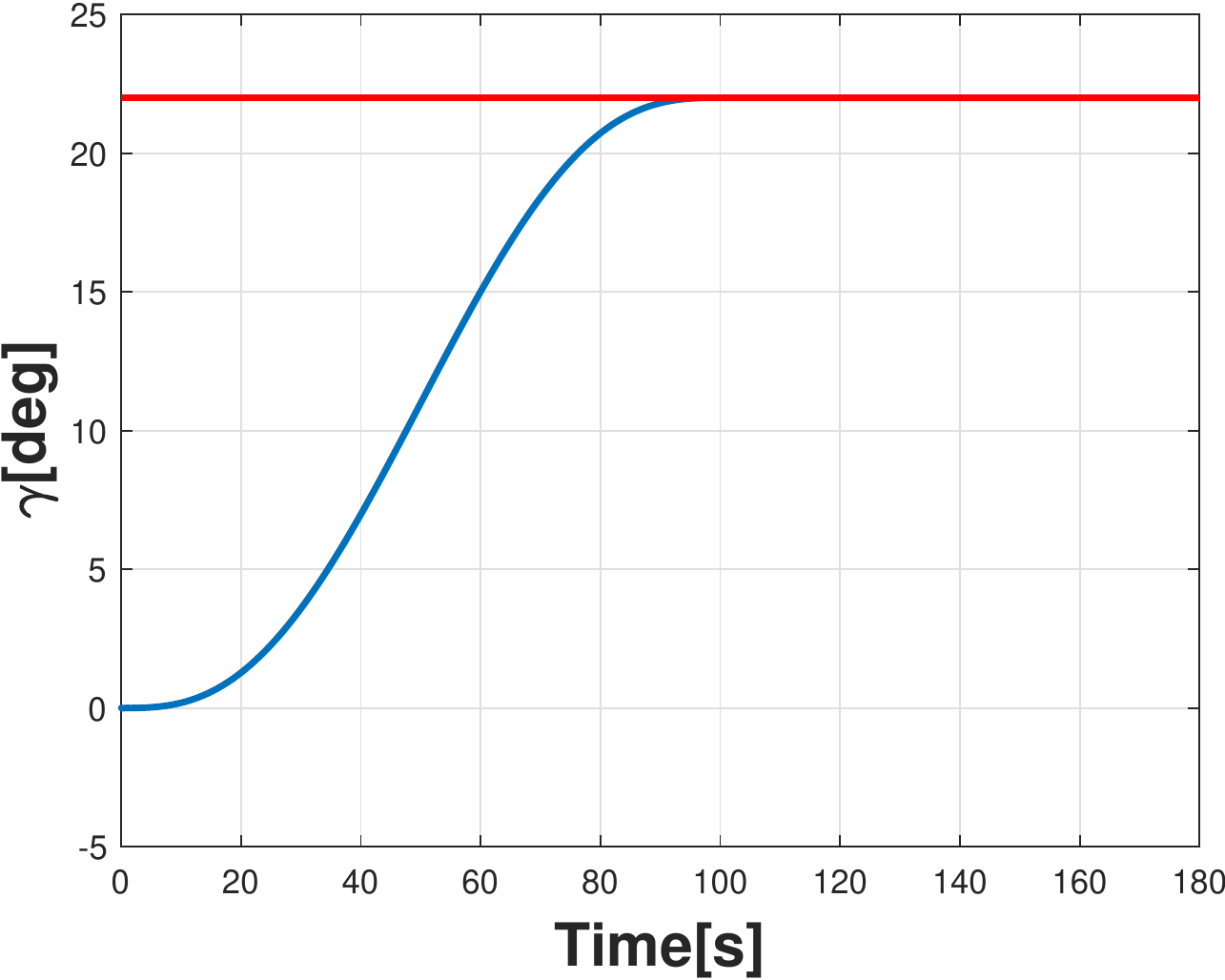}}\label{fig:gm1}}\hspace{100pt}
\subfloat[Scenario 2. Rope length $d$. Blue line: Nonlinear controller. Red line: Desired reference.]{%
\resizebox*{5cm}{!}{\includegraphics{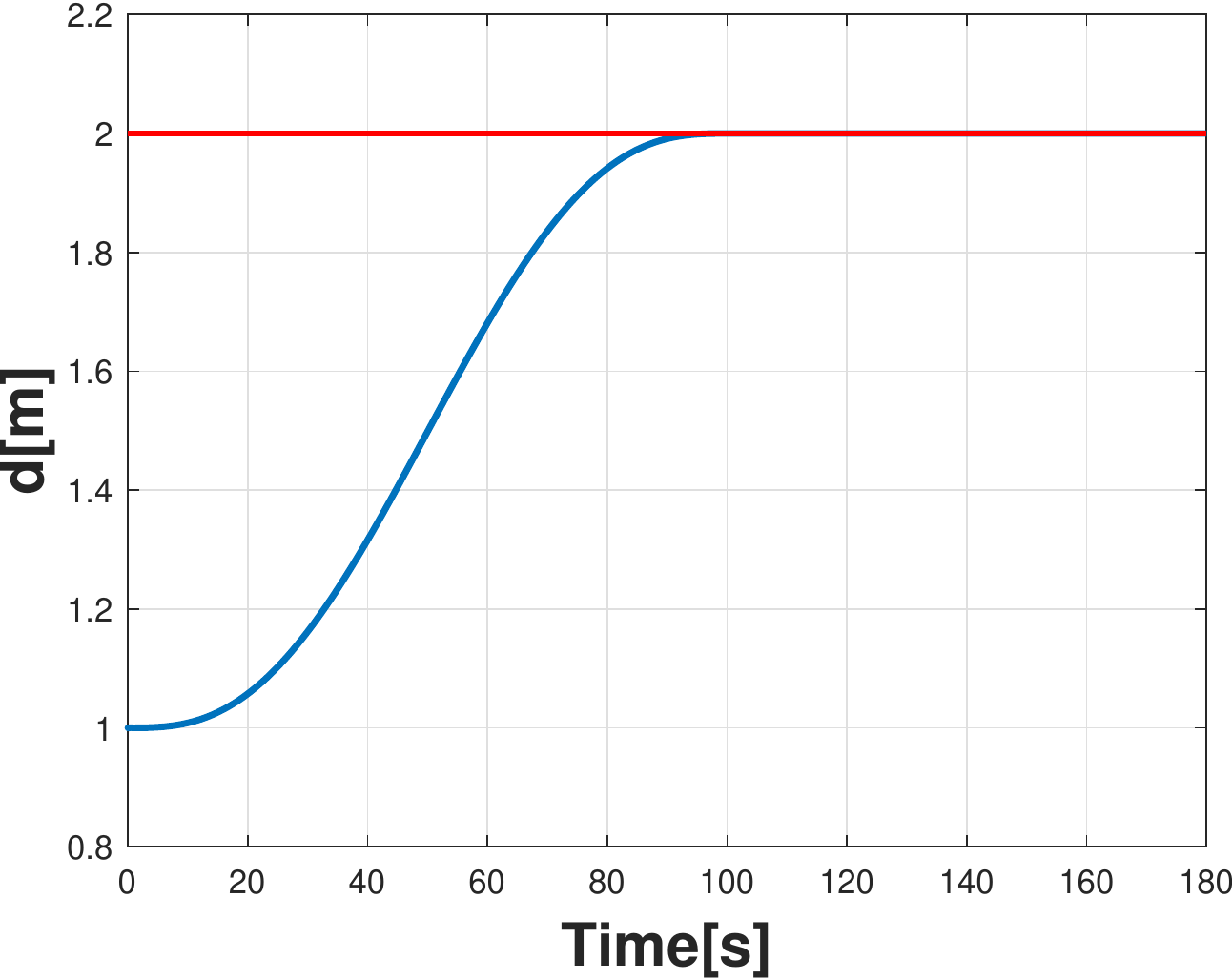}}\label{fig:d1}}
\caption{}
\end{figure}

\begin{figure}[ht]
\centering

\subfloat[Scenario 2. Payload swing angle $\theta_1$.]{%
\resizebox*{5cm}{!}{\includegraphics{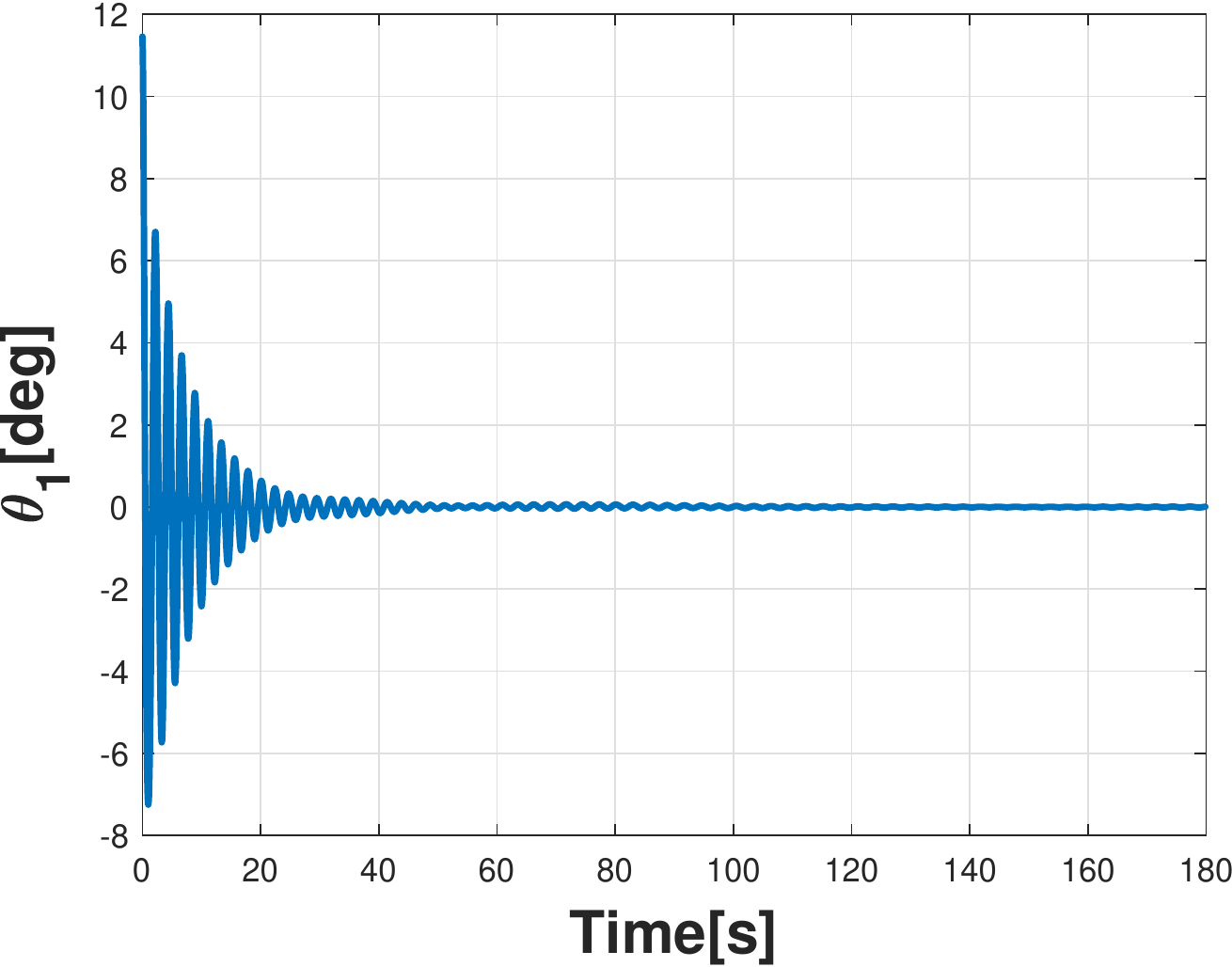}}\label{fig:th11}}\hspace{100pt}
\subfloat[Scenario 2. Payload swing angle $\theta_2$]{%
\resizebox*{5cm}{!}{\includegraphics{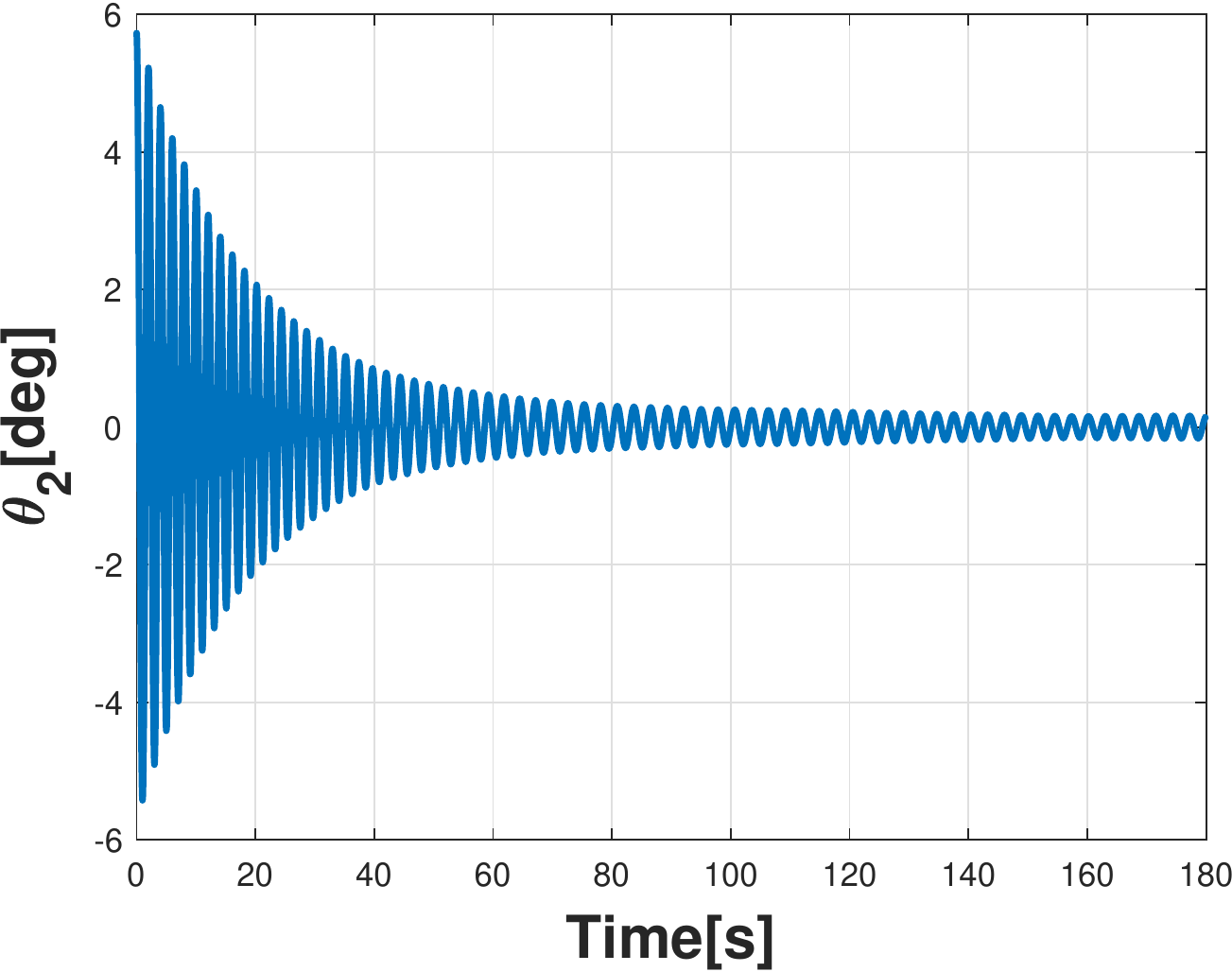}}\label{fig:th21}}
\caption{}
\end{figure}

\begin{figure}[ht]
\centering
\subfloat[Scenario 2. Control input $u_1, u_2$.]{%
\resizebox*{5cm}{!}{\includegraphics{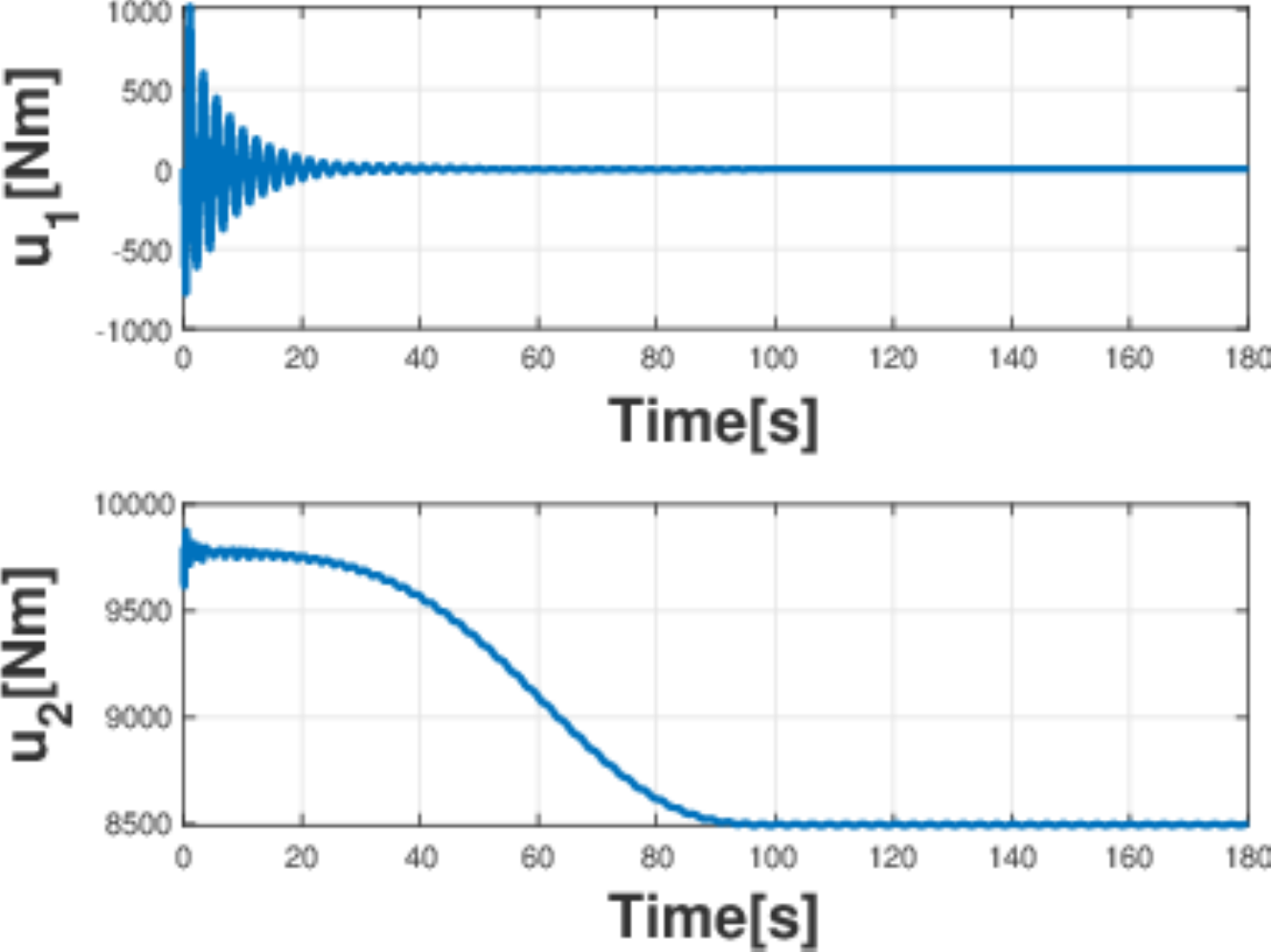}}\label{fig:u11}}\hspace{100pt}
\subfloat[Scenario 2. Control input $u_3, u_4$.]{%
\resizebox*{5cm}{!}{\includegraphics{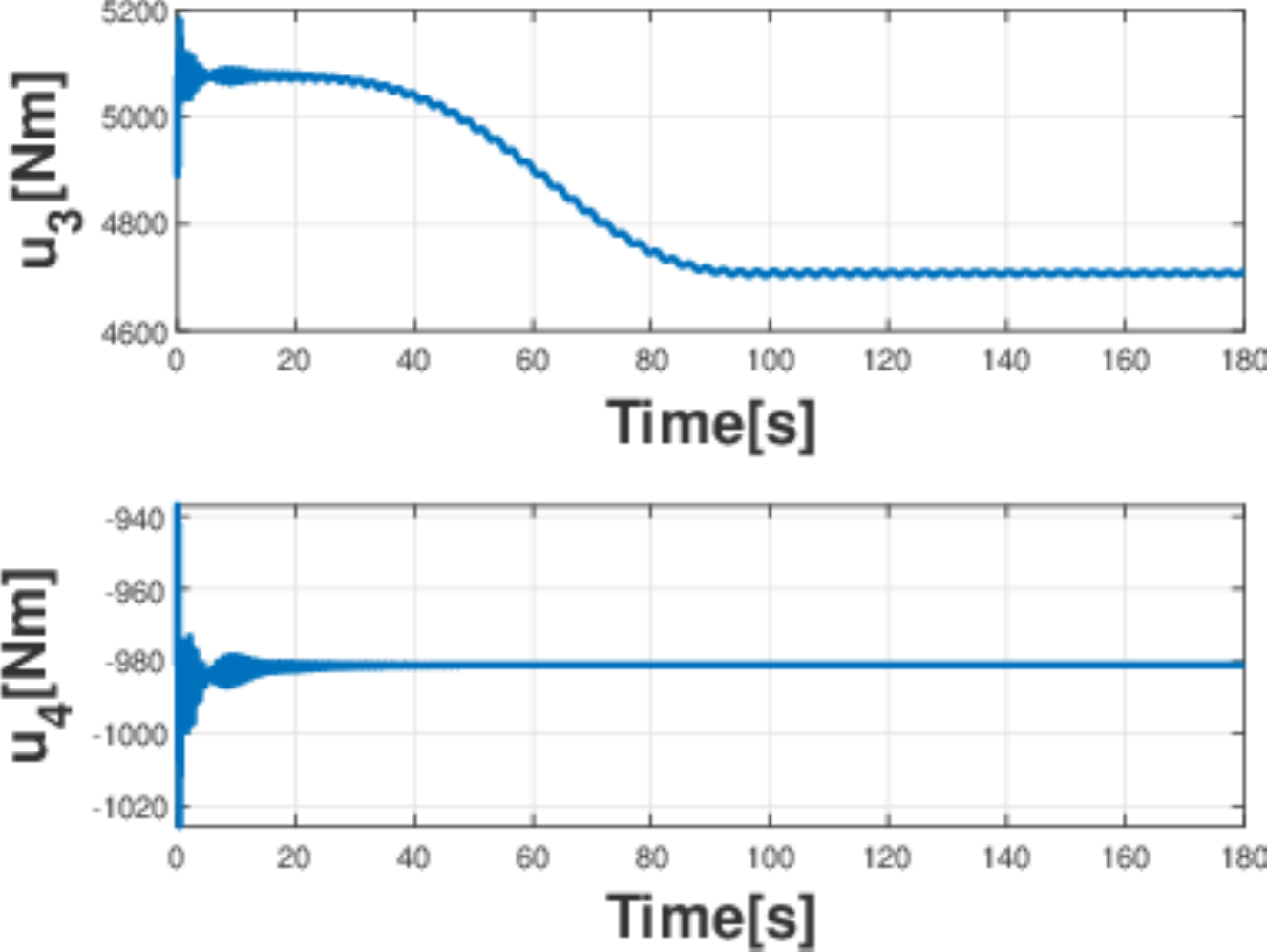}}\label{fig:u21}}
\caption{}
\end{figure}

\textit{Scenario 3.} In this simulation, to evaluate the robustness of our control scheme  w.r.t. non-perfect gravity compensation, we lift a payload with different mass than the one known by the control law. The payload mass has been set to 50 kg. The simulation results are reported in Figg.\ref{fig:al2}-\ref{fig:d2}, where we can see that the crane reaches the desired positions with a negligible error in terms of desired angles configurations and desired cable length. Additionally, the payload swing amplitudes (Figg.\ref{fig:th12}-\ref{fig:th22}) have negligible residual swings. As one can see from Figg.\ref{fig:u12}-\ref{fig:u22}, the input values, due to the different mass of the payload, have lower values than in the previous cases because they strictly depend on the payload mass value, especially for the terms relating to gravity compensation. As expected the input $u_4$ in Fig.\ref{fig:u21} is roughly doubled compared to the one in Fig.\ref{fig:u22}.

\medskip

\begin{figure}[ht!]
\centering

\subfloat[Scenario 3. Tower angle $\alpha$. Blue line: Nonlinear controller. Red line: Desired reference.]{%
\resizebox*{5cm}{!}{\includegraphics{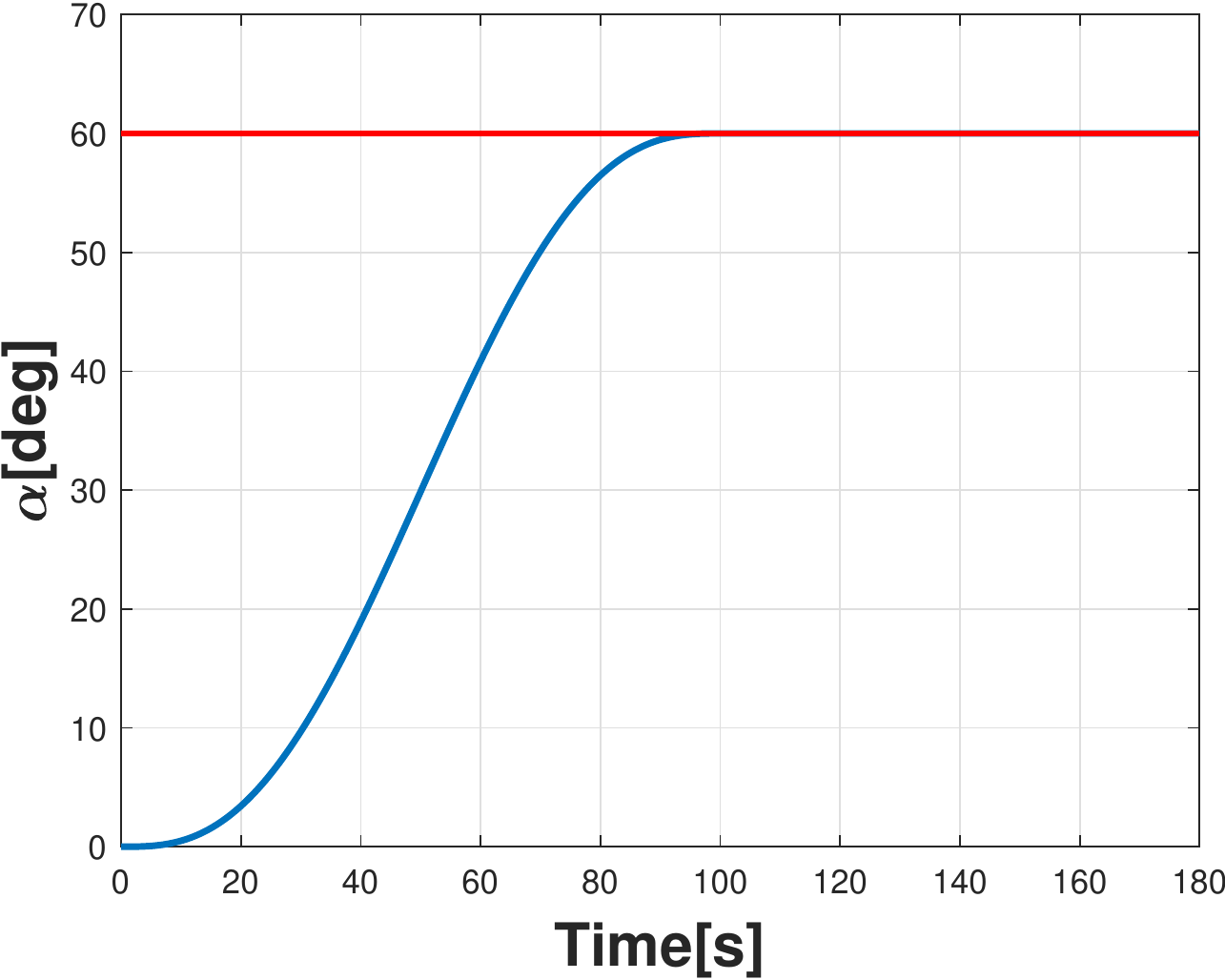}}\label{fig:al2}}\hspace{100pt}
\subfloat[Scenario 3. Boom angle $\beta$. Blue line: Nonlinear controller. Red line: Desired reference.]{%
\resizebox*{5cm}{!}{\includegraphics{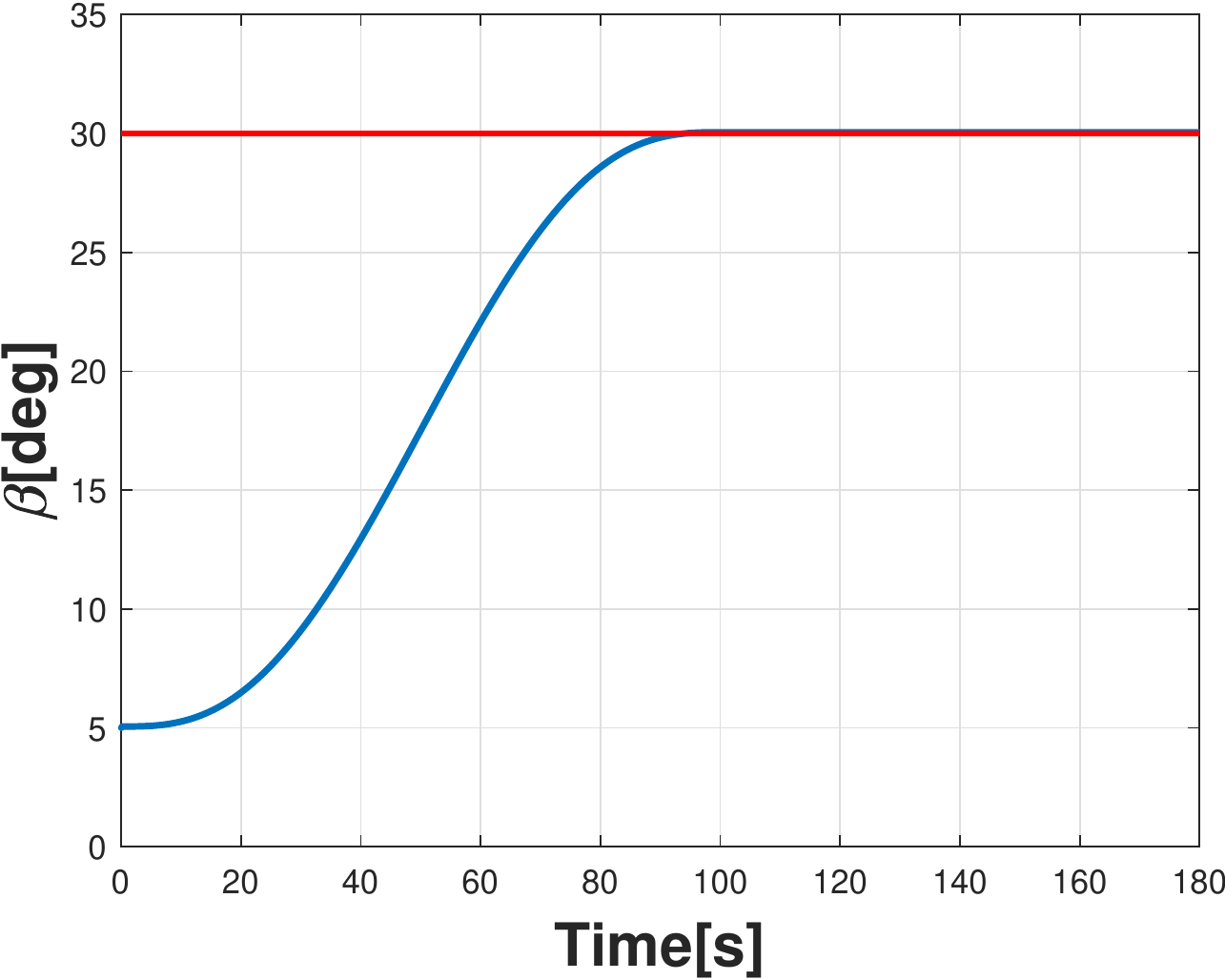}}\label{fig:bt2}}
\caption{}
\end{figure}

\begin{figure}[ht!]
\centering

\subfloat[Scenario 3. Jib angle $\gamma$. Blue line: Nonlinear controller. Red line: Desired reference.]{%
\resizebox*{5cm}{!}{\includegraphics{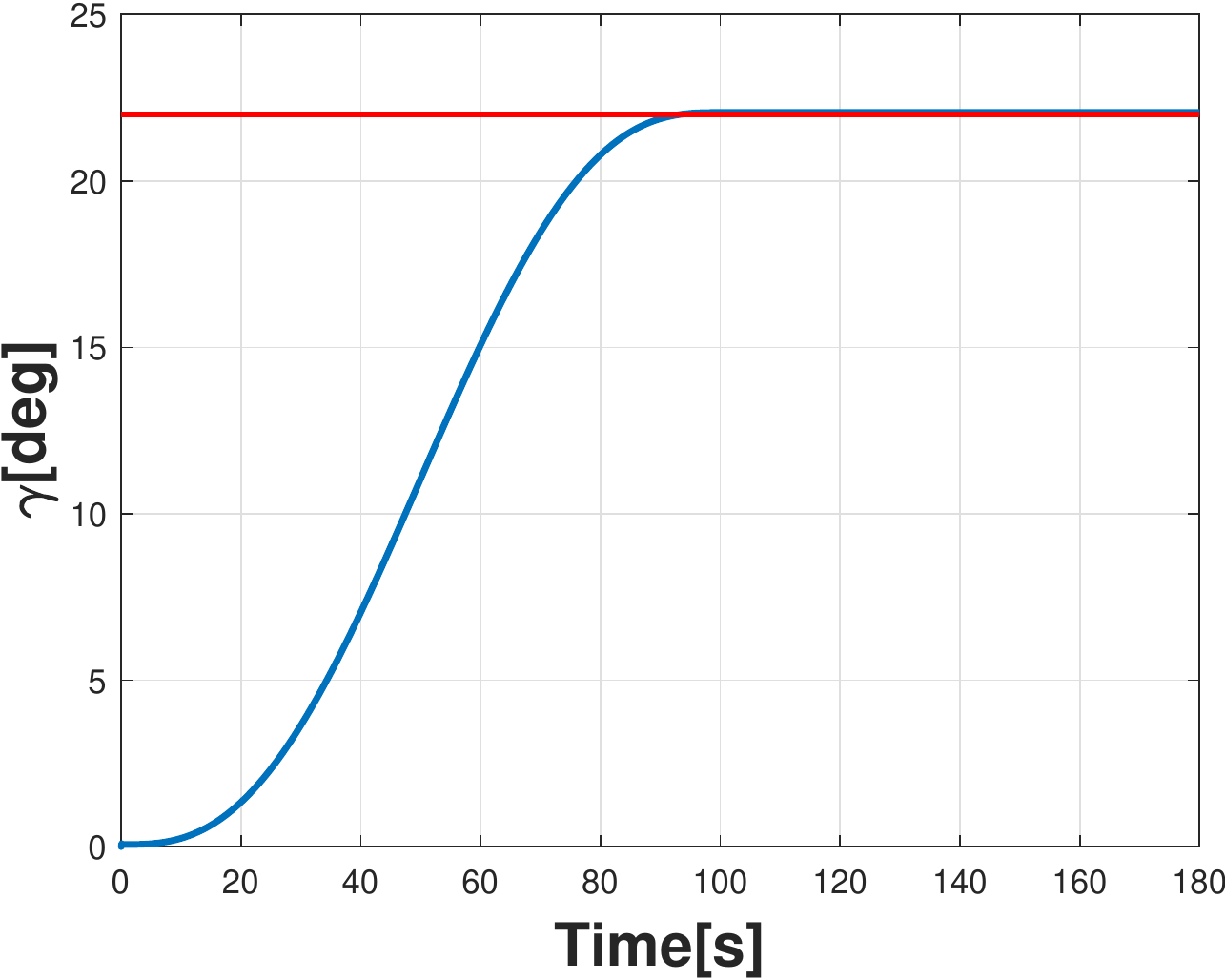}}\label{fig:gm2}}\hspace{100pt}
\subfloat[Scenario 3. Rope length $d$. Blue line: Nonlinear controller. Red line: Desired reference.]{%
\resizebox*{5cm}{!}{\includegraphics{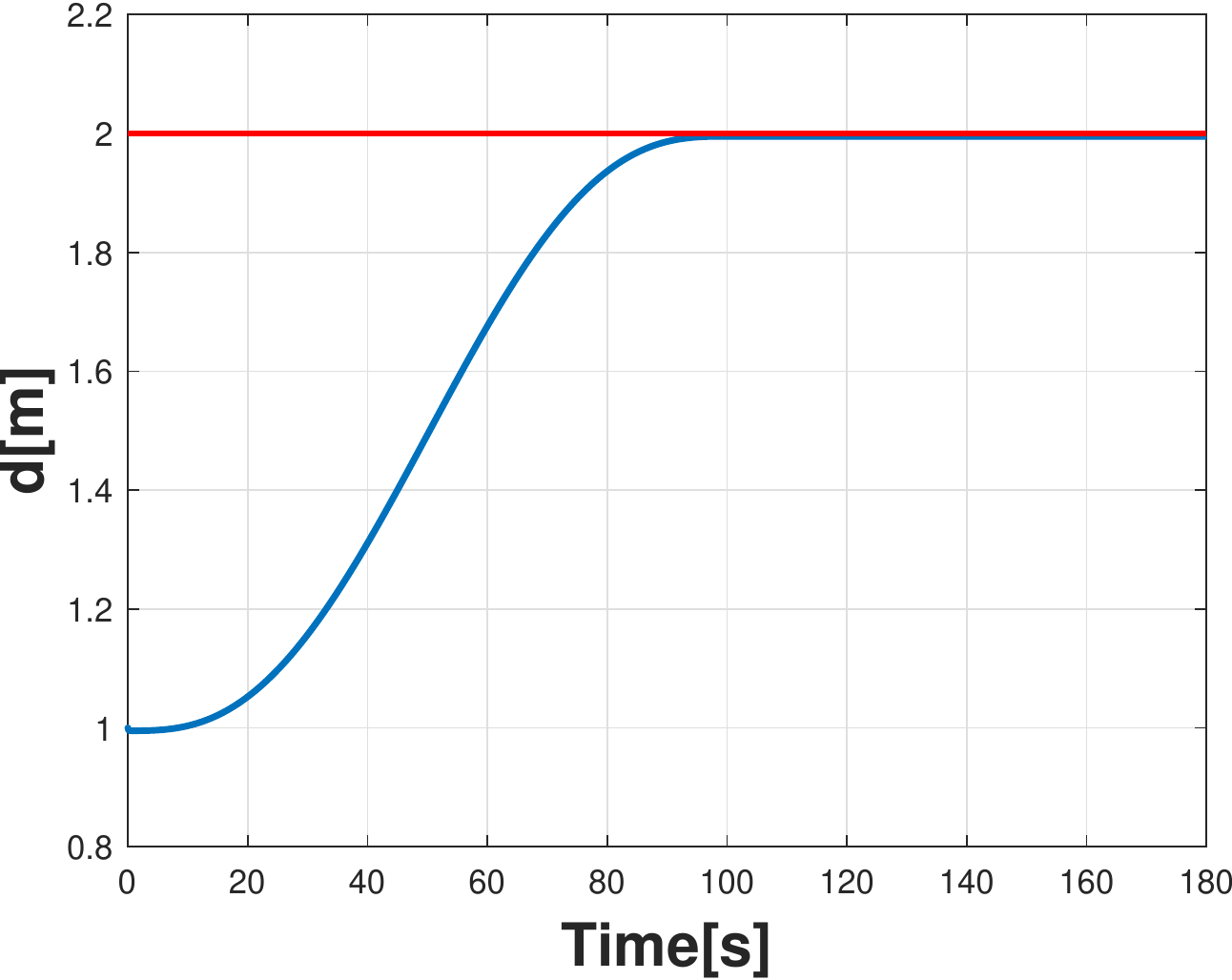}}\label{fig:d2}}
\caption{}
\end{figure}

\begin{figure}[ht!]
\centering

\subfloat[Scenario 3. Payload swing angle $\theta_1$.]{%
\resizebox*{5cm}{!}{\includegraphics{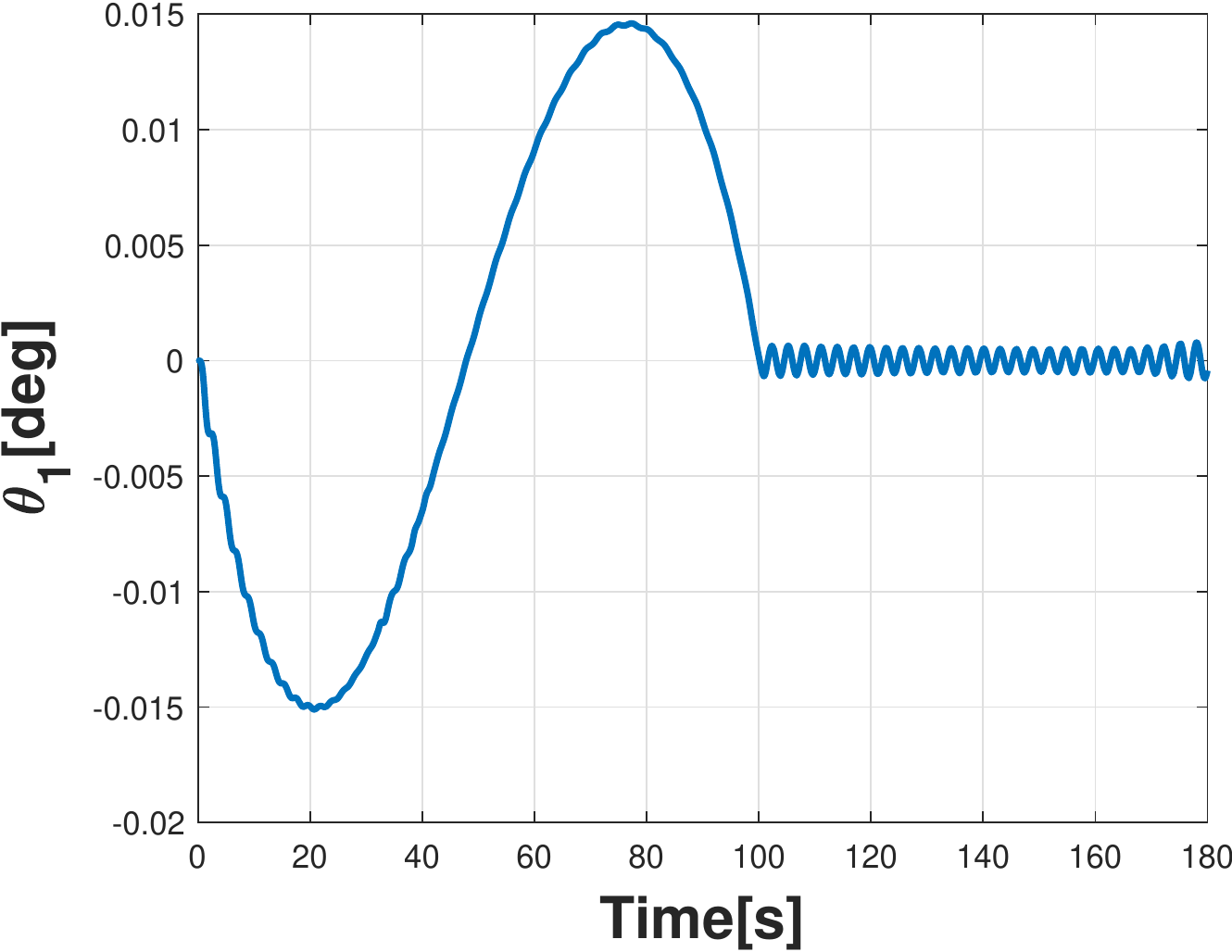} }\label{fig:th12}}\hspace{100pt}
\subfloat[Scenario 3. Payload swing angle $\theta_2$]{%
\resizebox*{5cm}{!}{\includegraphics{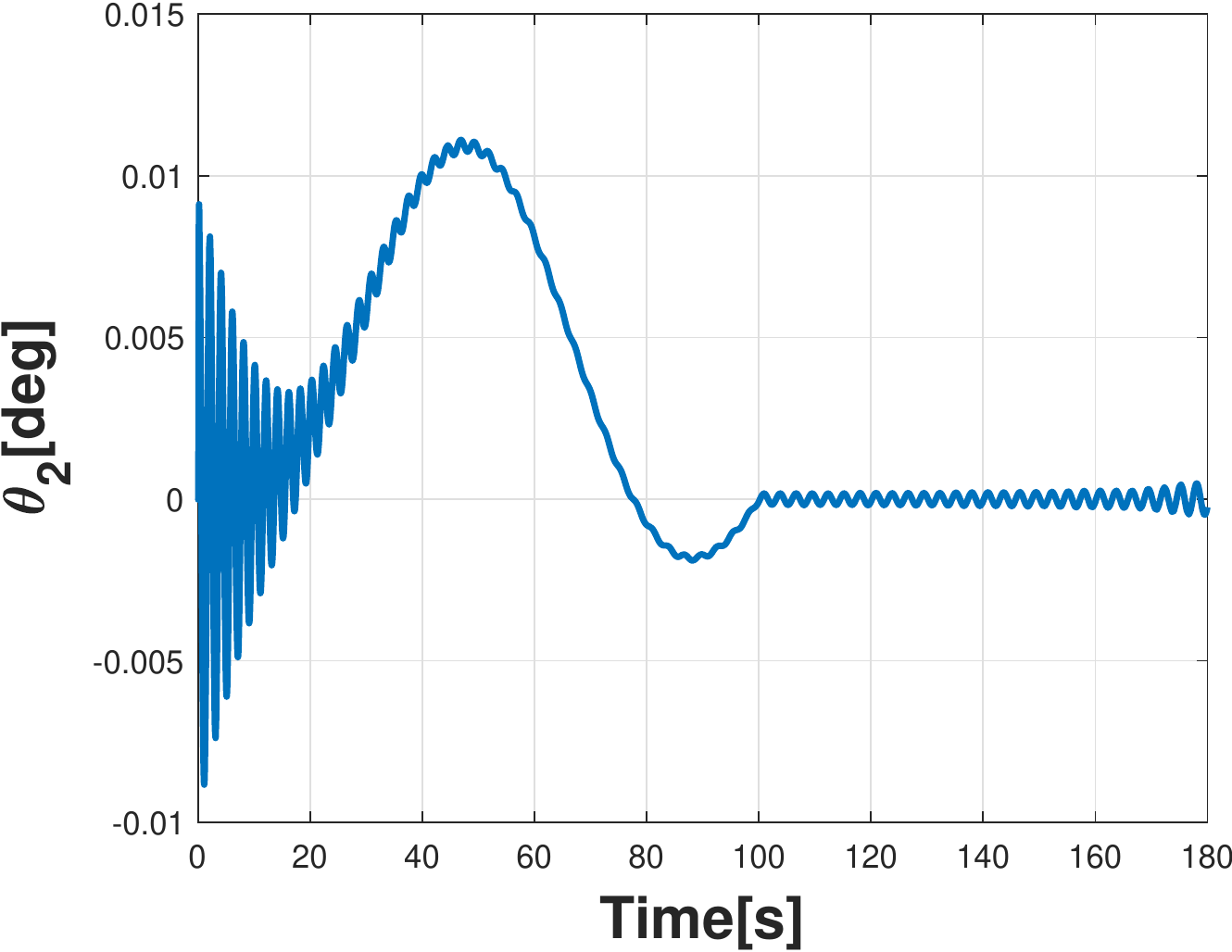}}\label{fig:th22}}
\caption{}
\end{figure}

\begin{figure}[ht!]
\centering
\subfloat[Scenario 3. Control input $u_1, u_2$.]{%
\resizebox*{5cm}{!}{\includegraphics{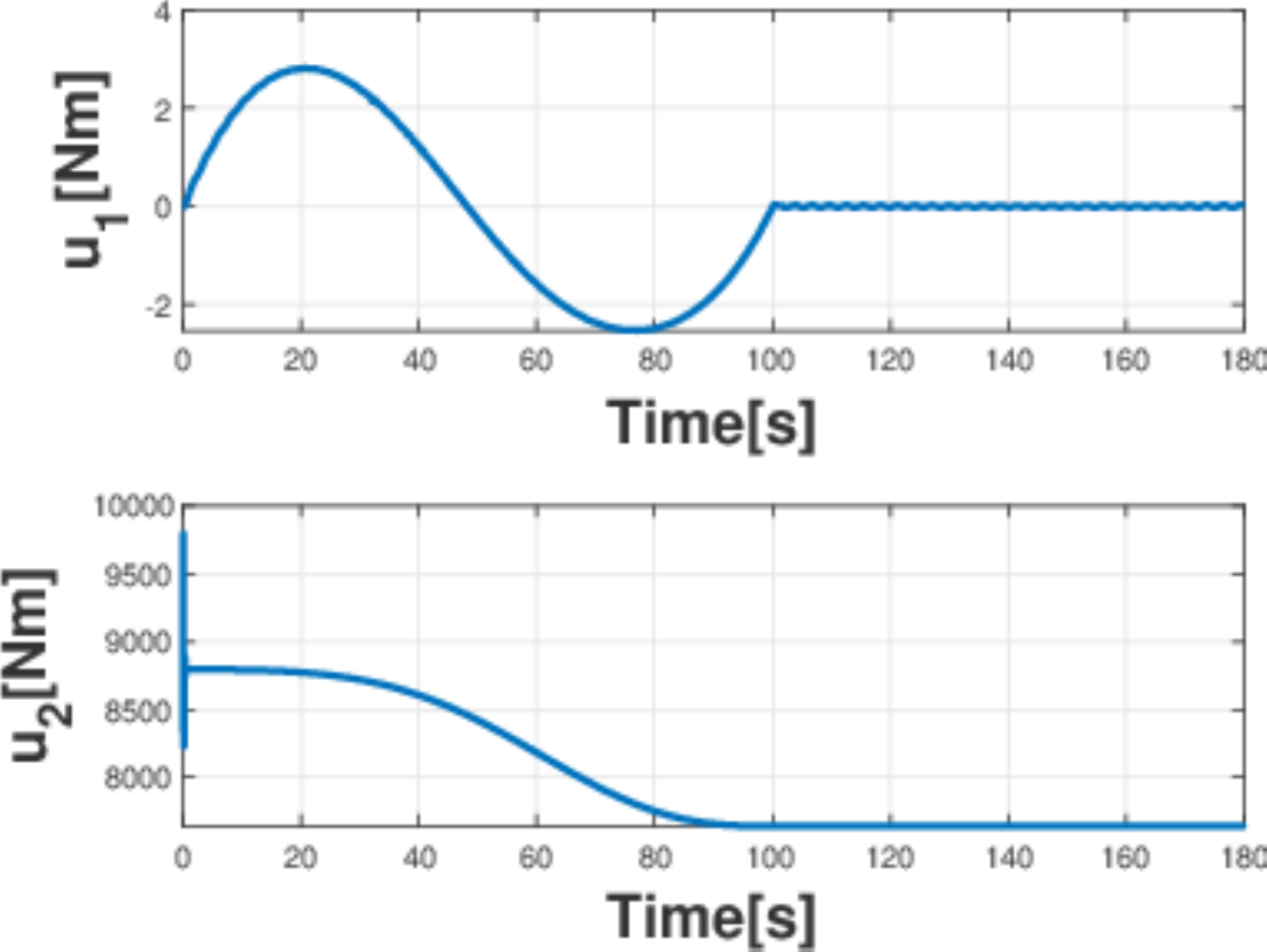}}\label{fig:u12}}\hspace{100pt}
\subfloat[Scenario 3. Control input $u_3, u_4$.]{%
\resizebox*{5cm}{!}{\includegraphics{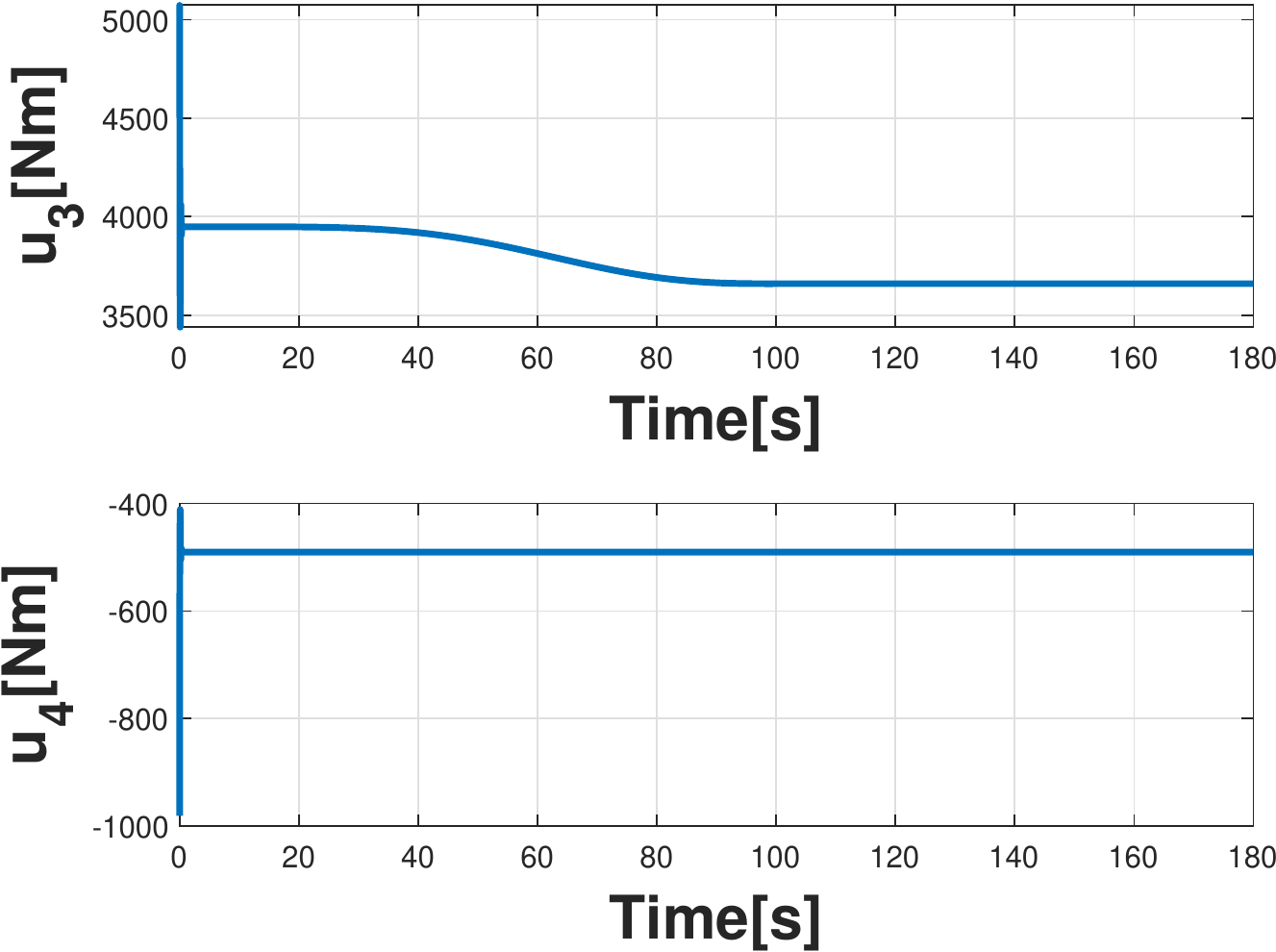}}\label{fig:u22}}
\caption{}
\end{figure}

\textit{Scenario 4.} In this simulation, we consider a gust of wind as external disturbance acting on the crane to demonstrate that the proposed control scheme is robust w.r.t an external unmodeled force. As one can see in Figg.~\ref{fig:th14}-\ref{fig:th24}, when the gust of wind occurs (at around 30 seconds), initially the swing angles $\theta_1$ and $\theta_2$ increase but then the residual oscillations of the payload swing angles $\theta_1$ and $\theta_2$ are confined within $1^{\circ}$ and $2^{\circ}$, respectively, in the time window considered in the simulations. As expected, the controller is able to counteract reasonably well the external disturbance, therefore no major change occurs during the desired trajectory (see Figg.\ref{fig:al4}-\ref{fig:d4}). 

\medskip

It is worth noting that the simulations shown in \textit{Scenario 3} and \textit{Scenario 4} allow us to state that even if the crane system model is not accurate, the proposed control law is able to effectively counteract this type of model mismatch.

\medskip

\begin{figure}[ht!]
\centering

\subfloat[Scenario 4. Tower angle $\alpha$. Blue line: Nonlinear controller. Red line: Desired reference.]{%
\resizebox*{5cm}{!}{\includegraphics{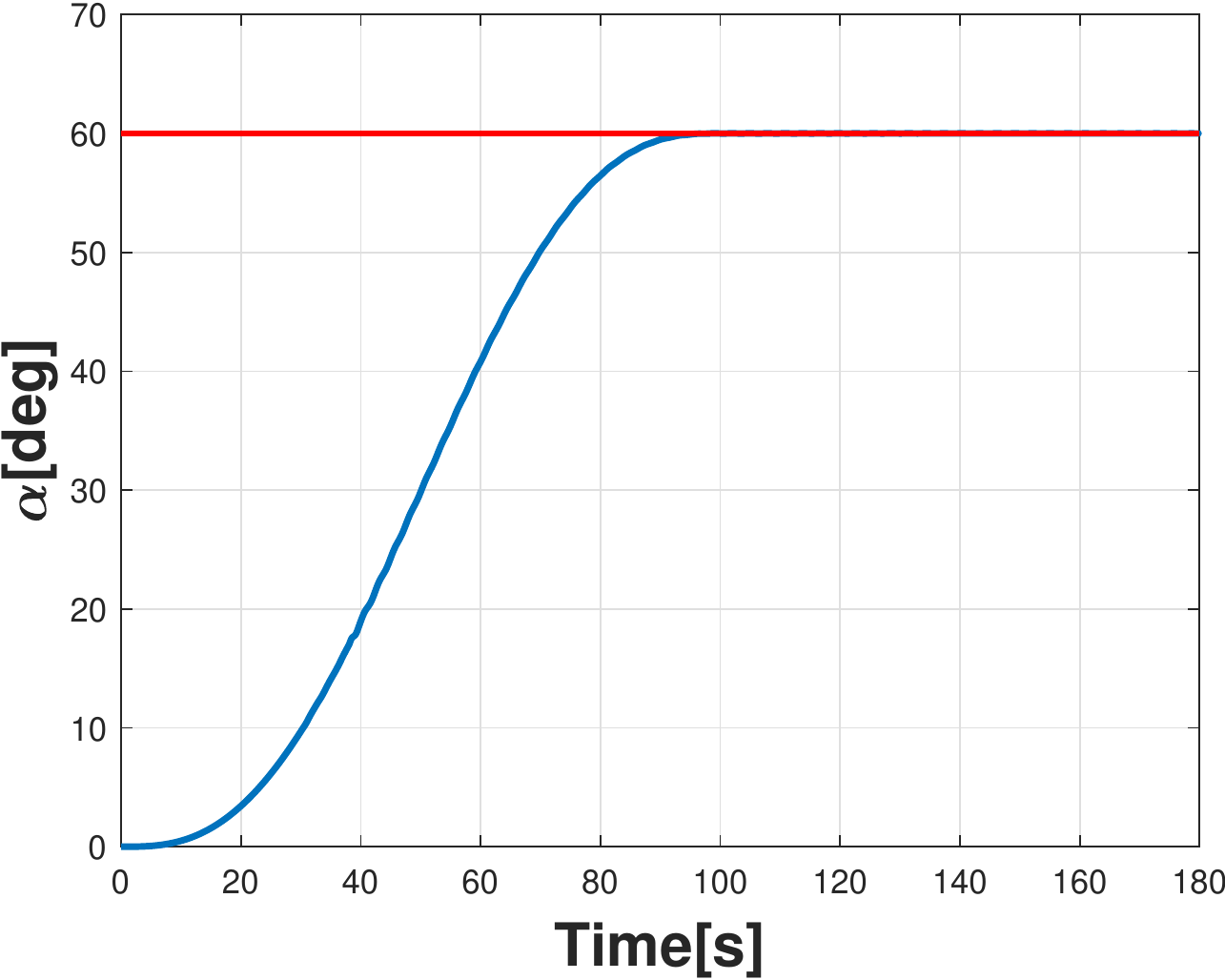}}\label{fig:al4}}\hspace{100pt}
\subfloat[Scenario 4. Boom angle $\beta$. Blue line: Nonlinear controller. Red line: Desired reference.]{%
\resizebox*{5cm}{!}{\includegraphics{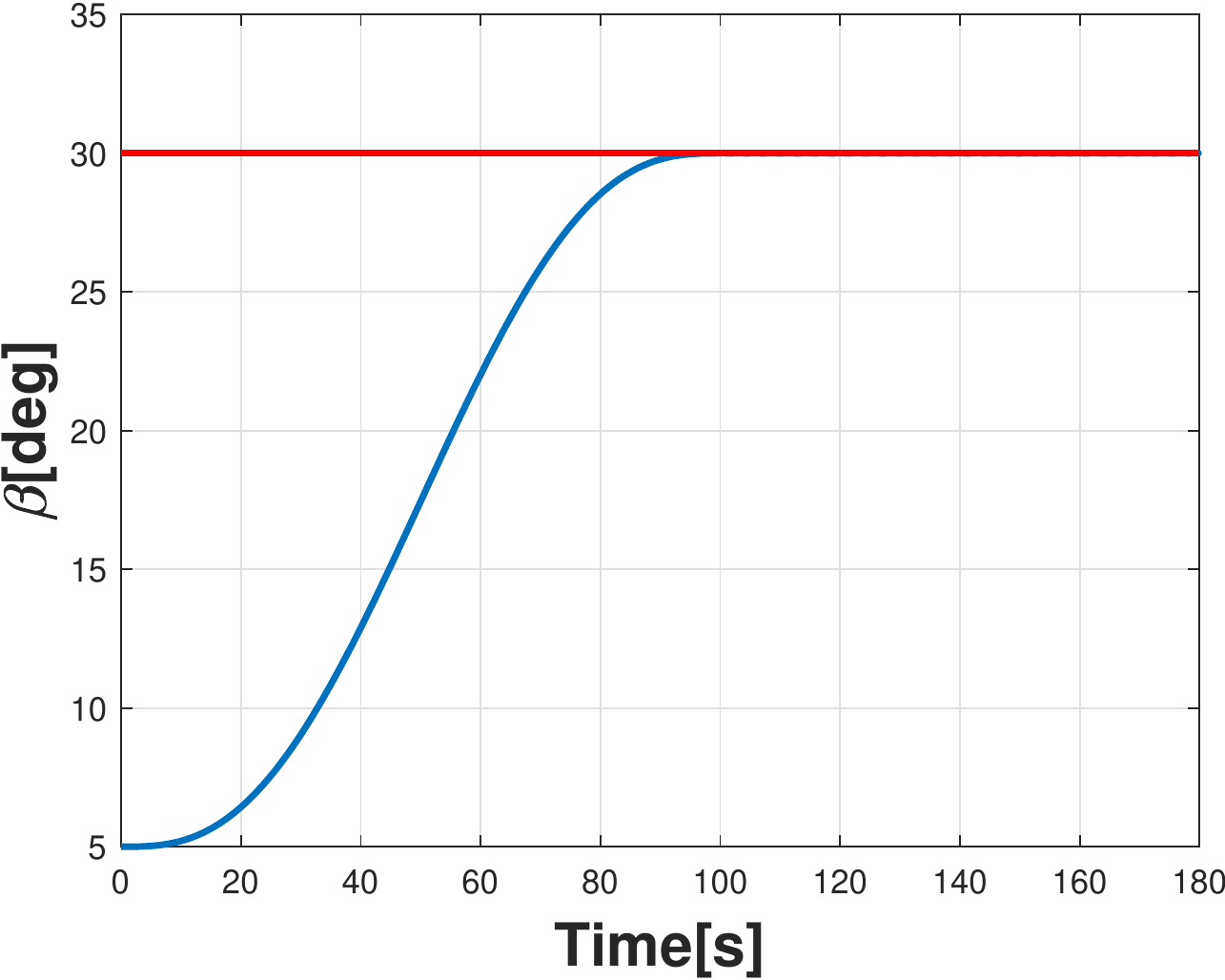}}\label{fig:bt4}}
\caption{}
\end{figure}

\begin{figure}[ht!]
\centering

\subfloat[Scenario 4. Jib angle $\gamma$. Blue line: Nonlinear controller. Red line: Desired reference.]{%
\resizebox*{5cm}{!}{\includegraphics{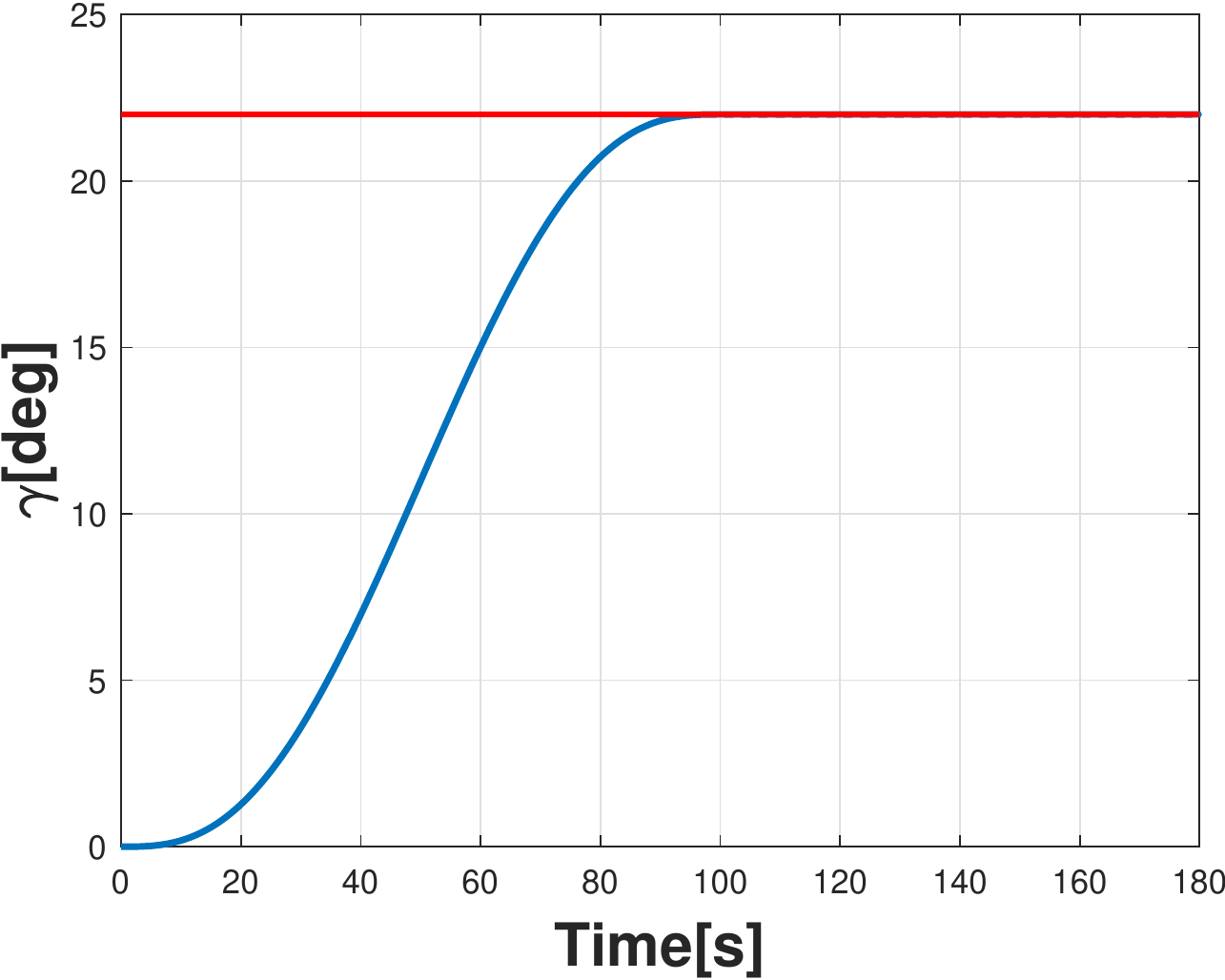}}\label{fig:gm4}}\hspace{100pt}
\subfloat[Scenario 4. Rope length $d$. Blue line: Nonlinear controller. Red line: Desired reference.]{%
\resizebox*{5cm}{!}{\includegraphics{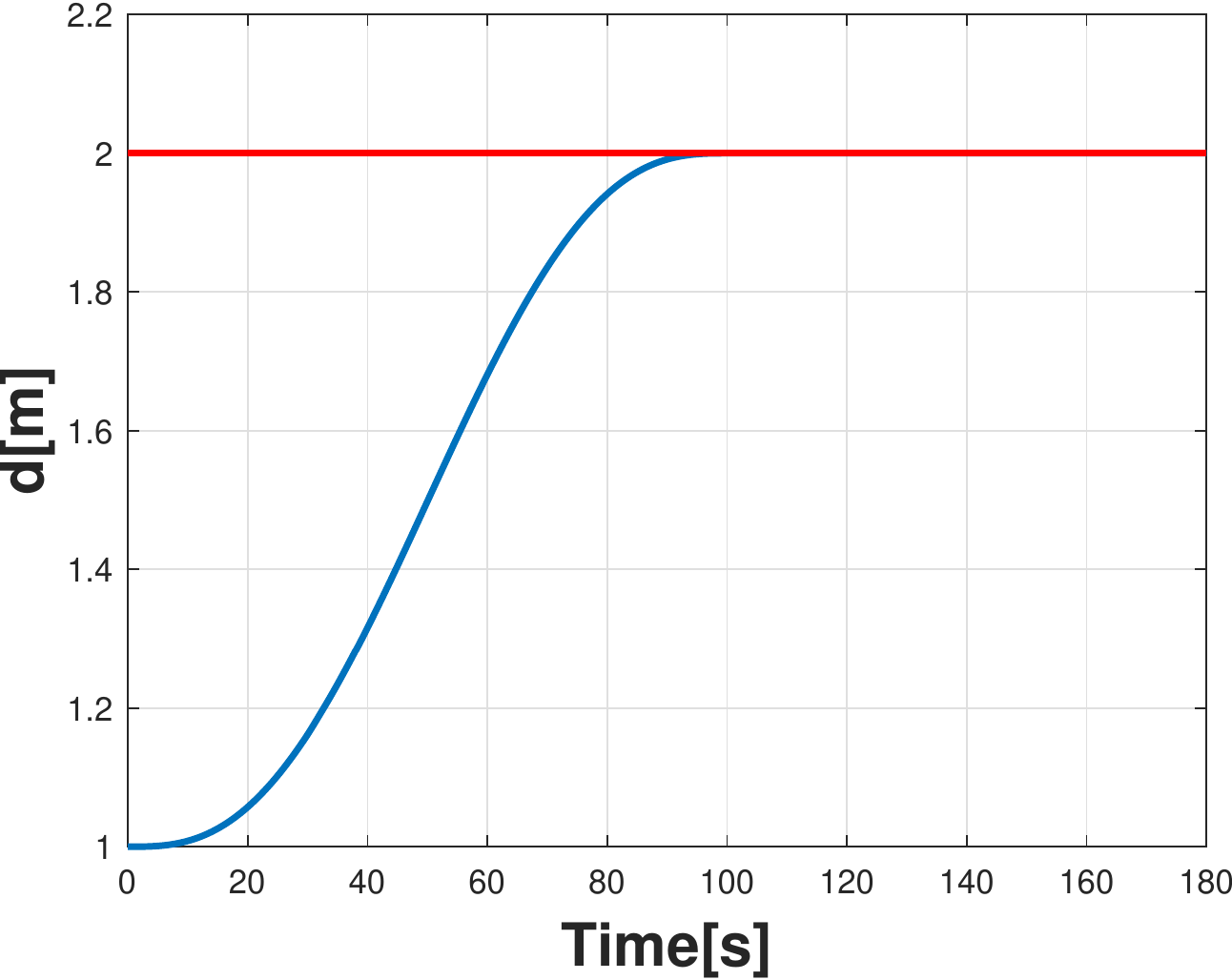}}\label{fig:d4}}
\caption{}
\end{figure}

\begin{figure}[ht!]
\centering

\subfloat[Scenario 4. Payload swing angle $\theta_1$.]{%
\resizebox*{5cm}{!}{\includegraphics{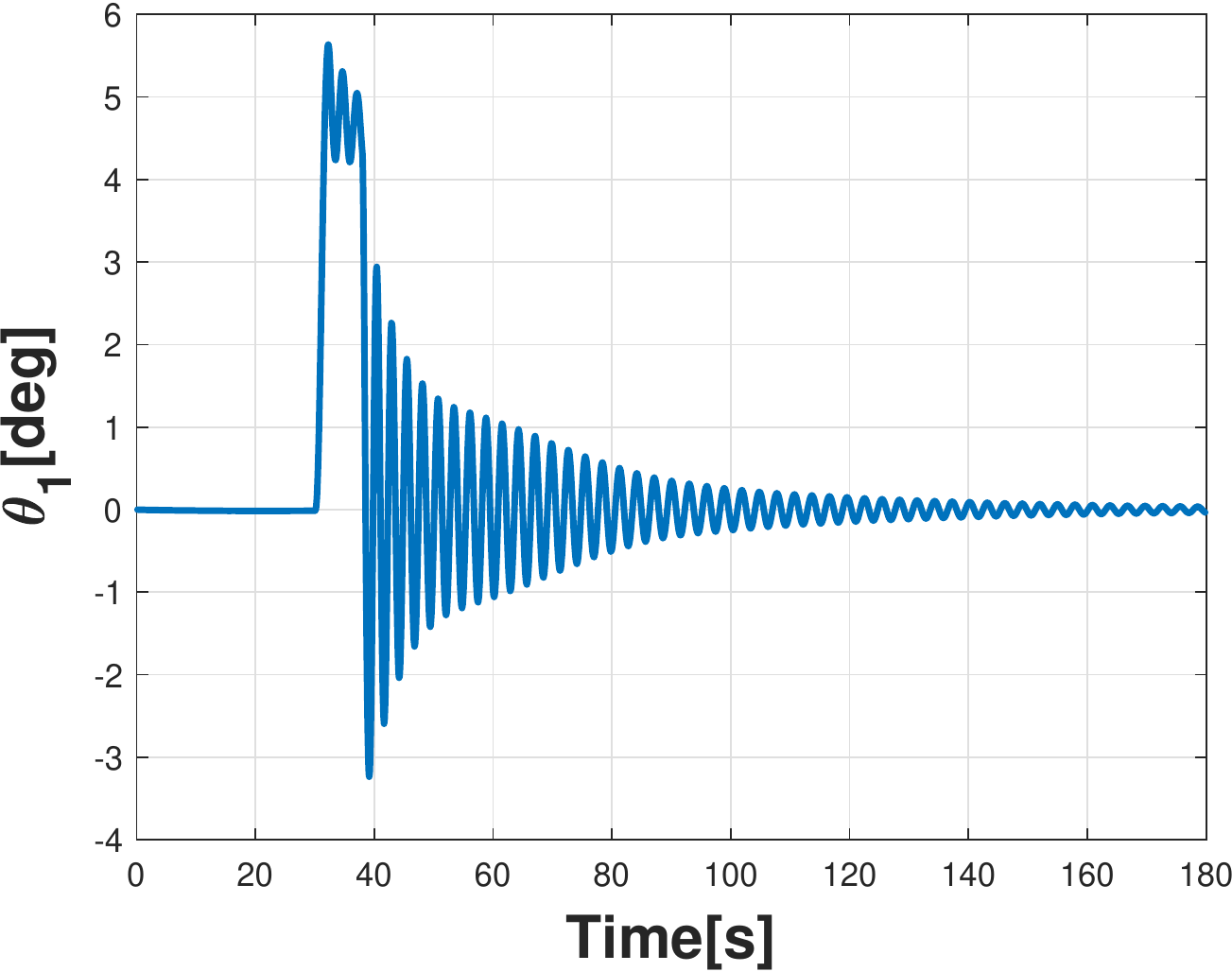} }\label{fig:th14}}\hspace{100pt}
\subfloat[Scenario 4. Payload swing angle $\theta_2$]{%
\resizebox*{5cm}{!}{\includegraphics{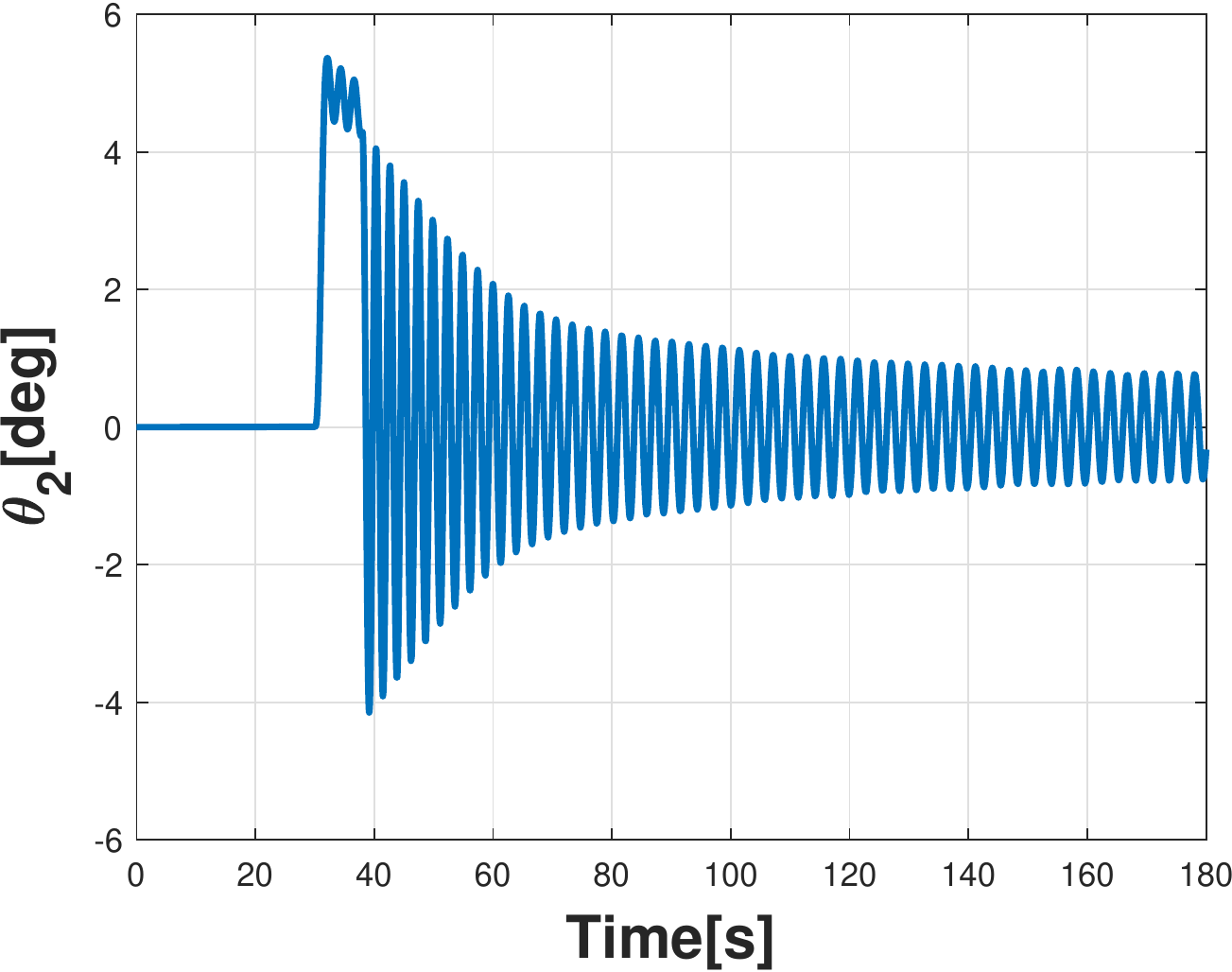}}\label{fig:th24}}
\caption{}
\end{figure}

\begin{figure}[ht!]
\centering
\subfloat[Scenario 4. Control input $u_1, u_2$.]{%
\resizebox*{5cm}{!}{\includegraphics{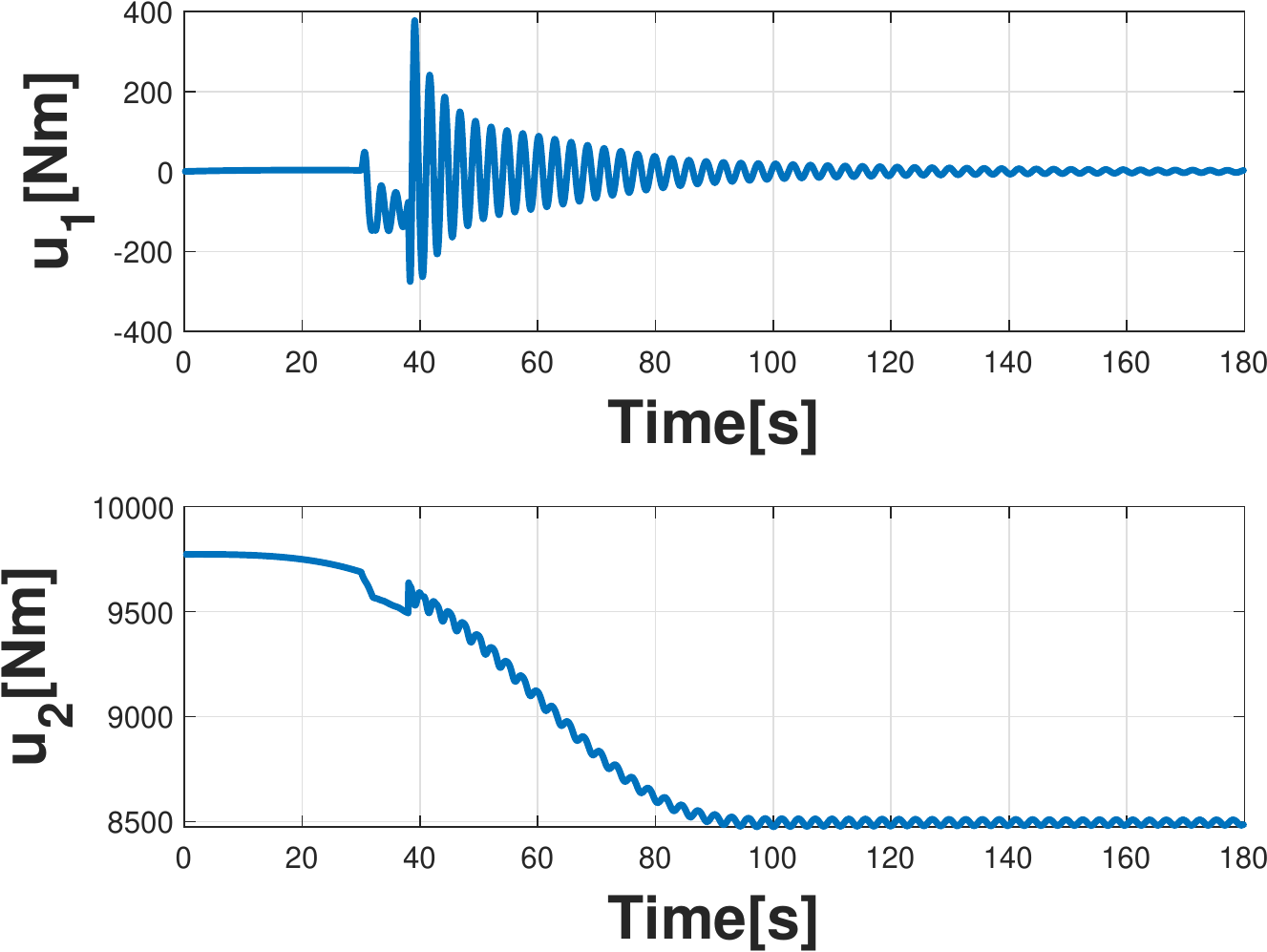}}\label{fig:u14}}\hspace{100pt}
\subfloat[Scenario 4. Control input $u_3, u_4$.]{%
\resizebox*{5cm}{!}{\includegraphics{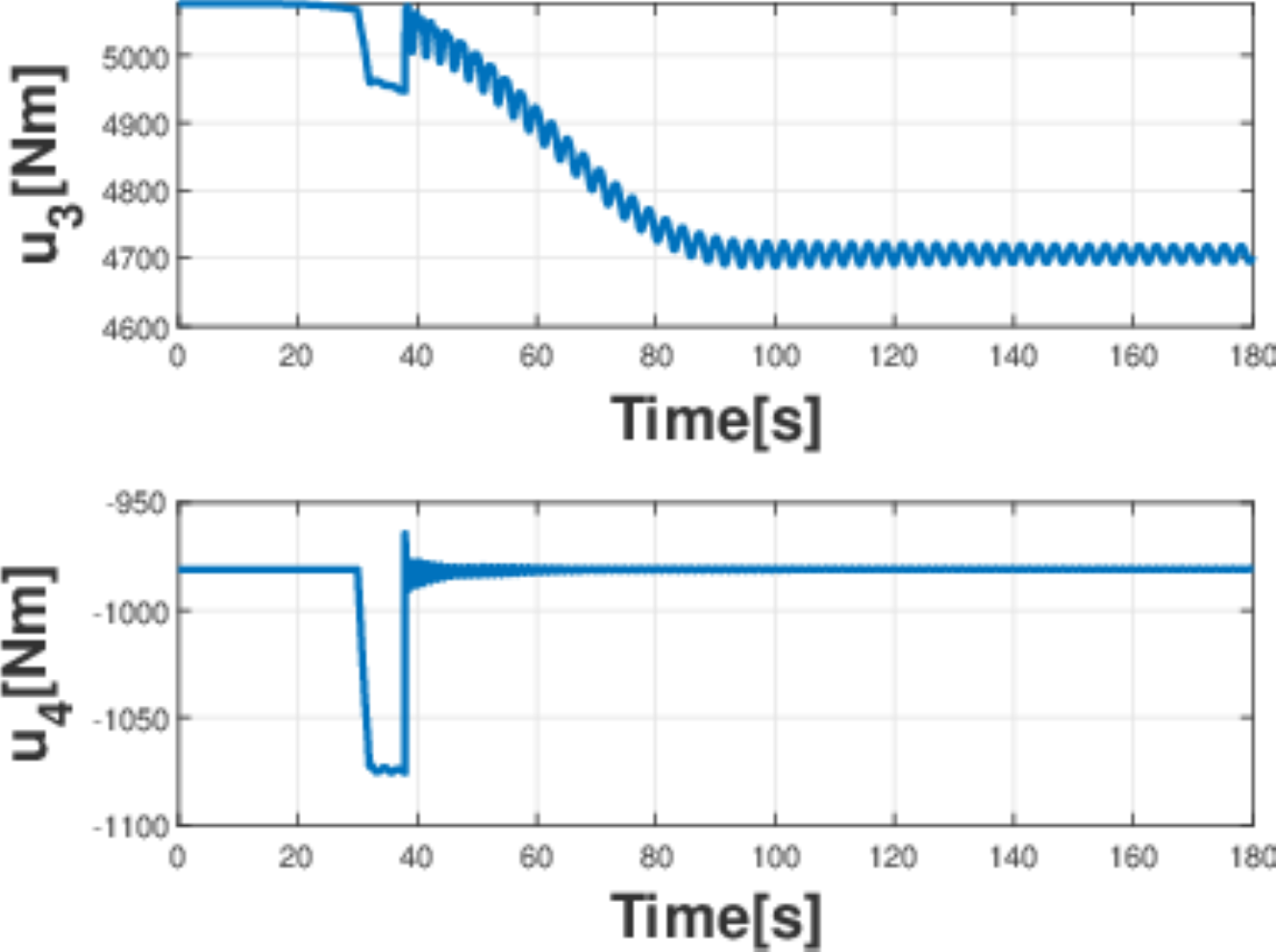}}\label{fig:u24}}
\caption{}
\end{figure}

\textit{Scenario 5.} In this simulation scenario we will consider the presence of measurement noise for the angular position measurements used by the control law. The simulation results are given in Figg.\ref{fig:al5}-\ref{fig:d5} where we can clearly see that the presence of noise has a very limited impact on the performance of the control law. As expected, the three actuated angles (\textit{i.e.} $\alpha,\beta,$ and $\gamma$) reach  the desired angular values in around 100 seconds. Additionally, the cable achieves the desired length. Moreover, at the beginning the payload swing amplitudes (\textit{i.e.} $\theta_1$ and $\theta_2$) present oscillations due to the presence of noise in the measurements.


\medskip

\begin{figure}[ht!]
\centering

\subfloat[Scenario 5. Tower angle $\alpha$. Blue line: Nonlinear controller. Red line: Desired reference.]{%
\resizebox*{5cm}{!}{\includegraphics{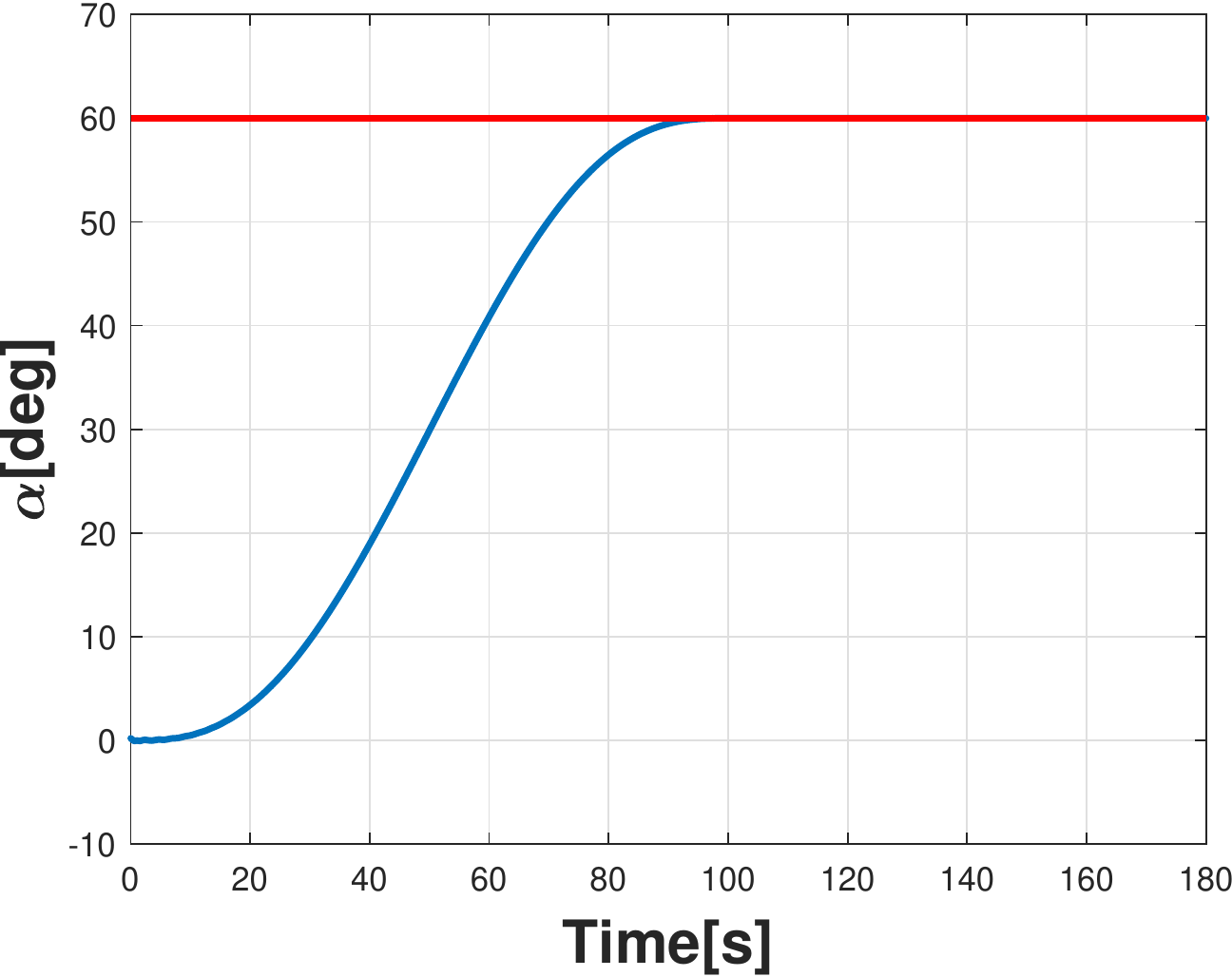}}\label{fig:al5}}\hspace{100pt}
\subfloat[Scenario 5. Boom angle $\beta$. Blue line: Nonlinear controller. Red line: Desired reference.]{%
\resizebox*{5cm}{!}{\includegraphics{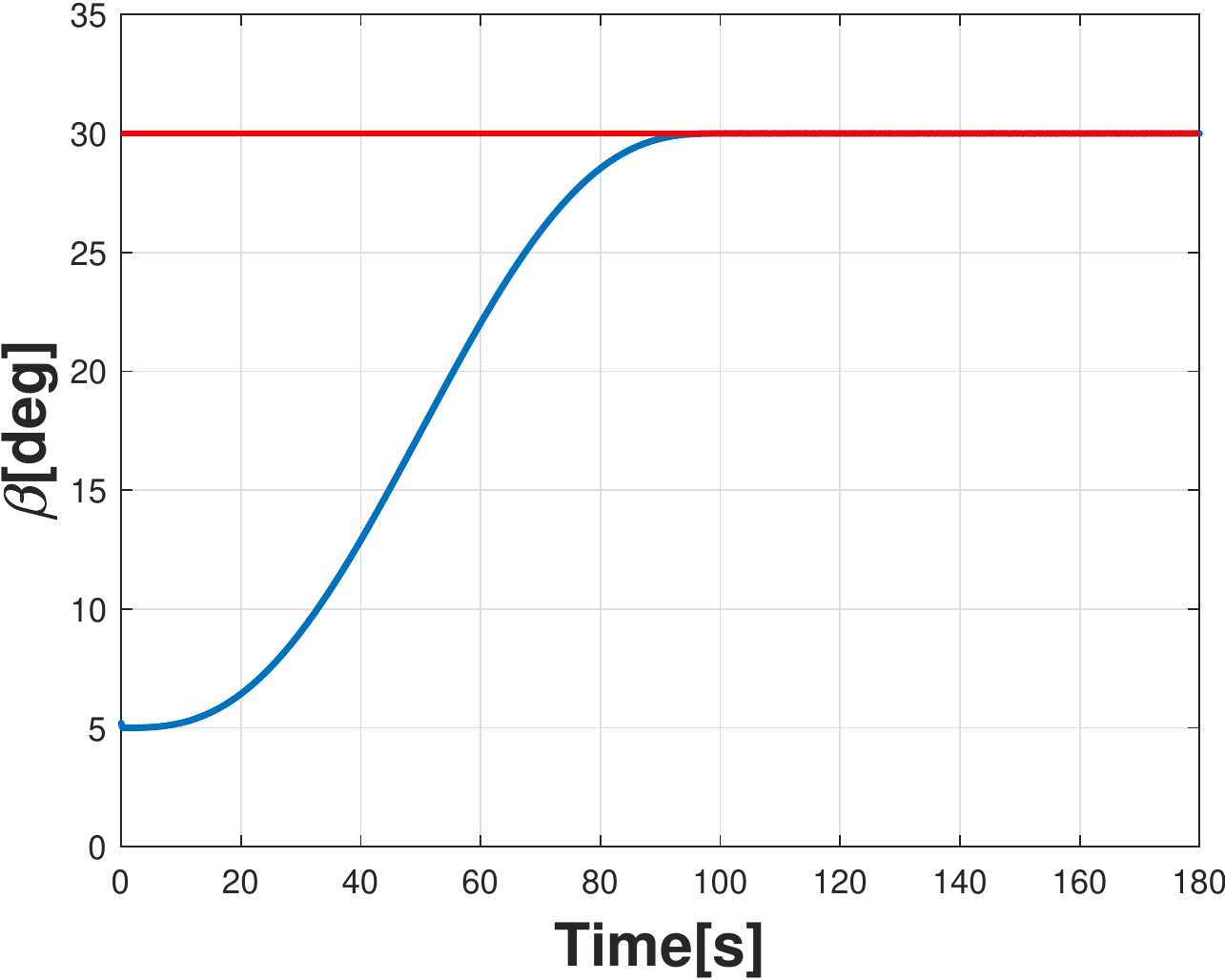}}\label{fig:bt5}}
\caption{}
\end{figure}

\begin{figure}[ht!]
\centering

\subfloat[Scenario 5. Jib angle $\gamma$. Blue line: Nonlinear controller. Red line: Desired reference.]{%
\resizebox*{5cm}{!}{\includegraphics{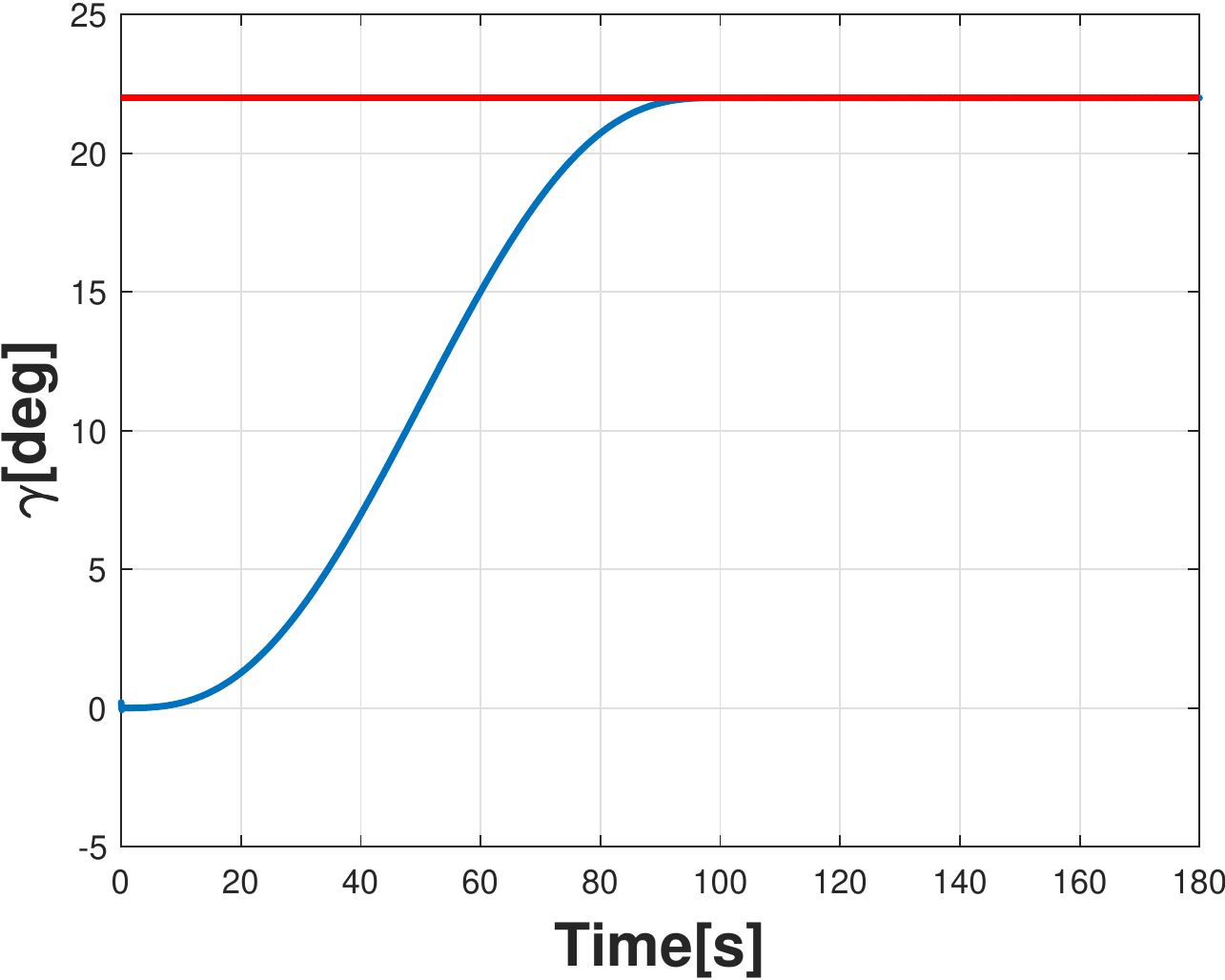}}\label{fig:gm5}}\hspace{100pt}
\subfloat[Scenario 5. Rope length $d$. Blue line: Nonlinear controller. Red line: Desired reference.]{%
\resizebox*{5cm}{!}{\includegraphics{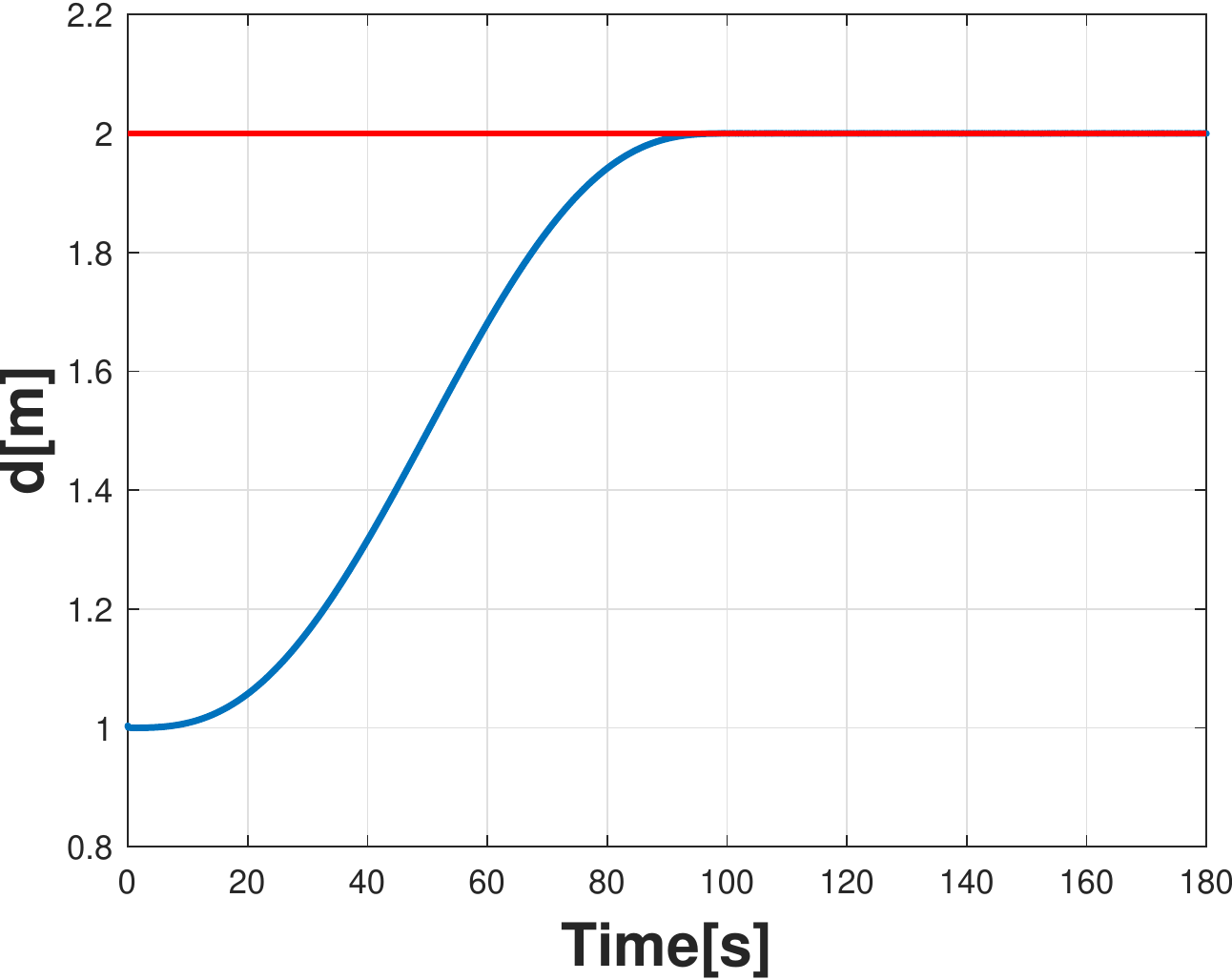}}\label{fig:d5}}
\caption{}
\end{figure}

\begin{figure}[ht!]
\centering

\subfloat[Scenario 5. Payload swing angle $\theta_1$.]{%
\resizebox*{5cm}{!}{\includegraphics{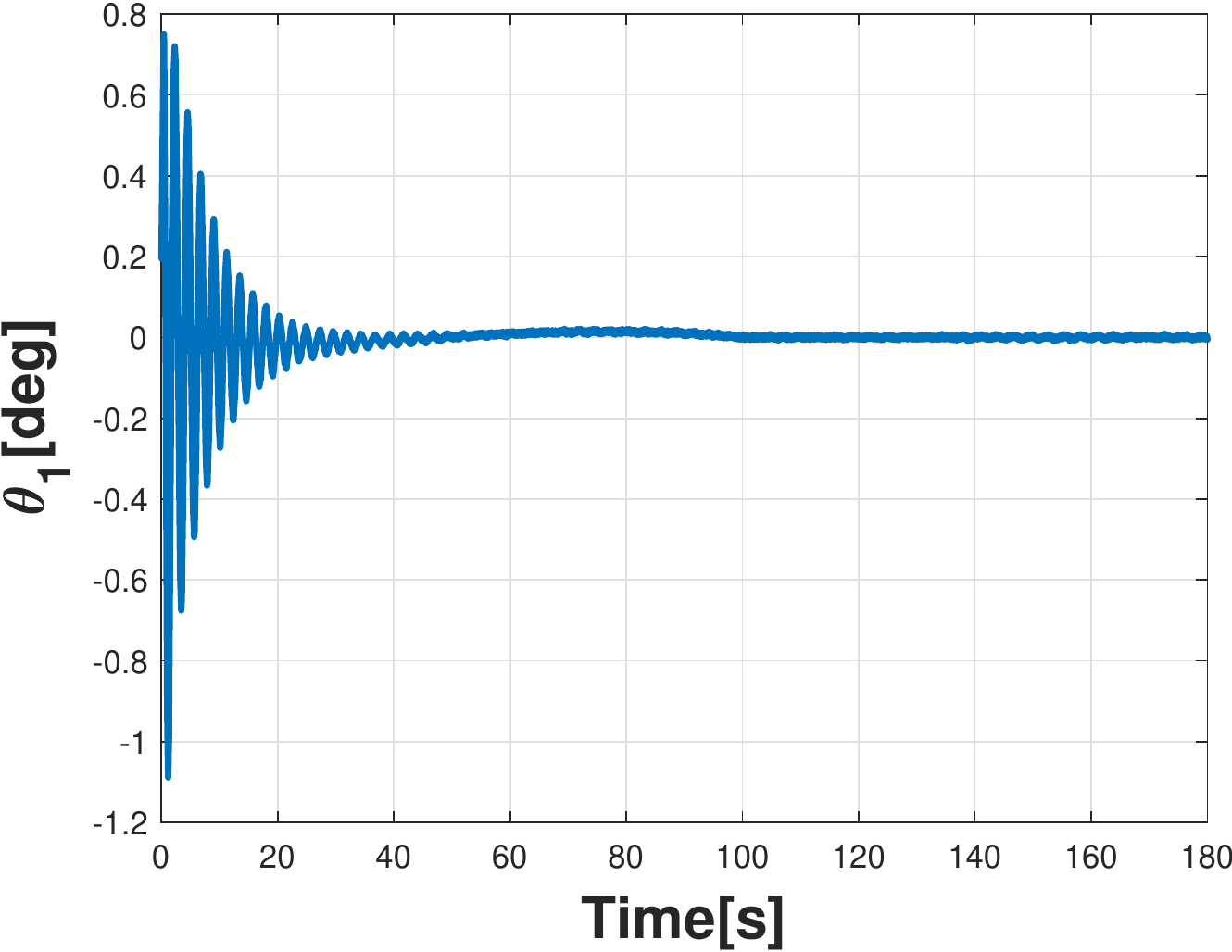} }\label{fig:th15}}\hspace{100pt}
\subfloat[Scenario 5. Payload swing angle $\theta_2$]{%
\resizebox*{5cm}{!}{\includegraphics{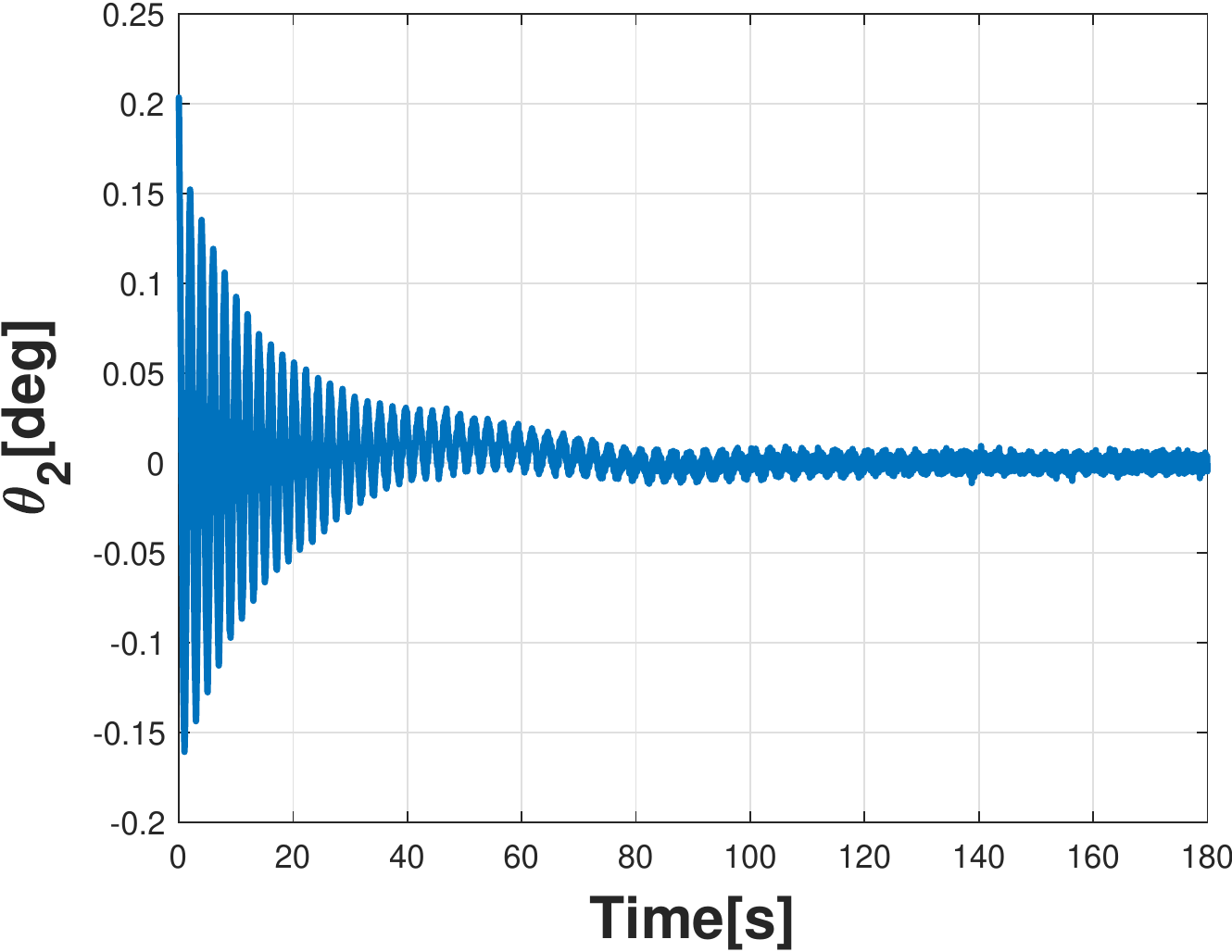}}\label{fig:th25}}
\caption{}
\end{figure}

\begin{figure}[ht!]
\centering
\subfloat[Scenario 5. Control input $u_1, u_2$.]{%
\resizebox*{5cm}{!}{\includegraphics{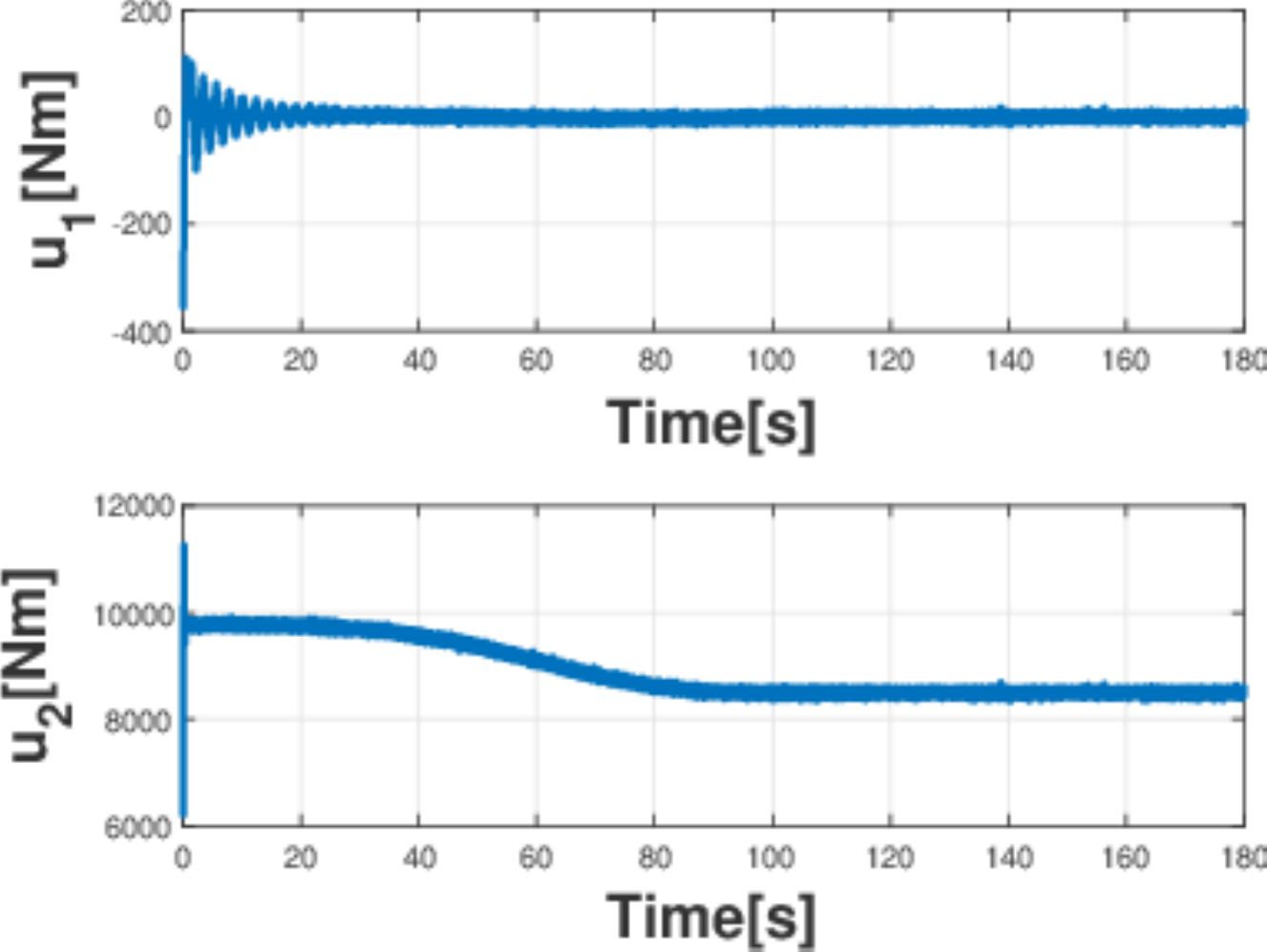}}\label{fig:u15}}\hspace{100pt}
\subfloat[Scenario 5. Control input $u_3, u_4$.]{%
\resizebox*{5cm}{!}{\includegraphics{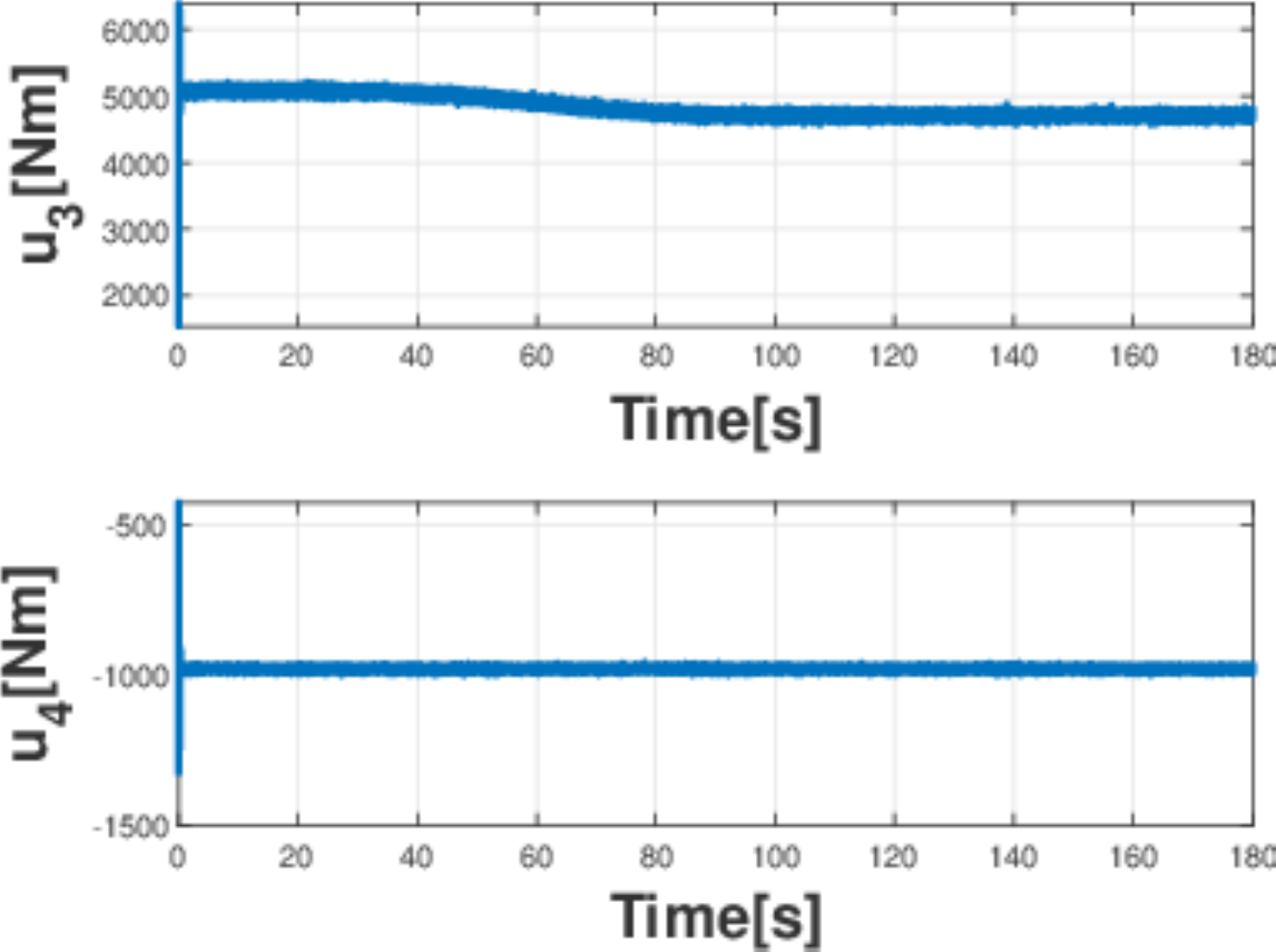}}\label{fig:u25}}
\caption{}
\end{figure}

\textit{Comparison analysis}. To show the efficiency of the proposed nonlinear control scheme, we have performed a comparative analysis considering a linear quadratic regulator (LQR). For the LQR control approach, first, the crane dynamics is linearized around the equilibrium point, and then, the following weight matrices have been chosen to stabilize the plant: $Q = diag\{10^2, 10^3, 10^3, 5^2, 10, 10, 10, 10, 10, 50, 10^2, 10^2\}$ and $R = diag\{50, 50, 50, 10\}$. The weight matrices have been chosen so that in nominal conditions the two controllers (\textit{i.e.} the proposed control scheme and the LQR controller) give similar dynamics during the crane movement (see Fig.\ref{fig:qcomp}) and similar input values for the each joint torques (see Fig.\ref{fig:ucomp}). Although in nominal circumstances, the LQR gives good results, in case of model mismatch, this control approach tends to not behave properly. Repeating the simulations carried out in \textit{Scenario 3}, in Fig.\ref{fig:qcomp2} it is possible to notice that the LQR is not able to follow the desired trajectory when the mass of the payload is different from the nominal one. 

\begin{figure}[ht!]
\centering
\includegraphics[width=0.8\linewidth]{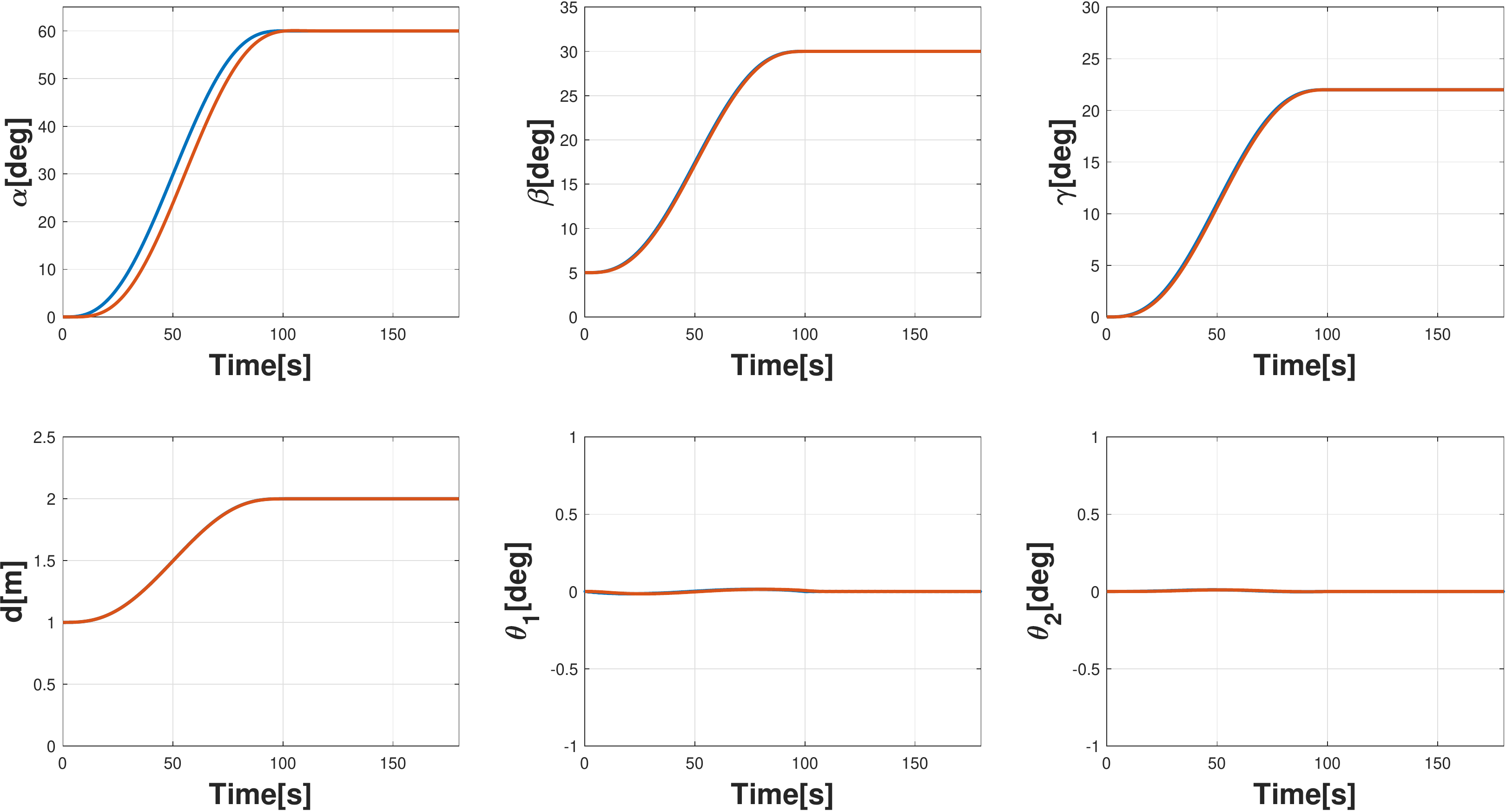}
\caption{\label{fig:qcomp}Blue line: Nonlinear controller. Red line: LQR.}
\end{figure}

\begin{figure}[ht!]
\centering
\includegraphics[width=0.6\linewidth]{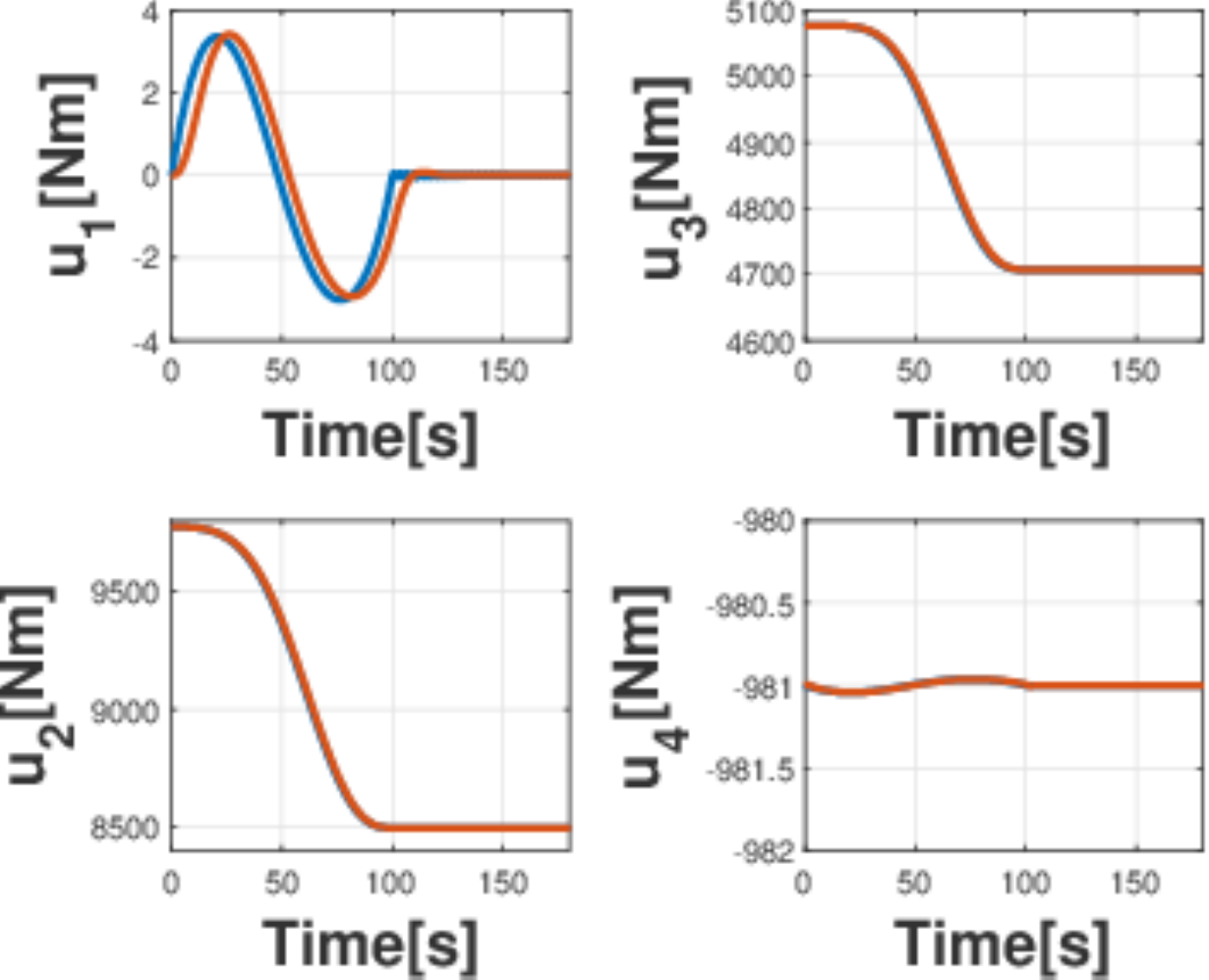}
\caption{\label{fig:ucomp}Blue line: Nonlinear controller. Red line: LQR.}
\end{figure}

\begin{figure}[ht!]
\centering
\includegraphics[width=0.6\linewidth]{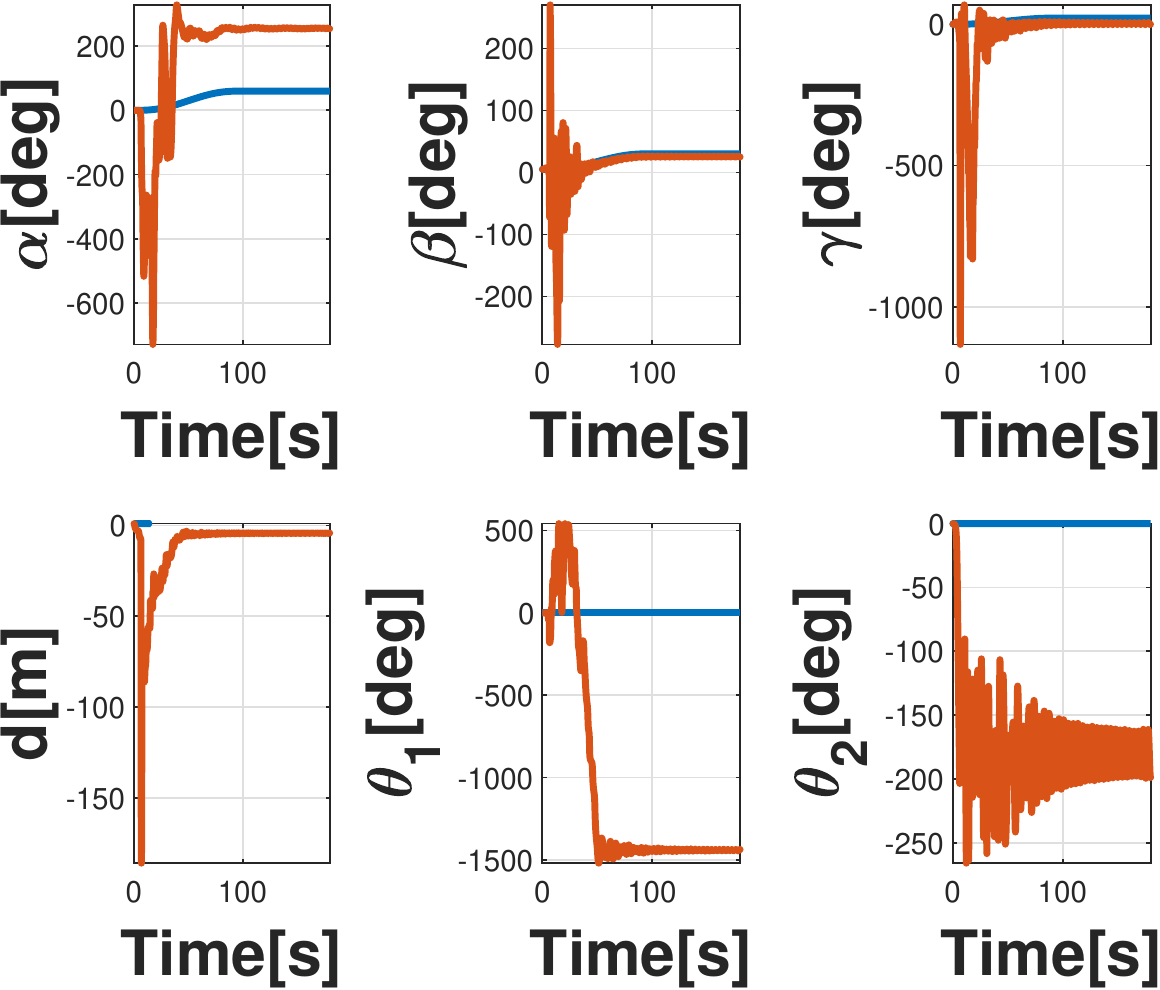}
\caption{\label{fig:qcomp2}Blue line: Nonlinear controller. Red line: LQR.}
\end{figure}

\section{Conclusion}\label{sec:concl}

The paper proposed for the first time a detailed mathematical model of a knuckle crane which takes into account all of the degrees of freedom (DoFs) that characterize this type of system (\textit{i.e.} the three rotations, the length of the rope and the payload swing angles).
On the basis of this model, a nonlinear control law has been developed which is able to perform position control of the crane while actively damping the oscillations of the load.  It is worth noting that, the proposed control law is based directly on the nonlinear model of the knuckle crane avoiding any form of linearization of the mathematical model. The stability properties of the scheme have been proved using the LaSalle's invariance principle. Five different simulation scenarios and a comparison analysis with a LQR control have been included in the paper to show the effectiveness of the proposed control approach. 

\clearpage
\newpage
\appendix

\section{}\label{App}

The nonzero entries of the system matrices $M(q)$, $C(q,\dot q)$, $F(\dot q)$ and $g(q)$ are define as follows:
\begin{align*}
m_{11} = I_{tot} + d^2m + A_1C_{\beta}^2 + A_2C_{\gamma}^2 + A_3C_{\beta}C_{\gamma} 
+ 2A_4dC_{\beta}S_{\theta_2} + 2A_5dC_{\gamma}S_{\theta_2} - d^2mC_{\theta_1}^2C_{\theta_2}^2, \\ 
m_{22} = A_1 + I_B,\quad  m_{33} = A_2 + I_J, m_{44} = m, \quad
m_{55} = dmC_{\theta_2}^2, \quad m_{66} = dm, \\ m_{12} = m_{21} = A_4dC_{\theta_2}S_{\beta}S_{\theta_1},\quad
m_{13} = m_{31} = A_5dC_{\theta_2}S_{\gamma}S_{\theta_1}, \\
m_{14} = m_{41} = C_{\theta_2}S_{\theta_1}(A_4C_{\beta} + A_5C_{\gamma}),\quad
m_{15} = m_{51} = dC_{\theta_1}C_{\theta_2}(A_4C_{\beta} + A_5C_{\gamma} + dmS_{\theta_2}),\\
m_{16} = m_{61} =-dS_{\theta_1}(dm + A_4C_{\beta}S_{\theta_2} + A_5C_{\gamma}S_{\theta_2}),\quad
m_{23} = m_{32} = {1\over{2}}(A_3C_{\beta-\gamma}), \\
m_{24} = m_{42} = -A_4(S_{\beta}S_{\theta_2} + C_{\beta}C_{\theta_1}C_{\theta_2})\\
m_{25} = m_{52} = A_4dC_{\beta}C_{\theta_2}S_{\theta_1},\quad
m_{26} = m_{62} = -A_4d(C_{\theta_2}S_{\beta} - C_{\beta}C_{\theta_1}S_{\theta_2}), \\ 
m_{34} = m_{43} = -A_5(S_{\gamma}S_{\theta_2} + C_{\gamma}C_{\theta_1}C_{\theta_2}) \\
m_{35} = m_{53} = A_5dC_{\gamma}C_{\theta_2}S_{\theta_1}\quad
m_{36} = m_{63} = -A_5d(C_{\theta_2}S_{\gamma} - C_{\gamma}C_{\theta_1}S_{\theta_2})\\
c_{11} = 2d\dot{d}m - A_2\dot{\gamma}S_{2\gamma} - A_1\dot{\beta}S_{2\beta} + 2A_4\dot{d}C_{\beta}S_{\theta_2} 
+ 2A_5\dot{d}C_{\gamma}S_{\theta_2} - A_3\dot{\beta}C_{\gamma}S_{\beta}  - A_3\dot{\gamma}C_{\beta}S_{\gamma} \\
- 2d\dot{d}mC_{\theta_1}^2C_{\theta_2}^2 + 2A_4d\dot{\theta}_2C_{\beta}C_{\theta_2} + 2A_5d\dot{\theta}_2C_{\gamma}C_{\theta_2} 
- 2A_4d\dot{\beta}S_{\beta}S_{\theta_2} \\ - 2A_5d\dot{\gamma}S_{\gamma}S_{\theta_2} + 2d^2\dot{\theta}_1mC_{\theta_1}C_{\theta_2}^2S_{\theta_1}
+ 2d^2\dot{\theta}_2mC_{\theta_1}^2C_{\theta_2}S_{\theta_2}\\
c_{12} = A_4\dot{d}C_{\theta_2}S_{\beta}S_{\theta_1} + A_d\dot{\beta}C_{\beta}C_{\theta_2}S_{\theta_1} +  A_4d\dot{\theta}_1C_{\theta_1}C_{\theta_2}S_{\beta} - A_4d\dot{\theta}_2S_{\beta}S_{\theta_1}S_{\theta_2} \\
c_{13} = A_5\dot{d}C_{\theta_2}S_{\gamma}S_{\theta_1} + A_5d\dot{\gamma}C_{\gamma}C_{\theta_2}S_{\theta_1} +  A_5d\dot{\theta}_1C_{\theta_1}C_{\theta_2}S_{\gamma} - A_5d\dot{\theta}_2S_{\gamma}S_{\theta_1}S_{\theta_2} \\
c_{14} = \dot{\theta}_1C_{\theta_1}C_{\theta_2}(A_4C_{\beta} + A_5C_{\gamma}) - C_{\theta_2}S_{\theta_1}(A_4\dot{\beta}S_{\beta} + A_5\dot{\gamma}S_{\gamma})  - \dot{\theta}_2S_{\theta_1}S_{\theta_2}(A_4C_{\beta} + A_5C_{\gamma}) \\
c_{15} = dC_{\theta_1}C_{\theta_2}(\dot{d}mS_{\theta_2} - A_4\dot{\beta}S_{\beta} - A_5d\dot{\gamma}S_{\gamma} +  d\dot{\theta}_2mC_{\theta_2})  + \dot{d}C_{\theta_1}C_{\theta_2}(A_4C_{\beta} + A_5C_{\gamma} + dmS_{\theta_2}) \\- d\dot{\theta}_1C_{\theta_2}S_{\theta_1}(A_4C_{\beta}  + A_5C_{\gamma}  + dmS_{\theta_2}) - d\dot{\theta}_2C_{\theta_1}S_{\theta_2}(A_4C_{\beta} + A_5C_{\gamma} + dmS_{\theta_2}) \\
c_{16} = - dS_{\theta_1}(\dot{d}m + A_4\dot{\theta}_2C_{\beta}C_{\theta_2} + A_5\dot{\theta}_2C_{\gamma}C_{\theta_2} -  A_4d\dot{\beta}S_{\beta}S_{\theta_2} - A_5\dot{\gamma}S_{\gamma}S_{\theta_2}) \\ - \dot{d}S_{\theta_1}(dm + A_4C_{\beta}S_{\theta_2} +  A_5C_{\gamma}S_{\theta_2}) - d\dot{\theta}_1C_{\theta_1}(dm + A_4C_{\beta}S_{\theta_2} + A_5C_{\gamma}S_{\theta_2}) \\
c_{21} = {1\over{2}}(\dot{\alpha}S_{\beta}(2A_1C_{\beta} + A_3C_{\gamma} + 2A_4dS_{\theta_2})) +  {1\over{2}}(3A_4\dot{d}C_{\theta_2}S_{\beta}S_{\theta_1}) +  {1\over{2}}(A_4d\dot{\beta}C_{\beta}C_{\theta_2}S_{\theta_1}) \\ +  {1\over{2}}(3A_4d\dot{\theta}_1C_{\theta_1}C_{\theta_2}S_{\beta}) - {1\over{2}}(3A_4d\dot{\theta}_2S_{\beta}S_{\theta_1}S_{\theta_2}) \\
c_{22} = {1\over{4}}(A_3\dot{\gamma}S_{\beta-\gamma}) + {1\over{2}}(A_4\dot{d}(C_{\beta}S_{\theta_2} - C_{\theta_1}C_{\theta_2}S_{\beta})) +  {1\over{2}}(A_4d\dot{\theta}_2(C_{\beta}C_{\theta_2} \\ + C_{\theta_1}S_{\beta}S_{\theta_2})) - {1\over{2}}(A_4d\dot{\alpha}C_{\beta}C_{\theta_2}S_{\beta}) +
{1\over{2}}(A_4d\dot{\theta}_1C_{\theta_2}S_{\beta}S_{\beta}) \\
c_{23} = -{1\over{4}}(A_3S_{\beta-\gamma})(\dot{\beta} - 2\dot{\gamma})) \\
c_{24} = {1\over{2}}(A_4\dot{\beta}C_{\theta_1}C_{\theta_2}S_{\beta}) - A_4\dot{\theta}_2C_{\theta_2}S_{\beta} - {1\over{2}}(A_4\dot{\beta}C_{\beta}S_{\theta_2}) + A_4\dot{\theta}_1C_{\beta}C_{\theta_2}S_{\beta} \\ + A_4\dot{\theta}_2C_{\beta}C_{\theta_1}S_{\theta_2} +  {1\over{2}}(A_4\dot{\alpha}C_{\theta_2}S_{\beta}S_{\beta}) \\
\end{align*}
\begin{align*}
c_{25} =A_4\dot{d}C_{\beta}C_{\theta_2}S_{\beta} + {1\over{2}}(A_4d\dot{\alpha}C_{\theta_1}C_{\theta_2}S_{\beta}) + 
A_4d\dot{\theta}_1C_{\beta}C_{\theta_1}C_{\theta_2}  -
{1\over{2}}(A_4d\dot{\beta}C_{\theta_2}S_{\beta}S_{\beta}) - A_4d\dot{\theta}_2C_{\beta}S_{\beta}S_{\theta_2} \\
c_{26} = A_4d\dot{\theta}_2S_{\beta}S_{\theta_2} - {1\over{2}}(A_4d\dot{\beta}C_{\beta}C_{\theta_2}) - A_4\dot{d}C_{\theta_2}S_{\beta} +  A_4\dot{d}C_{\beta}C_{\theta_1}S_{\theta_2} -  {1\over{2}}(A_4d\dot{\beta}C_{\theta_1}S_{\beta}S_{\theta_2}) \\- A_4d\dot{\theta}_1C_{\beta}S_{\beta}S_{\theta_2} - {1\over{2}}(A_4d\dot{\alpha}S_{\beta}S_{\beta}S_{\theta_2}) + A_4d\dot{\theta}_2C_{\beta}C_{\theta_1}C_{\theta_2} \\
c_{31} = {1\over{2}}(\dot{\alpha}S_{\gamma}(2A_2C_{\gamma} + A_3C_{\beta} + 2A_5dS_{\theta_2}))  + {1\over{2}}(3A_5\dot{d}C_{\theta_2}S_{\gamma}S_{\theta_1}) +  {1\over{2}}(A_5d\dot{\gamma}C_{\gamma}C_{\theta_2}S_{\theta_1})  \\ + {1\over{2}}(3A_5d\dot{\theta}_1C_{\theta_1}C_{\theta_2}S_{\gamma})  -{1\over{2}}(3A_5d\dot{\theta}_2S_{\gamma}S_{\theta_1}S_{\theta_2})\\
c_{32} = -{1\over{4}}(A_3S_{\beta-\gamma})(2\dot{\beta} - \dot{\gamma}))\\
c_{33} = {1\over{2}}(A_5\dot{d}(C_{\gamma}S_{\theta_2} - C_{\theta_1}C_{\theta_2}S_{\gamma}))  -{1\over{4}} (A_3\dot{\beta}S_{\beta-\gamma})  + {1\over{2}}(A_5d\dot{\theta}_2(C_{\gamma}C_{\theta_2} + C_{\theta_1}S_{\gamma}S_{\theta_2})) \\ - {1\over{2}}(A_5d\dot{\alpha}C_{\gamma}C_{\theta_2}S_{\theta_1}) + {1\over{2}}(A_5d\dot{\theta}_1C_{\theta_2}S_{\gamma}S_{\theta_1})
\\
c_{34} = {1\over{2}}(A_5\dot{\gamma}C_{\theta_1}C_{\theta_2}S_{\gamma}) - A_5\dot{\theta}_2C_{\theta_2}S_{\gamma} - {1\over{2}}(A_5\dot{\gamma}C_{\gamma}S_{\theta_2}) + A_5\dot{\theta}_1C_{\gamma}C_{\theta_2}S_{\theta_1} \\+ A_5\dot{\theta}_2C_{\gamma}C_{\theta_1}S_{\theta_2} + {1\over{2}}(A_5\dot{\alpha}C_{\theta_2}S_{\gamma}S_{\theta_1})\\
c_{35} = A_5ddC_{\gamma}C_{\theta_2}S_{\theta_1} + {1\over{2}}(A_5d\dot{\alpha}C_{\theta_1}C_{\theta_2}S_{\gamma})  \\- {1\over{2}}(A_5d\dot{\gamma}C_{\theta_2}S_{\gamma}S_{\theta_1}) - A_5d\dot{\theta}_2C_{\gamma}S_{\theta_1}S_{\theta_2} + A_5d\dot{\theta}_1C_{\gamma}C_{\theta_1}C_{\theta_2}
\\
c_{36} = A_5d\dot{\theta}_2S_{\gamma}S_{\theta_2} - {1\over{2}}(A_5d\dot{\gamma}C_{\gamma}C_{\theta_2}) - A_5\dot{d}C_{\theta_2}S_{\gamma} + A_5\dot{d}C_{\gamma}C_{\theta_1}S_{\theta_2} -{1\over{2}} (A_5d\dot{\gamma}C_{\theta_1}S_{\gamma}S_{\theta_2}) \\ - A_5d\dot{\theta}_1C_{\gamma}S_{\theta_1}S_{\theta_2} -{1\over{2}} (A_5d\dot{\alpha}S_{\gamma}S_{\theta_1}S_{\theta_2}) + A_5d\dot{\theta}_2C_{\gamma}C_{\theta_1}C_{\theta_2}\\
c_{41} = d\dot{\theta}_2mS_{\theta_1} - d\dot{\alpha}m - A_4\dot{\alpha}C_{\beta}S_{\theta_2} - A_5\dot{\alpha}C_{\gamma}S_{\theta_2} + d\dot{\alpha}mC_{\theta_1}^2C_{\theta_2}^2 + {1\over{2}}(A_4\dot{\theta}_1C_{\beta}C_{\theta_1}C_{\theta_2}) \\ + {1\over{2}}(A_5\dot{\theta}_1C_{\gamma}C_{\theta_1}C_{\theta_2}) -{1\over{2}}(3A_4\dot{\beta}C_{\theta_2}S_{\beta}S_{\theta_1}) - {1\over{2}}(A_4\dot{\theta}_2C_{\beta}S_{\theta_1}S_{\theta_2})  - {1\over{2}}(3A_5\dot{\gamma}C_{\theta_2}S_{\gamma}S_{\theta_1})\\ -{1\over{2}}(A_5\dot{\theta}_2C_{\gamma}S_{\theta_1}S_{\theta_2})-d\dot{\theta}_1mC_{\theta_1}C_{\theta_2}S_{\theta_2}\\
c_{42} = A_4\dot{\beta}C_{\theta_1}C_{\theta_2}S_{\beta} - (A_4\dot{\theta}_2C_{\theta_2}S_{\beta})/2 - A_4\dot{\beta}C_{\beta}S_{\theta_2} + {1\over{2}}(A_4\dot{\theta}_1C_{\beta}C_{\theta_2}S_{\theta_1}) \\ + {1\over{2}}(A_4\dot{\theta}_2C_{\beta}C_{\theta_1}S_{\theta_2}) - {1\over{2}}(A_4\dot{\alpha}C_{\theta_2}S_{\beta}S_{\theta_1}) \\
c_{43} = A_5\dot{\gamma}C_{\theta_1}C_{\theta_2}S_{\gamma} - {1\over{2}}(A_5\dot{\theta}_2C_{\theta_2}S_{\gamma})  - A_5\dot{\gamma}C_{\gamma}S_{\theta_2}  + {1\over{2}}(A_5\dot{\theta}_1C_{\gamma}C_{\theta_2}S_{\theta_1})  \\ + {1\over{2}}(A_5\dot{\theta}_2C_{\gamma}C_{\theta_1}S_{\theta_2}) -
{1\over{2}}(A_5\dot{\alpha}C_{\theta_2}S_{\gamma}S_{\theta_1}) \\
c_{45} = d\dot{\theta}_1m(S_{\theta_2}^2 - 1) - {1\over{2}}(\dot{\alpha}C_{\theta_1}C_{\theta_2}(A_4C_{\beta}  + A_5C_{\gamma} + 2dmS_{\theta_2})) \\ - {1\over{2}}(A_4\dot{\beta}C_{\beta}C_{\theta_2}S_{\theta_1})  - {1\over{2}}(A_5\dot{\gamma}C_{\gamma}C_{\theta_2}S_{\theta_1})\\
\end{align*}
\begin{align*}
c_{46} = {1\over{2}}(A_4\dot{\beta}(C_{\theta_2}S_{\beta} - C_{\beta}C_{\theta_1}S_{\theta_2})) + {1\over{2}}(A_5\dot{\gamma}(C_{\theta_2}S_{\gamma} - C_{\gamma}C_{\theta_1}S_{\theta_2})) \\ - d\dot{\theta}_2m + {1\over{2}}(\dot{\alpha}S_{\theta_1}(2dm + A_4C_{\beta}S_{\theta_2} + A_5C_{\gamma}S_{\theta_2}))  \\
c_{51} = 2d^2\dot{\theta}_2mC_{\theta_1}C_{\theta_2}^2 - {1\over{2}}(d^2\dot{\theta}_2mC_{\theta_1}) +  {1\over{2}}(A_4\dot{d}C_{\beta}C_{\theta_1}C_{\theta_2})  \\ + {1\over{2}}(A_5\dot{d}C_{\gamma}C_{\theta_1}C_{\theta_2}) - {1\over{2}}(3A_4d\dot{\beta}C_{\theta_1}C_{\theta_2}S_{\beta})  - {1\over{2}}(A_4d\dot{\theta}_1C_{\beta}C_{\theta_2}S_{\theta_1}) \\ - {1\over{2}}(A_4d\dot{\theta}_2C_{\beta}C_{\theta_1}S_{\theta_2}) - {1\over{2}}(3A_5d\dot{\gamma}C_{\theta_1}C_{\theta_2}S_{\gamma}) - {1\over{2}}(A_5d\dot{\theta}_1C_{\gamma}C_{\theta_2}S_{\theta_1}) \\ - {1\over{2}}(A_5d\dot{\theta}_2C_{\gamma}C_{\theta_1}S_{\theta_2}) + 2d\dot{d}mC_{\theta_1}C_{\theta_2}S_{\theta_2} \\ - {1\over{2}}(d^2\dot{\theta}_1mC_{\theta_2}S_{\theta_1}S_{\theta_2}) - d^2\dot{\alpha}mC_{\theta_1}C_{\theta_2}^2S_{\theta_1} \\
c_{52} = {1\over{2}}(A_4\dot{d}C_{\beta}C_{\theta_2}S_{\theta_1}) - {1\over{2}}(A_4d\dot{\alpha}C_{\theta_1}C_{\theta_2}S_{\beta}) \\- A_4d\dot{\beta}C_{\theta_2}S_{\beta}S_{\theta_1} - {1\over{2}}(A_4d\dot{\theta}_2C_{\beta}S_{\theta_1}S_{\theta_2}) + {1\over{2}}(A_4d\dot{\theta}_1C_{\beta}C_{\theta_1}C_{\theta_2}) \\
c_{53} = {1\over{2}}(A_5\dot{d}C_{\gamma}C_{\theta_2}S_{\theta_1}) - {1\over{2}}(A_5d\dot{\alpha}C_{\theta_1}C_{\theta_2}S_{\gamma}) \\ - A_5d\dot{\gamma}C_{\theta_2}S_{\gamma}S_{\theta_1}  - {1\over{2}}(A_5d\dot{\theta}_2C_{\gamma}S_{\theta_1}S_{\theta_2}) + {1\over{2}}(A_5d\dot{\theta}_1C_{\gamma}C_{\theta_1}C_{\theta_2}) \\
c_{54} = - {1\over{2}}(\dot{\alpha}C_{\theta_1}C_{\theta_2}(A_4C_{\beta} + A_5C_{\gamma})) - {1\over{2}}(A_4\dot{\beta}C_{\beta}C_{\theta_2}S_{\theta_1}) - {1\over{2}}(A_5\dot{\gamma}C_{\gamma}C_{\theta_2}S_{\theta_1})\\
c_{55} = {1\over{2}}(d\dot{\alpha}C_{\theta_2}S_{\theta_1}(A_4C_{\beta} + A_5C_{\gamma} + dmS_{\theta_2})) - {1\over{2}}(A_4d\dot{\beta}C_{\beta}C_{\theta_1}C_{\theta_2}) - {1\over{2}}(A_5d\dot{\gamma}C_{\gamma}C_{\theta_1}C_{\theta_2}) \\
c_{56} = {1\over{2}}(d\dot{\alpha}C_{\theta_1}(dm + A_4C_{\beta}S_{\theta_2} + A_5C_{\gamma}S_{\theta_2})) + {1\over{2}}(A_4d\dot{\beta}C_{\beta}S_{\theta_1}S_{\theta_2}) + {1\over{2}}(A_5d\dot{\gamma}C_{\gamma}S_{\theta_1}S_{\theta_2}) \\
c_{61} =  {1\over{2}}(3A_4d\dot{\beta}S_{\beta}S_{\theta_1}S_{\theta_2}) - {1\over{2}}(d^2\dot{\theta}_1mC_{\theta_1})  - d^2\dot{\theta}_1mC_{\theta_1}C_{\theta_2}^2  - A_4d\dot{\alpha}C_{\beta}C_{\theta_2} - A_5d\dot{\alpha}C_{\gamma}C_{\theta_2} \\- {1\over{2}}(A_4d\dot{d}C_{\beta}S_{\theta_1}S_{\theta_2})  - {1\over{2}}(A_5d\dot{d}C_{\gamma}S_{\theta_1}S_{\theta_2})  - {1\over{2}}(A_4d\dot{\theta}_1C_{\beta}C_{\theta_1}S_{\theta_2})  - {1\over{2}}(A_4d\dot{\theta}_2C_{\beta}C_{\theta_2}S_{\theta_1}) \\ - {1\over{2}}(A_5d\dot{\theta}_1C_{\gamma}C_{\theta_1}S_{\theta_2})  - {1\over{2}}(A_5d\dot{\theta}_2C_{\gamma}C_{\theta_2}S_{\theta_1}) - 2d\dot{d}mS_{\theta_1} + {1\over{2}}(3A_5d\dot{\gamma}S_{\gamma}S_{\theta_1}S_{\theta_2}) - d^2\dot{\alpha}mC_{\theta_1}^2C_{\theta_2}S_{\theta_2} \\
c_{62} = {1\over{2}}(A_4d\dot{\theta}_2S_{\beta}S_{\theta_2}) - A_4d\dot{\beta}C_{\beta}C_{\theta_2} -  {1\over{2}}(A_4d\dot{d}C_{\theta_2}S_{\beta}) + {1\over{2}}(A_4d\dot{d}C_{\beta}C_{\theta_1}S_{\theta_2})\\ - A_4d\dot{\beta}C_{\theta_1}S_{\beta}S_{\theta_2} - {1\over{2}}(A_4d\dot{\theta}_1C_{\beta}S_{\theta_1}S_{\theta_2})  + {1\over{2}}(A_4d\dot{\alpha}S_{\beta}S_{\theta_1}S_{\theta_2}) + {1\over{2}}(A_4d\dot{\theta}_2C_{\beta}C_{\theta_1}C_{\theta_2})\\
c_{63} = {1\over{2}}(A_5d\dot{\theta}_2S_{\gamma}S_{\theta_2}) - A_5d\dot{\gamma}C_{\gamma}C_{\theta_2} - {1\over{2}}(A_5d\dot{d}C_{\theta_2}S_{\gamma})  +{1\over{2}}(A_5d\dot{d}C_{\gamma}C_{\theta_1}S_{\theta_2}) \\ - A_5d\dot{\gamma}C_{\theta_1}S_{\gamma}S_{\theta_2} - {1\over{2}}(A_5d\dot{\theta}_1C_{\gamma}S_{\theta_1}S_{\theta_2})  + {1\over{2}}(A_5d\dot{\alpha}S_{\gamma}S_{\theta_1}S_{\theta_2}) + {1\over{2}}(A_5d\dot{\theta}_2C_{\gamma}C_{\theta_1}C_{\theta_2})\\
c_{64} = {1\over{2}}(A_4\dot{\beta}(C_{\theta_2}S_{\beta} - C_{\beta}C_{\theta_1}S_{\theta_2})) + {1\over{2}}(A_5\dot{\gamma}(C_{\theta_2}S_{\gamma}  - C_{\gamma}C_{\theta_1}S_{\theta_2})) + {1\over{2}}(\dot{\alpha}S_{\theta_1}S_{\theta_2}(A_4C_{\beta} + A_5C_{\gamma})) \\
c_{65} = {1\over{2}}(d^2\dot{\alpha}mC_{\theta_1}) - d^2\dot{\alpha}mC_{\theta_1}C_{\theta_2}^2  +  {1\over{2}}(A_4d\dot{\alpha}C_{\beta}C_{\theta_1}S_{\theta_2}) + {1\over{2}}(A_5d\dot{\alpha}C_{\gamma}C_{\theta_1}S_{\theta_2})  \\ + {1\over{2}}(A_4d\dot{\beta}C_{\beta}S_{\theta_1}S_{\theta_2}) + {1\over{2}}(A_5d\dot{\gamma}C_{\gamma}S_{\theta_1}S_{\theta_2}) \\
c_{66} = {1\over{2}}(d\dot{\alpha}C_{\theta_2}S_{\theta_1}(A_4C_{\beta} + A_5C_{\gamma})) -  {1\over{2}}(A_5d\dot{\gamma}(S_{\gamma}S_{\theta_2} + C_{\gamma}C_{\theta_1}C_{\theta_2})) - {1\over{2}}(A_4d\dot{\beta}(S_{\beta}S_{\theta_2}  + C_{\beta}C_{\theta_1}C_{\theta_2}))\\
\end{align*}
\begin{align*}
g_2 = {1\over{2}}gl_Bcos\beta(2m+m_B+2m_J),\quad 
g_3 = {1\over{2}}gl_Jcos\gamma(2m+m_J), \\
g_4 = -gmcos\theta_1cos\theta_2,\quad
g_5 = gmdcos\theta_2sin\theta_1,\quad 
g_6 = gmdcos\theta_1sin\theta_2
\end{align*}

\bibliographystyle{apacite}
\bibliography{biblio}

\end{document}